\documentclass{ejm}

\usepackage{amssymb}
\usepackage[leqno,intlimits]{amsmath}
\usepackage{psfrag}
\usepackage{setspace}
\usepackage{epsfig}
\usepackage{subfig}
\usepackage{graphicx}
\def\ni{\noindent}
\def\proof{{\ni \bf \underline{Proof:} }}

\newcommand{\bsub}{\begin{subequations}}
\newcommand{\esub}{\end{subequations}$\!$}

\newcommand{\e}{{\rm e}}

\newcommand{\R}{{\mathbb{R}}}
\newcommand{\eps}{{\displaystyle \varepsilon}}
\newcommand{\lam}{{\lambda}}

\newcommand{\p}{\prime}

\newcommand{\xb}{{\bf{x}}}

\newcommand{\xbo}{\mbox{$0$}}

\newcommand{\lb}{{\pmb l}}
\newcommand{\pbi}{{\pmb P}_i}
\newcommand{\pb}{{\pmb P}}
\newcommand{\etai}{{\pmb \eta}_i}
\newcommand{\etaip}{{\pmb \eta}_{i}^{\perp}}
\newcommand{\kb}{{\pmb k}}
\newcommand{\kapb}{{\pmb \kappa}}

\newcommand{\yb}{{\bf{y}}}

\newcommand{\zero}{\mbox{$0$}}

\newtheorem{example}{Example}[section]

\newcommand{\bex}{\begin{example}\rm}
\newcommand{\eex}{\end{example}}

\def\R{{\mathbb{R}}}
\def\Z{{\mathbb{Z}}}

\def\C{{\mathbb{C}}}

\def\im{{\rm Im}}
\def\re{{\rm Re}}

\def\<{\langle}
\def\>{\rangle}
\def\l{\pmb{l}}

\numberwithin{equation}{section}

\renewcommand{\theequation}{\arabic{section}.\arabic{equation}}
\renewcommand{\thesection}{\arabic{section}}

\newtheorem{lemma}{Lemma}[section]

\newtheorem{conjecture}{Conjecture}[section]
\newtheorem{result}{Principal Result}[section]
\newtheorem{remark}{Remark}[section]

\title[The Stability of Periodic Patterns of Localized Spots for
  Reaction-Diffusion Systems]
 {Logarithmic Expansions and the Stability of Periodic Patterns of
   Localized Spots for Reaction-Diffusion Systems in $\R^2$}

\author[D. Iron, J. Rumsey, M. J. Ward, J. Wei ]{%
  D.\ns I\ls R\ls O\ls N, \ns J. \ns R\ls U\ls M\ls S\ls E\ls Y, \ns
   M.\ns J. \ns W\ls A\ls R\ls D, \ns
  \and  J.\ns W\ls E\ls I
}
\affiliation{David Iron; Department of Mathematics,
 Dalhousie University, Halifax, Nova Scotia, B3H 3J5, Canada,

 John Rumsey; Faculty of Management,
 Dalhousie University, Halifax, Nova Scotia, B3H 3J5, Canada,

 Michael Ward; Department of Mathematics,  University of British Columbia,
 Vancouver, British Columbia, V6T 1Z2, Canada,

Juncheng Wei, Department of Mathematics,  University of British Columbia,
 Vancouver, British Columbia, V6T 1Z2, Canada and Department of Mathematics, Chinese University of Hong Kong, Shatin, New Territories, Hong Kong.}

\date{\today}
\pubyear{2004}
\volume{000}
\pagerange{\pageref{firstpage}--\pageref{lastpage}}

\begin{document}

\label{firstpage}
\maketitle

\baselineskip=12pt

\begin{abstract}

The linear stability of steady-state periodic patterns of localized
spots in $\R^2$ for the two-component Gierer-Meinhardt (GM) and
Schnakenburg reaction-diffusion models is analyzed in the semi-strong
interaction limit corresponding to an asymptotically small diffusion
coefficient $\eps^2$ of the activator concentration. In the limit
$\eps\to 0$, localized spots in the activator are centered at the
lattice points of a Bravais lattice with constant area $|\Omega|$.  To
leading order in $\nu={-1/\log\eps}$, the linearization of the
steady-state periodic spot pattern has a zero eigenvalue when the
inhibitor diffusivity satisfies $D={D_0/\nu}$, for some $D_0$
independent of the lattice and the Bloch wavevector $\kb$.  From a
combination of the method of matched asymptotic expansions,
Floquet-Bloch theory, and the rigorous study of certain nonlocal
eigenvalue problems, an explicit analytical formula for the continuous
band of spectrum that lies within an ${\mathcal O}(\nu)$ neighborhood
of the origin in the spectral plane is derived when $D={D_0/\nu} +
D_1$, where $D_1={\mathcal O}(1)$ is a de-tuning parameter. The
periodic pattern is linearly stable when $D_1$ is chosen small enough
so that this continuous band is in the stable left-half plane
$\mbox{Re}(\lambda)<0$ for all $\kb$. Moreover, for both the
Schnakenburg and GM models, our analysis identifies a model-dependent
objective function, involving the regular part of the Bloch Green's
function, that must be maximized in order to determine the specific
periodic arrangement of localized spots that constitutes a linearly
stable steady-state pattern for the largest value of $D$. From a
numerical computation, based on an Ewald-type algorithm, of the
regular part of the Bloch Green's function that defines the objective
function, it is shown within the class of oblique Bravais lattices
that a regular hexagonal lattice arrangement of spots is optimal for
maximizing the stability threshold in $D$.

\end{abstract}

\noindent Key words: singular perturbations, localized spots,
logarithmic expansions, Bravais lattice, Floquet-Bloch theory, Green's
function, nonlocal eigenvalue problem.

\baselineskip=16pt

\setcounter{equation}{0}
\setcounter{section}{0}
\section{Introduction} \label{section:1}

Spatially localized spot patterns occur for various classes of
reaction-diffusion (RD) systems with diverse applications to
theoretical chemistry, biological morphogenesis, and applied physics.
A survey of experimental and theoretical studies, through RD modeling,
of localized spot patterns in various physical or chemical contexts is
given in \cite{vanag}. Owing to the widespread occurrence of localized
patterns in various scientific applications, there has been
considerable focus over the past decade on developing a theoretical
understanding of the dynamics and stability of localized solutions to
singularly perturbed RD systems.  A brief survey of some open
directions for the theoretical study of localized patterns in various
applications is given in \cite{kn_survey}. More generally, a wide
range of topics in the analysis of far-from-equilibrium patterns
modeled by PDE systems are discussed in \cite{book_Nishiura}.

In this broad context, the goal of this paper is to analyze the linear
stability of steady-state periodic patterns of localized spots in
$\R^2$ for two-component RD systems in the semi-strong interaction
regime characterized by an asymptotically large diffusivity ratio.
For concreteness, we will focus our analysis on two specific
models. One model is a simplified Schnakenburg-type system
\begin{equation}
v_{t}    =\eps^{2} \Delta v-v+ uv^{2} \,, \qquad \tau u_t =
 D \Delta u+a-\eps^{-2} uv^{2} \,, \label{1:sc}
\end{equation}
where $0<\eps\ll 1$, $D>0$, $\tau>0$, and $a>0$, are parameters.  The
second model is the prototypical Gierer-Meinhardt (GM) model
formulated as
\begin{equation}
v_{t}    =\eps^{2} \Delta v-v+ {v^{2}/u}\,, \qquad
\tau u_t =  D \Delta u - u +\eps^{-2} v^{2} \,,
 \label{1:gm}
\end{equation}
where $0<\eps\ll 1$, $D>0$, and $\tau>0$, are parameters.

Our linear stability analysis for these two models will focus on the
semi-strong interaction regime $\eps\to 0$ with $D={\mathcal O}(1)$.
For $\eps \to 0$, the localized spots for $v$ are taken to be centered
at the lattice points of a general Bravais lattice $\Lambda$, where
the area $|\Omega|$ of the primitive cell is held constant.  A brief
outline of lattices and reciprocal lattices is given in \S
\ref{2:lattice}.  Our main goal for the Schnakenburg and GM models is
to formulate an explicit objective function to be maximized that will
identify the specific lattice arrangement of localized spots that is a
linearly stable steady-state pattern for the largest value of
$D$. Through a numerical computation of this objective function we
will show that it is a regular hexagonal lattice arrangement of spots that
yields this optimal stability threshold.

For the corresponding problem in 1-D, the stability of periodic
patterns of spikes for the GM model was analyzed in \cite{vpd} by
using the geometric theory of singular perturbations combined with
Evans-function techniques.  On a bounded 1-D domain with homogeneous
Neumann boundary conditions, the stability of $N$-spike steady-state
solutions was analyzed in \cite{iww} and \cite{ww} through a detailed
study of certain nonlocal eigenvalue problems. On a bounded $2-D$
domain with Neumann boundary conditions, a leading order in
$\nu={-1/\log\eps}$ rigorous theory was developed to analyze the
stability of multi-spot steady-state patterns for the GM model (cf.  \cite{w},
\cite{WGM1}), the Schnakenburg model (cf.~\cite{survey_Wei:2008}), and
the Gray-Scott (GS) model (cf.~\cite{WGS2}), in the parameter regime where
$D={D_0/\nu}\gg 1$.  For the Schnakenburg and GM models, the
leading-order stability threshold for $D_0$ corresponding to a zero
eigenvalue crossing was determined explicitly. A hybrid
asymptotic-numerical theory to study the stability, dynamics, and
self-replication patterns of spots, that is accurate to all
powers in $\nu$, was developed for the Schnakenburg model in
\cite{KWW_schnak} and for the GS model in \cite{cw_1}. In \cite{MO1}
and \cite{MO3}, the stability and self-replication behavior of a one-spot
solution for the GS model was analyzed.

One of the key features of the finite domain problem in comparison
with the periodic problem is that the spectrum of the linearization of
the former is discrete rather than continuous. As far as we are aware,
to date there has been no analytical study of the stability of
periodic patterns of localized spots in $\R^2$ on Bravais lattices for
singularly perturbed two-component RD systems. In the weakly nonlinear
Turing regime, an analysis of the stability of patterns on Bravais lattices
in $\R^3$ using group-theoretic tools of bifurcation theory with symmetry
was done in \cite{call_1} and \cite{call_2}.

By using the method of matched asymptotic expansions, in the limit
$\eps\to 0$ a steady-state localized spot solution is constructed for
(\ref{1:sc}) and for (\ref{1:gm}) within the fundamental Wigner-Seitz
cell of the lattice. The solution is then extended periodically to all
of $\R^2$. The stability of this solution with respect to ${\mathcal
  O}(1)$ time-scale instabilities arising from zero eigenvalue
crossings is then investigated by first using the Floquet-Bloch
theorem (cf.~\cite{KR}, \cite{K}) to formulate a singularly perturbed
eigenvalue problem in the Wigner-Seitz cell $\Omega$ with
quasi-periodic boundary conditions on $\partial\Omega$ involving the
Bloch vector $\kb$. In \S~\ref{2:gr_lattice}, the Floquet-Bloch theory
is formulated and a few key properties of the Bloch Green's function
for the Laplacian are proved. In \S~\ref{schnak} and \S~\ref{gm}, the
spectrum of the linearized eigenvalue problem is analyzed by using the
method of matched asymptotic expansions combined with a spectral
analysis based on perturbations of a nonlocal eigenvalue problem. More
specifically, to leading-order in $\nu={-1/\log\eps}$ it is shown that
a zero eigenvalue crossing occurs when $D\sim{D_0/\nu}$, where $D_0$
is a constant that depends on the parameters in the RD system, but is
independent of the lattice geometry except through the area $|\Omega|$
of the Wigner-Seitz cell. Therefore, to leading-order in $\nu$, the
stability threshold is the same for any periodic spot pattern on a
Bravais lattice $\Lambda$ when $|\Omega|$ is held fixed. In order to
determine the effect of the lattice geometry on the stability
threshold, an expansion to higher-order in $\nu$ must be
undertaken. In related singularly perturbed eigenvalue problems for
the Laplacian in 2-D domains with holes, the leading-order eigenvalue
asymptotics in the limit of small hole radius only depends on the
number of holes and the area of the domain, and not on the arrangement
of the holes within the domain. An analytical theory to calculate
higher order terms in the eigenvalue asymptotics for these problems,
which have applications to narrow-escape and capture phenomena in
mathematical biology, is given in \cite{WHK}, \cite{KTW}, and
\cite{PWPK}.

To determine a higher-order approximation for the stability threshold
for the periodic spot problem we perform a more refined perturbation
analysis in order to calculate the continuous band $\lambda\sim
\nu\lambda_1(\kb,D_1,\Lambda)$ of spectra that lies within an
${\mathcal O}(\nu)$ neighborhood of the origin, i.e that satisfies
$|\lambda(\kb,D_1,\Lambda)|\leq {\mathcal O}(\nu)$, when
$D={D_{0}/\nu} + D_1$ for some de-tuning parameter
$D_1={\mathcal O}(1)$. This band is found to depend on the lattice
geometry $\Lambda$ through the regular part of certain Green's
functions. For the Schnakenburg model, $\lambda_1$ depends on the
regular part $R_{b0}(\kb)$of the Bloch Green's function for the
Laplacian, which depends on both $\kb$ and the lattice. For the GM
Model, $\lambda_1$ depends on both $R_{b0}(\kb)$ and the regular part
$R_{0p}$ of the periodic source-neutral Green's function on $\Omega$. For both
models, this band of continuous spectrum that lies near the origin
when $D-{D_0/\nu}={\mathcal O}(1)$ is proved to be real-valued.

For both the Schnakenburg and GM models, the de-tuning parameter $D_1$
on a given lattice is chosen so that $\lambda_1<0$ for all
$\kb$. Then, to determine the lattice for which the steady-state spot
pattern is linearly stable for the largest possible value of $D$, we
simply maximize $D_1$ with respect to the lattice geometry.  In this
way, for each of the two RD models, we derive a model-dependent
objective function in terms of the regular parts of certain Green's
functions that must be maximized in order to determine the specific
periodic arrangement of localized spots that is linearly stable for
the largest value of $D$. The calculation of the continuous band of
spectra near the origin, and the derivation of the objective function
to be maximized so as to identify the optimal lattice, is done for the
Schnakenburg and GM models in \S~\ref{schnak} and \S~\ref{gm},
respectively.

In \S~\ref{simp:schnak} and \S~\ref{simp:gm} we exhibit a very simple
alternative method to readily identify this objective function for the
Schnakenburg and GM models, respectively. In \S~\ref{simp:gs}, this
simple alternative method is then used to determine an optimal lattice
arrangement of spots for the GS RD model.

In \S~\ref{sec:ewald} we show how to numerically compute the regular
part $R_{b0}(\kb)$ of the Bloch Green's function for the Laplacian
that arises in the objective function characterizing the optimum
lattice. Similar Green's functions, but for the Helmholtz operator,
arise in the linearized theory of the scattering of water waves by a
periodic arrangement of obstacles, and in related wave phenomena in
electromagnetics and photonics.  The numerical computation of Bloch
Green's functions is well-known to be a challenging problem owing to
the very slow convergence of their infinite series representations in
the spatial domain, and methodologies to improve the convergence
properties based on the Poisson summation formula are surveyed in
\cite{Linton} and \cite{Moroz}.  The numerical approach we use to
compute $R_{b0}(\kb)$ is an Ewald summation method, based on the
Poisson summation formula involving the direct and reciprocal
lattices, and follows closely the methodology developed in
\cite{Beyl-1} and \cite{Beyl-2}. Our numerical results show that
within the class of oblique Bravais lattices having a common area
$|\Omega|$ of the primitive cell, it is a regular hexagonal lattice that
optimizes the stability threshold for the Schnakneburg, GM, and GS
models.

Finally, we remark that optimal lattice arrangements of localized
structures in other PDE models having a variational structure, such as
the study of vortices in Ginzburg-Landau theory (cf.~\cite{sandier}),
the analysis of Abrikosov vortex lattices in the magnetic Ginzburg-Landau
system (cf.~\cite{sigal1, sigal2}) and the study of droplets in
diblock copolymer theory (cf.~\cite{chen}), have been identified
through the minimization of certain energy functionals. In contrast,
for our RD systems having no variational structure, the optimal
lattice is identified not through an energy minimization criterion,
but instead from a detailed analysis that determines the spectrum of
the linearization near the origin in the spectral plane when $D$ is
near a critical value.

\setcounter{equation}{0}
\setcounter{section}{1}
\section{Lattices and the Bloch Green's Functions}\label{2:latt_gr}

In this section we recall some basic facts about lattices and we
introduce the Bloch-periodic Green's functions that plays a central
role in the analysis in \S \ref{schnak}--\ref{simp}. A few key lemmas
regarding this Green's function are established.

\subsection{A Primer on Lattices and Reciprocal Lattices}\label{2:lattice}
Let $\lb_1$ and $\lb_2$ be two linearly independent vectors in $\R^2$,
with angle $\theta$ between them, where without loss of generality we
take $\lb_1$ to be aligned with the positive $x$-axis. The Bravais lattice
$\Lambda$ is defined by
\begin{equation}
  \Lambda = \Big{\lbrace}{ m \lb_1 + n \lb_2 \, \Big{\vert} \,\,
    m,\, n \in \mathbb{Z}\Big{\rbrace}} \,, \label{lattice-def}
\end{equation}
where $\mathbb{Z}$ denotes the set of integers. The {\em primitive}
cell is the parallelogram generated by the vectors $\lb_1$
and $\lb_2$ of area $|\lb_1 \times \lb_2|$. We will set the area of
the primitive cell to unity, so that $|\lb_1||\lb_2|\sin\theta=1$.

We can also write ${\pmb l}_1,{\pmb l}_2\in\R^2$ as complex numbers
$\alpha,\beta\in\C$. Without loss of generality we set $\im(\beta)>0$,
$\im(\alpha)=0$, and $\re(\alpha)>0$. In terms of $\alpha$ and
$\beta$, the area of the primitive cell is
$\im(\overline\alpha\,\beta)$, which we set to unity. For a regular hexagonal
lattice, $|\alpha|=|\beta|$, with $\beta=\alpha\,e^{i\theta}$,
$\theta={\pi/3}$, and $\alpha>0$. This yields $\im(\beta)=\alpha
{\sqrt{3}/2}$ and the unit area requirement gives ${\alpha^2\sqrt3/2}
= 1$, which yields $\alpha=\left({4/3}\right)^{1/4}$.  For the square
lattice, we have $\alpha=1$, $\beta=i$, and $\theta={\pi/2}$.

In terms of ${\pmb l}_1,{\pmb l}_2\in\R^2$, we have that ${\pmb l}_1 =
\bigl(\re(\alpha),\im(\alpha)\bigr)$, ${\pmb l}_2=
\bigl(\re(\beta),\im(\beta)\bigr)$ generate the lattice
(\ref{lattice-def}). For a regular hexagonal lattice of unit area
for the primitive cell we have
 \begin{equation}\label{hex-alp-bet}
 {\pmb l}_1 = \left(\left(\frac 43\right)^{1/4},0\right)
 \qquad\text{ and }\qquad
  {\pmb l}_2 = \left(\frac 43\right)^{1/4}\left(\frac 12,
\frac{\sqrt 3}{2}\right).
 \end{equation}
In Fig.~\ref{hex-cells} we plot a portion of the hexagonal lattice
generated with this ${\pmb l}_1,{\pmb l}_2$ pair.

\begin{figure} [htb]
\begin{center}
\psfrag{p}{$\scriptstyle{\bf 0}$}
\psfrag{alpha}{$\scriptstyle{{\pmb l}_1}$}
\psfrag{beta}{$\scriptstyle{{\pmb l}_2}$}
\psfrag{p+b}{$\scriptstyle{{\pmb l}_2}$}
\psfrag{p-b}{$\scriptstyle{-{\pmb l}_2}$}
\psfrag{p+2a+b}{$\scriptstyle{2{\pmb l}_1+{\pmb l}_2}$}
\psfrag{p-2a+b}{$\scriptstyle{-2{\pmb l}_1+{\pmb l}_2}$}
\psfrag{p+a}{$\scriptstyle{{\pmb l}_1}$}
\psfrag{p-a}{$\scriptstyle{-{\pmb l}_1}$}
\psfrag{p-a+b}{$\scriptstyle{-{\pmb l}_1+{\pmb l}_2}$}
\psfrag{p+a-b}{$\scriptstyle{{\pmb l}_1-{\pmb l}_2}$}
\psfrag{p-a-b}{$\scriptstyle{-{\pmb l}_1-{\pmb l}_2}$}
\psfrag{p+a+b}{$\scriptstyle{{\pmb l}_1+{\pmb l}_2}$}
\psfrag{p+2a}{$\scriptstyle{2{\pmb l}_1}$}
\psfrag{p+2a-b}{$\scriptstyle{2{\pmb l}_1-{\pmb l}_2}$}
\psfrag{p+3a-b}{$\scriptstyle{3{\pmb l}_1-{\pmb l}_2}$}
\vskip-20truemm
\includegraphics[width=145truemm,angle=-90,origin=bc]%,trim=200 100 800 400]%
% origin = centre for rotating, bc is bottom centre
% trim = left bottom right top (before rotating, I guess)
                {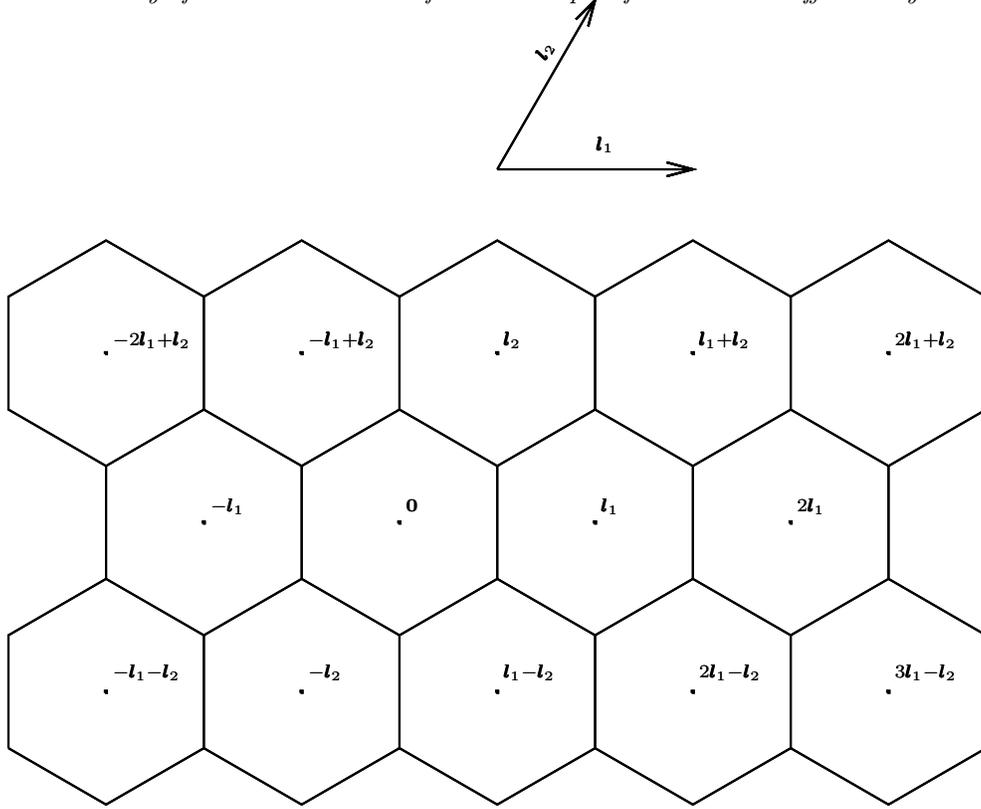}
\vskip-20truemm
\caption{Hexagonal lattice generated by the lattice vectors
  (\ref{hex-alp-bet}). The fundamental Wigner-Seitz cell $\Omega$ for
  this lattice is the regular hexagon centered at the origin. The area
  $\Omega$ and the primitive cell are the same, and are set to unity.}
\label{hex-cells}
\end{center}
\end{figure}

The {\em Wigner-Seitz or Voronoi cell} centered at a given lattice
point of $\Lambda$ consists of all points in the plane that are closer
to this point than to any other lattice point. It is constructed by
first joining the lattice point by a straight line to each of the
neighbouring lattice points. Then, by taking the perpendicular
bisector to each of these lines, the Wigner-Seitz cell is the smallest
area around this lattice point that is enclosed by all the
perpendicular bisectors. The Wigner-Seitz cell is a convex polygon
with the same area $|\lb_1\times\lb_2|$ of the primitive cell ${\cal
  P}$. In addition, it is well-known that the union of the
Wigner-Seitz cells for an arbitrary oblique Bravais lattice with
arbitrary lattice vectors $\lb_1,\lb_2$, and angle $\theta$, tile all of
$\R^2$ (cf.~\cite{am}). In other words, there holds
\begin{equation}
\label{news}
\R^2= \bigcup_{z \in \Lambda}  (z+\Omega) \,.
\end{equation}
 By periodicity and the property (\ref{news}), we need only consider the
Wigner-Seitz cell centered at the origin, which we denote by
$\Omega$. In Fig.~\ref{hex-cells} we show the fundamental Wigner-Seitz
cell for the hexagonal lattice.  In Fig.~\ref{wigner_2} we plot the
union of the Wigner-Seitz cells for an oblique Bravais lattice with
$\lb_1=(1,0)$, $\lb_2=(\cot\theta,1)$ and $\theta=74^{\circ}$.

\begin{figure}[htb]
\begin{center}
\includegraphics[width = 8cm,height=5.5cm,clip]{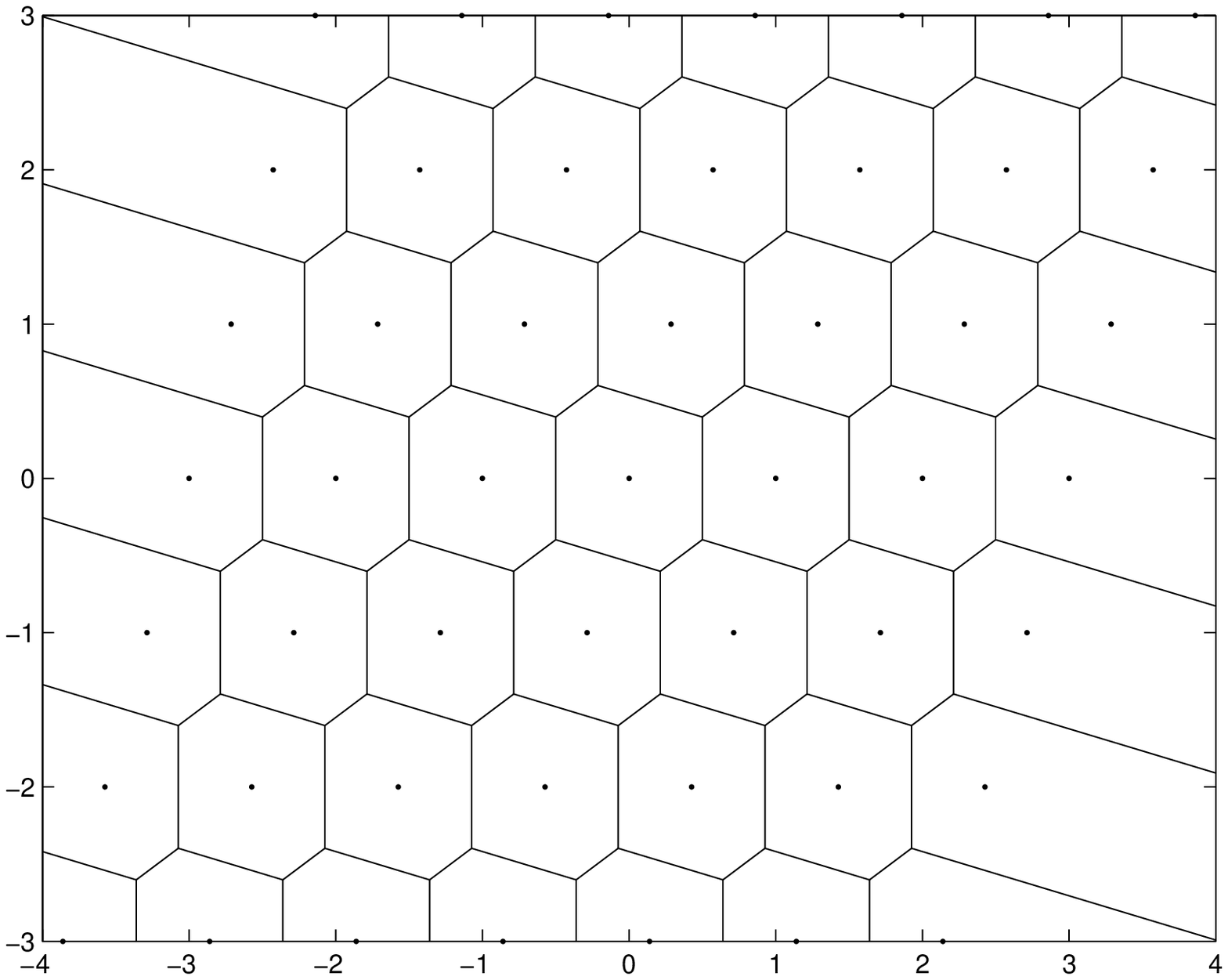}
\caption{Wigner-Seitz cells for an oblique lattice with
 $\lb_1=(1,0)$, $\lb_2=(\cot\theta,1)$, and
  $\theta=74^{\circ}$, so that $|\Omega|=1$. These cells tile the plane.
  The boundary of the Wigner-Seitz cells consist of three
  pairs of parallel lines of equal length.}
\label{wigner_2}
\end{center}
\end{figure}

As in \cite{Beyl-1}, we define the reciprocal lattice
$\Lambda^{\star}$ in terms of the two independent vectors $\pmb d_1$
and $\pmb d_2$, which are obtained from the lattice $\Lambda$ by
requiring that
  \begin{equation}\label{d-recip}
  {\pmb d}_i\cdot{\pmb l}_j = \delta_{ij}\,,
\end{equation}
where $\delta_{ij}$ is the Kronecker symbol.
The reciprocal lattice $\Lambda^{\star}$ is  defined by
 \begin{equation}
 \Lambda^{\star} = \Big{\lbrace}{ m \pmb d_1 + n \pmb d_2 \, \Big{\vert} \,\,
 m,\, n \in \mathbb{Z}\Big{\rbrace}} \,.  \label{lattice-recip}
\end{equation}
The first Brillouin zone, labeled by $\Omega_B$, is defined as the Wigner-Seitz
cell centered at the origin in the reciprocal space.

We remark that other authors (cf.~\cite{Linton}, \cite{Moroz}) define
the reciprocal lattice as $\Lambda^{\star}=\left\{ 2\pi m\, \pmb d_1,
2\pi n\,\pmb d_2\right\}_{m,n\in\Z}$. Our choice (\ref{lattice-recip})
for $\Lambda^{\star}$ is motivated by the form of the Poisson
summation formula of \cite{Beyl-1} given in (\ref{Poisson-sum}) below,
and which is used in \S~\ref{sec:ewald} to numerically compute the
Bloch Green's function.

\begin{figure} [htb]
\begin{center}
\psfrag{alpha}{$\scriptstyle{{\pmb l}_1}$}
\psfrag{beta}{$\scriptstyle{{\pmb l}_2}$}
\psfrag{p}{$\scriptstyle{\ \bf 0}$}
\psfrag{p+a}{$\scriptstyle{{\pmb l}_1}$}
\psfrag{p-a}{$\scriptstyle{-{\pmb l}_1}$}
\psfrag{p+2a}{$\scriptstyle{2{\pmb l}_1}$}
\psfrag{p-2a}{$\scriptstyle{-2{\pmb l}_1}$}
\psfrag{p+3a}{$\scriptstyle{3{\pmb l}_1}$}
\psfrag{p-3a}{$\scriptstyle{-3{\pmb l}_1}$}
\psfrag{p+b}{$\scriptstyle{\ \ {\pmb l}_2}$}
\psfrag{p-b}{$\scriptstyle{-{\pmb l}_2}$}
\psfrag{p+2b}{$\scriptstyle{2{\pmb l}_2}$}
\psfrag{p-2b}{$\scriptstyle{-2{\pmb l}_2}$}
\psfrag{p+3b}{$\scriptstyle{3{\pmb l}_2}$}
\psfrag{p-3b}{$\scriptstyle{-3{\pmb l}_2}$}
\psfrag{p+a+b}{$\scriptstyle{{\pmb l}_1+{\pmb l}_2}$}
\psfrag{p+a-b}{$\scriptstyle{{\pmb l}_1-{\pmb l}_2}$}
\psfrag{p-a+b}{$\scriptstyle{-{\pmb l}_1+{\pmb l}_2}$}
\psfrag{p-a-b}{$\scriptstyle{-{\pmb l}_1-{\pmb l}_2}$}
\psfrag{p-a-2b}{$\scriptstyle{-{\pmb l}_1-2{\pmb l}_2}$}
\psfrag{p-a-3b}{$\scriptstyle{-{\pmb l}_1-3{\pmb l}_2}$}
\psfrag{p+2a+b}{$\scriptstyle{2{\pmb l}_1+{\pmb l}_2}$}
\psfrag{p+2a-b}{$\scriptstyle{2{\pmb l}_1-{\pmb l}_2}$}
\psfrag{p-2a-b}{$\scriptstyle{-2{\pmb l}_1-{\pmb l}_2}$}
\psfrag{p-2a-3b}{$\scriptstyle{-2{\pmb l}_1-3{\pmb l}_2}$}
\psfrag{p-3a+b}{$\scriptstyle{-3{\pmb l}_1+{\pmb l}_2}$}
\psfrag{p-3a-2b}{$\scriptstyle{-3{\pmb l}_1-2{\pmb l}_2}$}
\psfrag{p-3a-3b}{$\scriptstyle{-3{\pmb l}_1-3{\pmb l}_2}$}
\psfrag{p-4a-b}{$\scriptstyle{-4{\pmb l}_1-{\pmb l}_2}$}
\psfrag{p-4a-2b}{$\scriptstyle{-4{\pmb l}_1-2{\pmb l}_2}$}
\psfrag{p-5a-3b}{$\scriptstyle{-5{\pmb l}_1-3{\pmb l}_2}$}
\vskip-5truemm\hskip-10truemm
\subfloat[Lattice $\Lambda$]
{\includegraphics[width=95truemm,angle=-90,origin=bc]
%,trim=200 100 800 400]%
% origin = centre for rotating, bc is bottom centre
% trim = left bottom right top (before rotating, I guess)
{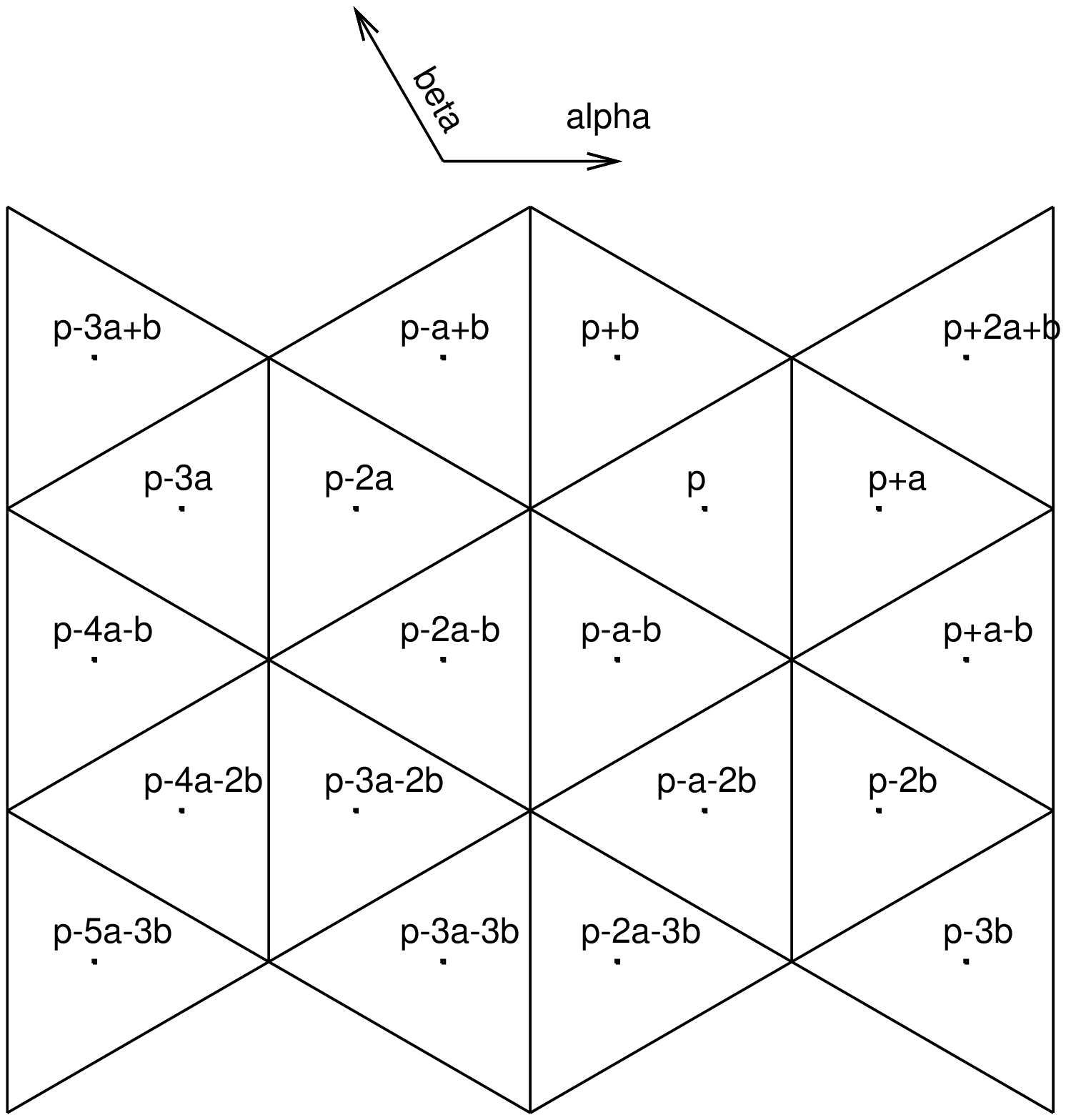} }
\psfrag{a}{$\scriptstyle{{\pmb d}_1}$}
\psfrag{b}{$\scriptstyle{{\pmb d}_2}$}
\psfrag{p}{$\scriptstyle{\ \bf 0}$}
\psfrag{p+a}{$\scriptstyle{{\pmb d}_1}$}
\psfrag{p-a}{$\scriptstyle{-{\pmb d}_1}$}
\psfrag{p+2a}{$\scriptstyle{2{\pmb d}_1}$}
\psfrag{p-2a}{$\scriptstyle{-2{\pmb d}_1}$}
\psfrag{p+3a}{$\scriptstyle{3{\pmb d}_1}$}
\psfrag{p-3a}{$\scriptstyle{-3{\pmb d}_1}$}
\psfrag{p+b}{$\scriptstyle{\ \ {\pmb d}_2}$}
\psfrag{p-b}{$\scriptstyle{-{\pmb d}_2}$}
\psfrag{p+2b}{$\scriptstyle{2{\pmb d}_2}$}
\psfrag{p-2b}{$\scriptstyle{-2{\pmb d}_2}$}
\psfrag{p+3b}{$\scriptstyle{3{\pmb d}_2}$}
\psfrag{p-3b}{$\scriptstyle{-3{\pmb d}_2}$}
\psfrag{p+a+b}{$\scriptstyle{{\pmb d}_1+{\pmb d}_2}$}
\psfrag{p+a+2b}{$\scriptstyle{{\pmb d}_1+2{\pmb d}_2}$}
\psfrag{p+a+3b}{$\scriptstyle{{\pmb d}_1+3{\pmb d}_2}$}
\psfrag{p+a-b}{$\scriptstyle{{\pmb d}_1-{\pmb d}_2}$}
\psfrag{p+a-2b}{$\scriptstyle{{\pmb d}_1-2{\pmb d}_2}$}
\psfrag{p+a-3b}{$\scriptstyle{{\pmb d}_1-3{\pmb d}_2}$}
\psfrag{p-a+b}{$\scriptstyle{-{\pmb d}_1+{\pmb d}_2}$}
\psfrag{p-a+2b}{$\scriptstyle{-{\pmb d}_1+2{\pmb d}_2}$}
\psfrag{p-a-b}{$\scriptstyle{-{\pmb d}_1-{\pmb d}_2}$}
\psfrag{p-a-2b}{$\scriptstyle{-{\pmb d}_1-2{\pmb d}_2}$}
\psfrag{p-a-3b}{$\scriptstyle{-{\pmb d}_1-3{\pmb d}_2}$}
\psfrag{p+2a+b}{$\scriptstyle{2{\pmb d}_1+{\pmb d}_2}$}
\psfrag{p+2a-2b}{$\scriptstyle{2{\pmb d}_1-2{\pmb d}_2}$}
\psfrag{p+2a-b}{$\scriptstyle{2{\pmb d}_1-{\pmb d}_2}$}
\psfrag{p-2a-b}{$\scriptstyle{-2{\pmb d}_1-{\pmb d}_2}$}
\psfrag{p-2a+2b}{$\scriptstyle{-2{\pmb d}_1+2{\pmb d}_2}$}
\psfrag{p-2a+3b}{$\scriptstyle{-2{\pmb d}_1+3{\pmb d}_2}$}
\psfrag{p-2a-3b}{$\scriptstyle{-2{\pmb d}_1-3{\pmb d}_2}$}
\psfrag{p-3a+2b}{$\scriptstyle{-3{\pmb d}_1+2{\pmb d}_2}$}
\psfrag{p-2a+b}{$\scriptstyle{-2{\pmb d}_1+{\pmb d}_2}$}
\psfrag{p-3a-3b}{$\scriptstyle{-3{\pmb d}_1-3{\pmb d}_2}$}
\psfrag{p+4a-3b}{$\scriptstyle{4{\pmb d}_1-3{\pmb d}_2}$}
\psfrag{p+3a-b}{$\scriptstyle{3{\pmb d}_1-{\pmb d}_2}$}
\psfrag{p+3a-2b}{$\scriptstyle{3{\pmb d}_1-2{\pmb d}_2}$}
\psfrag{p-5a-3b}{$\scriptstyle{-5{\pmb d}_1-3{\pmb d}_2}$}
\hskip-15mm
\subfloat[Reciprocal Lattice $\Lambda^*$]
{\includegraphics[width=95truemm,angle=-90,origin=bc]
%,trim=200 100 800 400]%
% origin = centre for rotating, bc is bottom centre
% trim = left bottom right top (before rotating, I guess)
{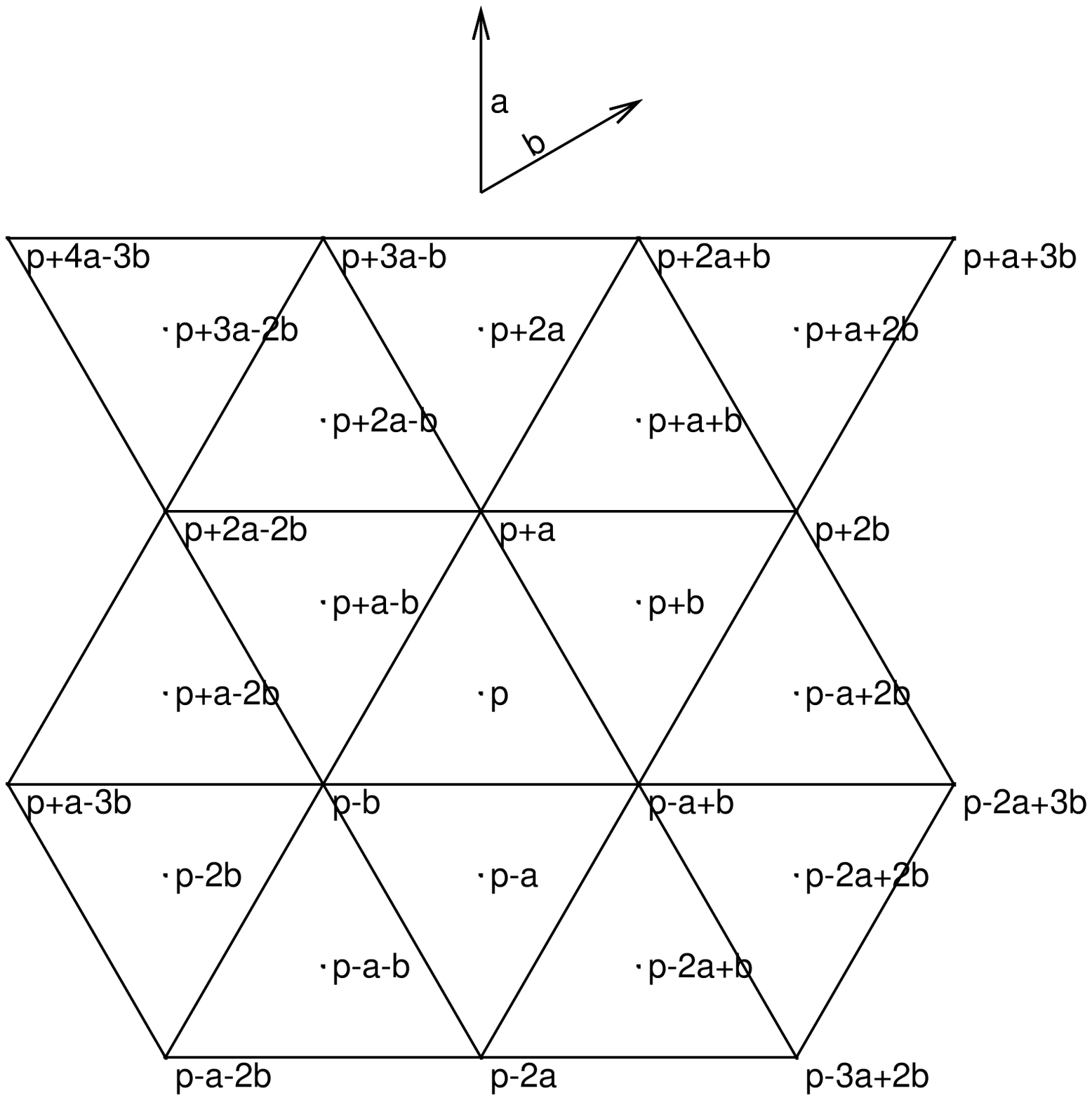} }
\caption{Left panel: Triangular lattice $\Lambda$ with unit area of
  the primitive cell generated by the lattice vectors in
  (\ref{spat-tri}). Right panel: the
  corresponding reciprocal lattice $\Lambda^{*}$ with reciprocal
  lattice vectors as in (\ref{tri:recip}).} \label{triang-cells}
 \end{center}
\end{figure}

Finally we make some remarks on the equilateral triangular lattice which
does not fall into the framework discussed above. As
observed in \cite{chen}, this special lattice requires a
different treatment. For the equilateral triangle lattice,
$\theta={2\pi/3}$ and $\im\!\left(e^{2i\pi/3}\right) = \sqrt3/2$, so that
the unit area requirement of the primitive cell again yields
$\alpha=(4/3)^{1/4}$. Since $\re\!\left(e^{2i\pi/3}\right) =-1/2$, it
follows that in terms of ${\pmb l}_i\in\R^2$ for $i=1,2$, an
equilateral triangle cell structure has
 \begin{equation}\label{spat-tri}
 {\pmb l}_1 =\left(\left(\frac 43\right)^{1/4},0\right)
  \qquad\text{and}\qquad
 {\pmb l}_2 =\left(\frac 43\right)^{1/4}
\left(-\frac 12,\frac{\sqrt 3}{2}\right)\,.
 \end{equation}
This triangular lattice is shown in Fig.~\ref{triang-cells}.  The
centers of the triangular cells are generated by (\ref{lattice-def}),
but there are points in $\Lambda$ which are not cell centers (see
Fig.~\ref{triang-cells}).  For example, $(3n+1){\pmb l}_1 +{\pmb
  l}_2$, $(3n+2){\pmb l}_1$, $3n{\pmb l}_1-{\pmb l}_2$, and
$(3n+1){\pmb l}_1 - 2{\pmb l}_2$ are not centers of cells of
equilateral triangles.  In general, for integers $p$ and $q$ the point
$p\,{\pmb l}_1 + q{\pmb l}_2$ will be a vertex instead of a cell
center when
\begin{equation}\label{trilat}
\left(p\!\!\mod 3\right) + \left(q\!\!\mod 3 \right)= 2\,,
\end{equation}
where the positive representation of the mod function is used,
i.e.~$(-1)\mod 3 = 2.$ Thus, for the equilateral triangular lattice the set
of lattice points is
\begin{equation}
  \Lambda_{tri} = \Big{\lbrace}{ m \lb_1 + n \lb_2 \, \Big{\vert} \,\,
    m,\, n \in \mathbb{Z}\,, \,\,
\left(m\!\!\mod 3\right) + \left(n\!\!\mod 3 \right)\not = 2\,\Big{\rbrace}}
 \,. \label{lattice-def1}
\end{equation}
The corresponding  Wigner-Seitz cell is also an equilateral triangle.

Regarding the reciprocal lattice for the equilateral triangular lattice
with ${\pmb l}_1$ and ${\pmb l}_2$ given by \eqref{spat-tri}, the
defining vectors for $\Lambda^{\star}$ are
 \begin{equation}
  {\pmb d}_1 = \frac 1{12^{1/4}} \left(\sqrt 3,1\right)
 \qquad\text{and}\qquad
  {\pmb d}_2 = \frac 1{12^{1/4}} \left(0,2\right) \,, \label{tri:recip}
 \end{equation}
as can be verified by substitution into \eqref{d-recip}. A plot of a
portion of this reciprocal lattice for the equilateral triangle
lattice is shown in the right panel of Fig.~\ref{triang-cells}. From
this plot it follows that, for integer $p$ and $q$, $p\, {\pmb d}_1 +
 q\, {\pmb d}_2$ will be a vertex, not a centre, when
\begin{equation}\label{trilat-recip}
(p-q)\!\!\mod 3= 1\,.
\end{equation}
Therefore the reduced reciprocal lattice becomes
\begin{equation}
  \Lambda_{tri}^{\star} = \Big{\lbrace}{ m {\pmb d}_1 + n {\pmb d}_2
    \, \Big{\vert} \,\, m,\, n \in \mathbb{Z}\,, \,\,\, (m-n)\!\!\mod 3\not =
    1\,\Big{\rbrace}} \,. \label{lattice-def1a}
\end{equation}

Unfortunately for the equilateral triangular lattice the property
(\ref{news}) does not hold. In other words, the whole $\R^2$ is not the
union of cells translated on the Bravais lattice, and thus one can not
restrict to one Wigner-Seitz cell at the origin. As such, it is
unclear whether the corresponding Poisson summation formula in
(\ref{Poisson-sum}) below still holds. However, if a homogeneous
Neumann boundary condition is imposed on the cell, it is possible to
reflect through the edges and fill the whole ${\mathbb R}^2$. (This
fact has been used in \cite{chen}.) Therefore, the equilibrium
contruction of a periodic spot pattern presented in Section 3.1 and
Section 4.1 still applies for the equilateral triangular
lattice. However, the stability of periodic spot patterns on the
triangular lattice is an open problem.

\subsection{A Few Key Properties of the Bloch Green's Functions}
\label{2:gr_lattice}

In our analysis of the stability of spot patterns in
\S~\ref{schnak:stab} and \S~\ref{gm:stab} below, the Bloch Green's
function $G_{b0}(\xb)$ for the Laplacian plays a prominent role. In
the Wigner-Seitz cell $\Omega$, $G_{b0}(\xb)$ for ${\kb/(2\pi)}\in
\Omega_B$, satisfies
\bsub \label{green:b0}
\begin{equation}
\Delta G_{b0} = -\delta(\xb)\,; \qquad \xb \in \Omega \,, \label{green:b01}
\end{equation}
subject to the quasi-periodicity condition on $\R^2$ that
\begin{equation}
 G_{b0}(\xb + \pmb l) = e^{-i\pmb k\cdot\pmb l}\, G_{b0}(\xb)\,,
 \qquad \pmb l\in\Lambda \,, \label{green:b1}
\end{equation}
where $\Lambda$ is the Bravais lattice (\ref{lattice-def}). As we show
below, (\ref{green:b1}) indirectly yields boundary conditions on the
boundary $\partial\Omega$ of the Wigner-Seitz cell.  The regular part
$R_{b0}(\kb)$ of this Bloch Green's function is defined by
\begin{equation}
     R_{b0}(\kb)\equiv \lim_{\xb\to \xbo} \left( G_{b0}(\xb) +
     \frac{1}{2\pi} \log|\xb| \right) \,. \label{green:b2}
\end{equation}
\esub

In order to study the properties of $G_{b0}(\xb)$ and $R_{b0}(\kb)$, we
first require a more refined description of the Wigner-Seitz cell. To
do so, we observe that there are eight nearest neighbor lattice
points to $\xb=\xbo$ given by the set
\begin{equation}
 P\equiv \lbrace{ \, m\lb_1+ n\lb_2 \, \vert \,\, m\in
   \lbrace{0,1,-1\rbrace}\,, \,\, n\in \lbrace{0,1,-1\rbrace}
   \,,\,\,\, (m,n)\neq 0 \rbrace} \,. \label{gr:nn}
\end{equation}
For each (vector) point $\pbi\in P$, for $i=1,\ldots,8$, we define a
Bragg line $L_i$. This is the line that crosses the point ${\pbi/2}$
orthogonally to $\pbi$. We define the unit outer normal to $L_i$
by $\etai\equiv{\pbi/|\pbi|}$. The convex hull generated by these
Bragg lines is the Wigner-Seitz cell $\Omega$, and the boundary
$\partial\Omega$ of the Wigner-Seitz cell is, generically, the union
of six Bragg lines. For a square lattice, $\partial\Omega$ has four Bragg
lines.  The centers of the Bragg lines generating $\partial\Omega$ are
re-indexed as $\pbi$ for $i=1,\ldots,L$, where
$L\in\lbrace{4,6\rbrace}$ is the number of Bragg lines de-marking
$\partial\Omega$.  The boundary $\partial\Omega$ of $\Omega$ is the
union of the re-indexed Bragg lines $L_i$, for $i=1,\ldots,L$, and is
parametrized segment-wise by a parameter $t$ as
\begin{equation}
   \partial\Omega = \Big{\lbrace}  \xb \in \bigcup_i \lbrace{
  \frac{\pbi}{2} + t \etaip\rbrace} \,\, \Big{\vert} \,\,\,
  -t_i\leq t\leq t_i\,, \,\,\, i=1,\ldots,L\,, \,\,\,
  L=\lbrace{4,6\rbrace} \Big{\rbrace} \,. \label{gr:boundary}
\end{equation}
Here $2t_i$ is the length of $L_i$, and $\etaip$ is the direction perpendicular
to $\pbi$, and therefore tangent to $L_i$.

The following observation is central to the analysis below:
Suppose that $\pb$ is a neighbor of $\xbo$ and that the Bragg line
crossing ${\pb/2}$ lies on $\partial\Omega$. Then, by symmetry, the
Bragg line crossing ${-\pb/2}$ must also lie on $\partial\Omega$. In
other words, Bragg lines on $\partial\Omega$ must come in pairs. This
fact is evident from the plot of the Wigner-Seitz cell for the oblique
lattice shown in Fig.~\ref{wigner_2}. With this more refined
description of the Wigner-Seitz cell, we now state and prove two key
Lemmas that are needed in \S~\ref{schnak:stab} and \S~\ref{gm:stab} below.

\begin{lemma}\label{lemma 2.1} The regular part $R_{b0}(\kb)$ of
the Bloch Green's function $G_{b0}(\xb)$ satisfying
(\ref{green:b0}) is real-valued for $|\kb|\neq 0$.
\end{lemma}

\begin{proof}
Let $0<\rho\ll 1$ and define
$\Omega_\rho\equiv \Omega-B_{\rho}(\xbo)$, where $B_{\rho}(\xbo)$ is the
ball of radius $\rho$ centered at $\xb=0$. We multiply
(\ref{green:b01}) by ${\bar G}_{b0}$, where the bar denotes
conjugation, and we integrate over $\Omega_\rho$ using the divergence
theorem to get
\begin{equation}
  \int_{\Omega_\rho} {\bar G}_{b0} \Delta G_{b0} \, d\xb
  +\int_{\Omega_\rho} \nabla {\bar G}_{b0} \cdot \nabla G_{b0} \, d\xb
  = \int_{\partial\Omega_\rho} {\bar G}_{b0} \, \partial_\nu G_{b0}\, d\xb =
\int_{\partial\Omega} {\bar G}_{b0} \, \partial_\nu G_{b0} \, d\xb - \int_{\partial
 B_{\rho}(\xbo)} {\bar G}_{b0} \, \partial_{|\xb|} G_{b0} \, d\xb \,.\label{gr:rd1}
\end{equation}
Here $\partial_\nu G_{b0}$ denotes the outward normal derivative of
$G_{b0}$ on $\partial\Omega$.  For $\rho\ll 1$, we use (\ref{green:b2}) to
calculate
\begin{equation}
 \int_{\partial B_{\rho}(\xbo)} {\bar G}_{b0} \, \partial_{|\xb|} G_{b0} \, d\xb
 \sim \int_{0}^{2\pi}
 \left(-\frac{1}{2\pi} \log\rho + R_{b0}(\kb) + o(1)\right)
 \left(-\frac{1}{2\pi\rho} + {\mathcal O}(1)\right)\, \rho \, d\theta
 \sim \frac{1}{2\pi}\log\rho - R_{b0}(\kb) + {\mathcal
   O}(\rho\log\rho)\,.
 \label{gr:rd2}
\end{equation}
Upon using (\ref{gr:rd2}), together with $\Delta G_{b0}=0$ in $\Omega_\rho$,
in equation (\ref{gr:rd1}), we let $\rho\to 0$ to obtain
\begin{equation}
    R_{b0}(\kb)= -\int_{\partial\Omega} {\bar G}_{b0}(\xb)
  \, \partial_\nu G_{b0}(\xb) \, d\xb +
  \lim_{\rho\to 0} \Big{\lbrack} \int_{\Omega_\rho} \vert \nabla G_{b0}\vert^2 \,
 d\xb + \frac{1}{2\pi} \log\rho \Big{\rbrack} \,. \label{gr:rd3}
\end{equation}

From (\ref{gr:rd3}), to show that $R_{b0}(\kb)$ is real-valued it
suffices to establish that the boundary integral term in (\ref{gr:rd3})
vanishes. To show this, we observe that since the Bragg lines come in pairs,
we have
\begin{equation}
\int_{\partial\Omega} {\bar G}_{b0}(\xb) \,
  \partial_\nu G_{b0}(\xb) \, d\xb =
  \sum_{i=1}^{L/2} \left( \int_{\frac{\pbi}{2}+t\etaip}
  {\bar G}_{b0}(\xb)   \nabla_{\xb} G_{b0}(\xb) \cdot \etai \, d\xb -
\int_{\frac{-\pbi}{2}+t\etaip} {\bar G}_{b0}(\xb)
  \nabla_{\xb} G_{b0}(\xb) \cdot \etai \, d\xb \right) \,. \label{gr:bd1}
\end{equation}
Here we have used the fact that the outward normals to the Bragg line pairs
${\pbi/2}+t\etaip$ and $-{\pbi/2}+t\etaip$ are in opposite directions.
We then translate $\xb$ by $\pbi$ to get
\begin{equation}
\int_{\frac{\pbi}{2}+t\etaip} {\bar G}_{b0}(\xb) \nabla_{\xb} G_{b0}(\xb) \cdot
 \etai \, d\xb =\int_{\frac{-\pbi}{2}+t\etaip + \pbi} {\bar G}_{b0}(\xb)
  \nabla_{\xb} G_{b0}(\xb) \cdot \etai \, d\xb =
\int_{\frac{-\pbi}{2}+t\etaip } {\bar G}_{b0}(\xb+\pbi)
  \nabla_{\xb} G_{b0}(\xb+\pbi) \cdot \etai \, d\xb\,. \label{gr:bd2}
\end{equation}

Then, since $\pbi\in\Lambda$, we have by the quasi-periodicity condition
(\ref{green:b1}) that
\begin{equation*}
  {\bar G_{b0}(\xb+\pbi)} \nabla_{\xb} G_{b0}(\xb+\pbi) =\left(
  {\bar G_{b0}(\xb)} e^{i\kb\cdot\pbi} \right) \left(
 \nabla_{\xb} G_{b0}(\xb) e^{-i\kb\cdot\pbi} \right) =
   {\bar G}_{b0}(\xb) \nabla_{\xb} G_{b0}(\xb) \,.
\end{equation*}
Therefore, from (\ref{gr:bd2}) we conclude that
\begin{equation*}
\int_{\frac{\pbi}{2}+t\etaip} {\bar G}_{b0}(\xb) \nabla_{\xb} G_{b0}(\xb) \cdot
 \etai \, d\xb = \int_{\frac{-\pbi}{2}+t\etaip } {\bar G}_{b0}(\xb)
  \nabla_{\xb} G_{b0}(\xb) \cdot \etai \, d\xb\,,
\end{equation*}
which establishes from (\ref{gr:bd1}) that
$\int_{\partial{\Omega}} \bar{G}_{b0}(\xb) \, \partial_\nu G_{b0}(\xb) \, d\xb=0$.
From (\ref{gr:rd3}) we conclude that $R_{b0}(\kb)$ is real.
\end{proof}

Next, we determine the asymptotic behavior of $R_{b0}(\kb)$ as
$|\kb|\to 0$.  Since (\ref{green:b0}) has no solution if $\kb=0$, it
suggests that $R_{b0}(\kb)$ is singular as $|\kb|\to 0$. To determine the
asymptotic behavior of $G_{b0}$ as $|\kb|\to 0$, we introduce a
small parameter $\sigma\ll 1$, and define $\kb=\sigma\kapb$ where
$|\kapb|={\mathcal O}(1)$. For $\sigma\ll 1$, we expand $G_{b0}(\xb)$
as
\begin{equation}
    G_{b0}(\xb) = \sigma^{-2} {\cal U}_0(\xb) + \sigma^{-1} {\cal U}_1(\xb)
  + {\cal U}_2(\xb) + \cdots \,. \label{gre:exp}
\end{equation}
For any $\lb\in \Omega$, and for $\sigma\ll 1$, we have from
(\ref{green:b1}) that
\begin{equation}
  \frac{{\cal U}_0(\xb+\lb)}{\sigma^2} +  \frac{{\cal U}_1(\xb+\lb)}{\sigma} +
  {\cal U}_2(\xb+\lb) + \cdots = \left[ 1 - i\sigma(\kapb\cdot\lb) -
  \frac{\sigma^2}{2}(\kapb\cdot\lb)^2 + \cdots\right]\left(
  \frac{{\cal U}_0(\xb)}{\sigma^2} + \frac{{\cal U}_1(\xb)}{\sigma} +
  {\cal U}_2(\xb) + \cdots \right)\,. \label{gre:exp1}
\end{equation}
Upon substituting (\ref{gre:exp}) into (\ref{green:b01}), and then equating
powers of $\sigma$ in (\ref{gre:exp1}), we obtain the sequence of problems
\bsub \label{gre:ueq}
\begin{align}
   \Delta {\cal U}_0 &= 0 \,; \qquad {\cal U}_0(\xb+\lb)={\cal U}_0(\xb) \,,
  \label{gre:ueq_1}\\
   \Delta {\cal U}_1 &= 0 \,; \qquad {\cal U}_1(\xb+\lb)={\cal U}_1(\xb)
  - i \left( \kapb \cdot \lb \right) {\cal U}_0(\xb)  \,,
  \label{gre:ueq_2}\\
   \Delta {\cal U}_2 &= -\delta(\xb) \,; \qquad {\cal
     U}_2(\xb+\lb)={\cal U}_2(\xb) -i \left( \kapb \cdot \lb \right)
          {\cal U}_1(\xb) - \frac{ (\kapb \cdot \lb)^2}{2} {\cal
            U}_0(\xb) \,. \label{gre:ueq_3}
\end{align}
\esub

The solution to (\ref{gre:ueq_1}) is that ${\cal U}_0$ is an arbitrary
constant, while the solution to (\ref{gre:ueq_2}) is readily calculated
as ${\cal U}_1(\xb)=-i\left(\kapb \cdot \xb\right){\cal U}_0 + {\cal
  U}_{10}$, where ${\cal U}_{10}$ is an arbitrary constant. Upon
substituting ${\cal U}_0$ and ${\cal U}_1$ into (\ref{gre:ueq_3}), we
obtain for any $\lb\in \Lambda$ that ${\cal U}_2$ satisfies
\begin{equation}
   \Delta {\cal U}_2 = -\delta(\xb) \,; \qquad
 {\cal U}_2(\xb+\lb)={\cal U}_2(\xb)  - \left( \kapb \cdot \lb \right)
  \left(\kapb \cdot \xb\right) {\cal U}_0
  - i\left(\kapb \cdot \lb \right) {\cal U}_{10} - \frac{ (\kapb \cdot \lb)^2}{2}
   {\cal U}_0  \,. \label{gre:nueq_3}
\end{equation}
By differentiating the periodicity condition in (\ref{gre:nueq_3})
with respect to $\xb$, we have for any $\lb\in \Lambda$ that
\begin{equation}
    \nabla_{\xb} {\cal U}_2(\xb+\lb) = \nabla_{\xb} {\cal U}_2(\xb) -
            \kapb \left(\kapb \cdot \lb\right) {\cal U}_0 \,. \label{gre:bc2}
\end{equation}

Next, to determine ${\cal U}_0$, we integrate $\Delta {\cal U}_2=0$
over $\Omega$ to obtain from the divergence theorem and a subsequent
decomposition of the boundary integral over the Bragg line pairs, as
in (\ref{gr:bd1}), that
\begin{equation}
 -1 = \int_{\partial\Omega} \partial_\nu {\cal U}_2 \, d\xb =
  \sum_{i=1}^{L/2} \left( \int_{\frac{\pbi}{2}+t\etaip}
  \nabla_{\xb} {\cal U}_2(\xb) \cdot \etai \, d\xb -
\int_{\frac{-\pbi}{2}+t\etaip}   \nabla_{\xb} {\cal U}_2(\xb) \cdot \etai \, d\xb
 \right) \,. \label{gre:bd1}
\end{equation}
Then, as similar to the derivation in (\ref{gr:bd2}), we calculate the
boundary integrals as
\begin{equation}
\int_{\frac{\pbi}{2}+t\etaip} \nabla_{\xb} {\cal U}_2(\xb) \cdot \etai \, d\xb =
\int_{\frac{-\pbi}{2}+t\etaip + \pbi} \nabla_{\xb} {\cal U}_2(\xb) \cdot \etai \, d\xb
 = \int_{\frac{-\pbi}{2}+t\etaip } \nabla_{\xb} {\cal U}_{2}(\xb+\pbi) \cdot
    \etai \, d\xb\,. \label{gre:bd2}
\end{equation}
Upon using (\ref{gre:bd2}) in (\ref{gre:bd1}), we obtain
\begin{equation}
 -1 = \sum_{i=1}^{L/2} \int_{-\frac{\pbi}{2}+t\etaip}
  \left( \nabla_{\xb} {\cal U}_2(\xb+\pbi)- \nabla_{\xb} {\cal U}_2(\xb)\right)
 \cdot \etai \, d\xb \,. \label{gre:bd3}
\end{equation}
Since $\pbi\in \Lambda$ and $\etai={\pbi/|\pbi|}$, we calculate the integrand
in (\ref{gre:bd3}) by using (\ref{gre:bc2}) as
\begin{equation}
\left( \nabla_{\xb} {\cal U}_2(\xb+\pbi)- \nabla_{\xb} {\cal
    U}_2(\xb)\right) \cdot \etai = -\left(\kapb \cdot \etai\right)
\left(\kapb\cdot \pbi\right){\cal U}_0 = -\left( \kapb\cdot
\pbi\right)^2 \frac{ {\cal U}_0}{|\pbi|}\,.
  \label{gre:bd4}
\end{equation}

Then, upon substituting (\ref{gre:bd4}) into (\ref{gre:bd3}), and by
integrating the constant integrand over the Bragg lines, we obtain
that ${\cal U}_0$ satisfies
\begin{equation}
   -1 = -{\cal U}_0 \sum_{i=1}^{L/2}
 \frac{ \left( \kapb\cdot \pbi\right)^2}{|\pbi|} 2t_i =
-{\cal U}_0 \sum_{i=1}^{L}
 \frac{ \left( \kapb\cdot \pbi\right)^2}{|\pbi|} t_i  =
 -{\cal U}_0 \sum_{i=1}^{L} \left( \kapb\cdot \etai\right)^2 t_i |\pbi| \,,
  \label{gre:solve}
\end{equation}
where $2t_i$ is the length of the Bragg line $L_i$. Upon solving for
${\cal U}_0$, we obtain that
\begin{equation}
   {\cal U}_0 = \frac{1}{\kapb^{T} {\cal Q} \kapb} \,, \qquad \mbox{where} \quad
  {\cal Q} \equiv \sum_{i=1}^{L} \etai \omega_i \etai^T \,, \quad \mbox{and}
 \quad \omega_i \equiv t_i |\pbi| \,. \label{gre:fin}
\end{equation}
Since $\omega_i>0$, for $i=1,\ldots,L$, we have
$\yb^T {\cal Q} \yb = \sum_{i=1}^{L} \left(\etai^T\yb\right)^2 \omega_i>0$ for
any $\yb\neq 0$, which proves that the matrix ${\cal Q}$ is
positive definite.  We summarize the results of this perturbation
calculation in the following (formal) lemma:

\begin{lemma}\label{lemma 2.2} For $|\kb|\to 0$, the regular part $R_{b0}(\kb)$
of the Bloch Green's function of (\ref{green:b0}) has the
leading-order singular asymptotic behavior
\begin{equation}
   R_{b0}(\kb) \sim  \frac{1}{\kb^T {\cal Q} \kb} \,,
\end{equation}
where the positive definite matrix ${\cal Q}$ is defined in terms of the
parameters of the Wigner-Seitz cell by (\ref{gre:fin}).
\end{lemma}

We remark that a similar analysis can be done for the quasi-periodic
reduced-wave Green's function, which satisfies
\bsub \label{ngreen:b}
\begin{equation}
\Delta G(\xb)-\sigma^2 G = -\delta(\xb)\,; \qquad \xb \in \Omega \,; \qquad
 G(\xb + \pmb l) = e^{-i\pmb k\cdot\pmb l}\, G(\xb)\,,
 \qquad \pmb l\in\Lambda \,, \label{ngreen:b0}
\end{equation}
where $\Lambda$ is the Bravais lattice (\ref{lattice-def}) and
${\kb/(2\pi)}\in\Omega_B$. The regular part $R(\kb)$ of this Green's
function is defined by
\begin{equation}
     R(\kb)\equiv \lim_{\xb\to \xbo} \left( G(\xb) +
     \frac{1}{2\pi} \log|\xb| \right) \,. \label{ngreen:b2}
\end{equation}
\esub

By a simple modification of the derivation of Lemma \ref{lemma 2.1} and
\ref{lemma 2.2}, we obtain the following result:

\begin{lemma}\label{lemma 2.3} Let ${\kb/(2\pi)}\in\Omega_B$. For the regular
part $R(\kb)$ of the reduced-wave Bloch Green's function satisfying
(\ref{ngreen:b}), we have the following:
\begin{itemize}
   \item \mbox{(i)} $\,$ Let $\sigma^2$ be real. Then $R(\kb)$ is
     real-valued.
   \item \mbox{(ii)} $\,$ $R(\kb)\sim R_{b0}(\kb) + {\mathcal O}(\sigma^2)$
     for $\sigma\to 0$ when $|\kb|>0$ with $|\kb|={\mathcal O}(1)$. Here
     $R_{b0}(\kb)$ is the regular part of the Bloch Green's function
     (\ref{green:b0}).
    \item \mbox{(iii)} $\,$ Let $\sigma\to 0$, and
    consider the long-wavelength regime $|\kb|={\mathcal O}(\sigma)$, where
    $\kb=\sigma \kapb$ with $|\kapb|={\mathcal O}(1)$. Then,
\begin{equation}
      R(\kb) \sim  \frac{1}{\sigma^2\left[ |\Omega| + \kapb^T {\cal Q}\kapb
   \right]} \,, \label{ngre:blow}
\end{equation}
where the positive definite matrix ${\cal Q}$ is defined in (\ref{gre:fin}).
\end{itemize}
\end{lemma}

\begin{proof}
To prove (i) we proceed as in the derivation of Lemma \ref{lemma 2.1} to get
\begin{equation}
 R(\kb) = \lim_{\rho\to 0} \Big{\lbrack} \int_{\Omega_\rho} \left(
 \vert \nabla G\vert^2 + \sigma^2 \vert G\vert^2 \right) \, d\xb
  + \frac{1}{2\pi} \log\rho \Big{\rbrack} \,, \label{ngr:rd3}
\end{equation}
which is real-valued. The second result (ii) is simply a regular
perturbation result for the solution to (\ref{ngreen:b}) for
$\sigma\to 0$ when $|\kb|$ is bounded away from zero and ${\kb/(2\pi)}
\in \Omega_B$, so that $\kb\cdot \lb\neq 2\pi N$. Therefore, when
${\kb/(2\pi)}\in\Omega_B$, $R(\kb)$ is unbounded only as $|\kb|\to
0$ . To establish the third result, we proceed as in
(\ref{gre:exp})--(\ref{gre:bc2}), with the modification that $\Delta
{\cal U}_2={\cal U}_0 - \delta(\xb)$ in $\Omega$, Therefore, we must add
the term ${\cal U}_0|\Omega|$ to the left-hand sides of
(\ref{gre:bd1}), (\ref{gre:bd3}), and (\ref{gre:solve}). By solving
for ${\cal U}_0$ we get (\ref{ngre:blow}).
\end{proof}

In \S \ref{schnak:stab} and \S \ref{gm:stab} below, we will analyze
the spectrum of the linearization around a steady-state periodic spot
pattern for the Schnakenburg and GM models. For $\eps\to 0$, it is
the eigenfunction $\Psi$ corresponding to the long-range solution component
$u$ that satisfies an elliptic PDE with coefficients that are spatially
periodic on the lattice. As such, by the Floquet-Bloch theorem
(cf.~\cite{K} and \cite{KR}), this eigenfunction must satisfy the
quasi-periodic boundary conditions $\Psi(\xb+\lb)=e^{-i\kb\cdot \lb}\Psi(\xb)$
for $\lb\in \Lambda$, $\xb\in\R^2$ and ${\kb/(2\pi)}\in \Omega_B$.
This quasi-periodicity condition can be used to formulate a
boundary operator on the boundary $\partial\Omega$ of the fundamental
Wigner-Seitz cell $\Omega$. Let $L_{i}$ and $L_{-i}$ be two parallel
Bragg lines on opposite sides of $\partial\Omega$ for
$i=1,\ldots,{L/2}$. Let $\xb_{i1}\in L_i$ and $\xb_{i2}\in L_{-i}$ be
any two opposing points on these Bragg lines. We define the boundary
operator ${\cal P}_k\Psi$ by
\begin{equation}
  {\cal P}_{k}\Psi = \Big{\lbrace} \Psi \, \Big{\vert} \,
  \left( \begin{array}{c} \Psi(\xb_{i1}) \\ \partial_n
    \Psi(\xb_{i1}) \end{array} \right) = e^{-i\kb\cdot{\lb_i}}
  \left( \begin{array}{c} \Psi(\xb_{i2}) \\ \partial_n
    \Psi(\xb_{i2}) \end{array} \right) \,, \quad \forall\, \xb_{i1}\in
  L_i \,, \,\,\, \forall\,\xb_{i2}\in L_{-i}\,, \,\,\, {\lb}_i \in
  \Lambda\,, \,\,\, i=1,\ldots,{L/2} \Big{\rbrace} \,. \label{gr:pk}
\end{equation}
The boundary operator ${\cal P}_0\Psi$ simply corresponds to a periodicity
condition for $\Psi$ on each pair of parallel Bragg lines. These boundary
operators are used in \S \ref{schnak} and \S \ref{gm} below.

\setcounter{equation}{0} \setcounter{section}{2}
\section{Periodic Spot Patterns for the Schnakenburg Model}\label{schnak}

We study the linear stability of a steady-state periodic pattern of
localized spots for the Schnakenburg model (\ref{1:sc}) where the
spots are centered at the lattice points of (\ref{lattice-def}). The
analysis below is based on the fundamental Wigner-Seitz cell
$\Omega$, which contains exactly one spot centered at the origin.

\subsection{The Steady-State Solution}\label{schnak:equil}

We use the method of matched asymptotic expansions to construct a
steady-state one-spot solution to (\ref{1:sc}) centered at $\xb=\zero
\in \Omega$. The construction of such a solution consists of an outer
region where $v$ is exponentially small and $u={\mathcal O}(1)$, and
an inner region of extent ${\mathcal O}(\eps)$ centered at the origin
where both $v$ and $u$ have localized.

In the inner region we look for a locally radially symmetric steady-state
solution of the form
\begin{equation}
  u = \frac{1}{\sqrt{D}} \, U \,, \qquad v = \sqrt{D}
 V  \,, \qquad
  \yb = \eps^{-1}\xb \,. \label{eq1:var}
\end{equation}
Then, substituting (\ref{eq1:var}) into the steady-state equations of
(\ref{1:sc}), we obtain that $V\sim V(\rho)$ and $U\sim U(\rho)$, with
$\rho=|\yb|$, satisfy the following {\em core problem} in terms of an
unknown source strength $S\equiv \int_{0}^{\infty} U V^2\rho \, d\rho$
to be determined:
\bsub \label{eq1:core}
\begin{gather}
 \Delta_\rho V - V + U V^2 =0 \,, \qquad
 \Delta_\rho U - U V^2=0 \,, \qquad 0 < \rho
 < \infty \,,\label{eq1:core_1}\\
 U^{\p}(0)=V^{\p}(0)=0\,; \qquad V
 \to 0 \,, \qquad U \sim S \log\rho + \chi(S) + o(1) \,, \quad
 \mbox{as} \quad \rho\to \infty \,. \label{eq1:core_2}
\end{gather}
Here we have defined $\Delta_\rho V\equiv V^{\p\p} + \rho^{-1} V^{\p}$.
\esub

The core problem (\ref{eq1:core}), without the explicit far-field
condition (\ref{eq1:core_2}) was first identified and its solutions
computed numerically in \S 5 of \cite{MO1}. In \cite{KWW_schnak}, the
function $\chi(S)$ was computed numerically, and solutions to the core
problem were shown to closely related to the phenomena of
self-replicating spots.

The unknown source strength $S$ is determined by matching the the
far-field behavior of the core solution to an outer solution for $u$
valid away from ${\mathcal O}(\eps)$ distances from $\xbo$. In the
outer region, $v$ is exponentially small, and from (\ref{eq1:var}) we
get $\eps^{-2} u v^2 \rightarrow 2\pi \sqrt{D} S
\delta(\xb)$. Therefore, from (\ref{1:sc}), the outer steady-state
problem for $u$ is
\begin{equation}
\label{eq1:uout}
\begin{split}
   \Delta u  &= - \frac{a}{D} + \frac{2\pi}{\sqrt{D}} S \,\delta(\xb) \,,
  \quad \xb \in \Omega \,; \qquad
     {\cal P}_0 u=0 \,, \quad \xb\in \partial\Omega\,, \\
   u &\sim \frac{1}{\sqrt{D}} \left[S\log|\xb| + \chi(S) +
  \frac{S}{\nu}\right]\,,  \quad \mbox{as} \quad \xb\to \zero \,,
\end{split}
\end{equation}
where $\nu\equiv {-1/\log\eps}$ and $\Omega$ is the fundamental Wigner-Seitz
cell. The divergence theorem then yields
\begin{equation}
    S = \frac{a|\Omega|}{2\pi\sqrt{D}} \,. \label{eq1:sval}
\end{equation}
The solution to (\ref{eq1:uout}) is then written in terms of the
periodic Green's function $G_{0p}(\xb)$ as
\begin{equation}
   u(x)=-\frac{2\pi}{\sqrt{D}} \left[ S G_{0p}(\xb;\zero) - u_c \right] \,,
  \qquad u_c\equiv \frac{1}{2\pi\nu} \left[S + 2\pi\nu S R_{0p} + \nu
  \chi(S)\right]\,, \label{eq1:u0solve}
\end{equation}
where the periodic source-neutral Green's function $G_{0p}(\xb)$ and
its regular part $R_{0p}$ satisfy
\begin{equation}
\label{gr:source_neut}
\begin{split}
   \Delta G_{0p} &= \frac{1} {|\Omega|} - \delta(\xb) \,, \quad \xb
  \in \Omega \,; \qquad {\cal P}_0 G_{0p}=0 \, \quad \xb\in \partial\Omega\,,
  \\  G_{0p} &\sim -\frac{1}{2\pi}\log|\xb| + R_{0p} + o(1) \,, \quad \mbox{as}
   \quad \xb\to \zero \,; \qquad  \int_{\Omega} G_{0p} \, d\xb = 0 \,.
\end{split}
\end{equation}
An explicit expression for $R_{0p}$ on an oblique Bravais lattice was
derived in Theorem 1 of \cite{chen}.  A periodic pattern of spots is
then obtained through periodic extension to $\R^2$ of the one-spot
solution constructed within $\Omega$.

Since the stability threshold occurs when $D={\mathcal O}(1/\nu)$, for
which $S={\mathcal O}(\nu^{1/2})\ll 1$ from (\ref{eq1:sval}), we must
calculate an asymptotic expansion in powers of $\nu$ for the solution
to the core problem (\ref{eq1:core}). This result, which is required
for the stability analysis in \S \ref{schnak:stab}, is as follows:

\begin{lemma}\label{lemma 3.1} For
$S=S_0 \nu^{1/2} + S_1 \nu^{3/2}+\cdots$, where $\nu\equiv {-1/\log\eps}\ll 1$,
the asymptotic solution to the core problem (\ref{eq1:core}) is
\bsub \label{score:exp}
\begin{equation}
  V\sim \nu^{1/2}\left( V_0 + \nu V_1 + \cdots\right) \,, \qquad
  U\sim \nu^{-1/2}\left( U_0 + \nu U_1 + \nu^2 U_2 + \cdots\right)
  \,, \qquad
  \chi \sim \nu^{-1/2}\left( \chi_0 + \nu \chi_1 + \cdots\right)\,,
  \label{score:exp_1}
\end{equation}
where $U_0$, $U_{1}(\rho)$, $V_{0}(\rho)$, and $V_{1}(\rho)$ are defined by
\begin{equation}
  U_0 = \chi_0 \,, \qquad U_1 = \chi_1 + \frac{1}{\chi_0} U_{1p} \,, \qquad
  V_0=\frac{w}{\chi_0} \,, \qquad V_1 = -\frac{\chi_1}{\chi_0^2} w +
  \frac{1}{\chi_0^3} V_{1p} \,. \label{score:exp_2}
\end{equation}
 Here $w(\rho)$ is the unique ground-state solution to
  $\Delta_\rho w-w+w^2=0$ with $w(0)>0$, $w^{\p}(0)=0$, and $w\to 0$
  as $\rho\to \infty$. In terms of $w(\rho)$, the functions $U_{1p}$ and
 $V_{1p}$ are the unique solutions on $0\leq\rho<\infty$ to
\begin{equation}
\label{score:exp_3}
\begin{split}
   L_0 V_{1p} &= - w^2 U_{1p} \,, \qquad V_{1p}^{\p}(0)=0\,, \quad
   V_{1p}\to 0 \,, \quad \mbox{as} \quad \rho \to \infty \,,\\
  \Delta_\rho U_{1p} & = w^2 \,, \qquad U_{1p}^{\p}(0)=0\,, \quad
  U_{1p}\to b \log\rho + o(1) \,, \quad
  \mbox{as} \quad \rho \to \infty\,; \qquad b\equiv \int_{0}^{\infty}  w^2\rho \,
 d\rho \,,
\end{split}
\end{equation}
 where the linear operator $L_0$ is defined by $L_0 V_{1p}
  \equiv\Delta_\rho V_{1p} - V_{1p} + 2w V_{1p}$.  Finally, in
  (\ref{score:exp_1}), the constants $\chi_0$ and $\chi_1$ are related
  to $S_0$ and $S_1$ by
\begin{equation}
   \chi_0 = \frac{b}{S_0} \,, \qquad \chi_1=-\frac{S_1b}{S_0^2} +
  \frac{S_0}{b^2} \int_{0}^{\infty} V_{1p} \rho \, d\rho \,. \label{score:exp_4}
\end{equation}
\esub
\end{lemma}

The derivation of this result was given in \S 6 of \cite{KWW_schnak}
and is outlined in Appendix \ref{app:schnak} below.  We remark that
the $o(1)$ condition in the far-field behavior of $U_{1p}$ in
(\ref{score:exp_3}) eliminates an otherwise arbitrary constant in the
determination of $U_{1p}$. This condition, therefore, ensures that the
solution to the linear BVP system (\ref{score:exp_3}) is unique.

\subsection{The Spectrum of the Linearization Near the Origin}\label{schnak:stab}

To study the stability of the periodic pattern of spots with respect
to fast ${\mathcal O}(1)$ time-scale instabilities, we use the
Floquet-Bloch theorem that allows us to only consider the Wigner-Seitz
cell $\Omega$, centered at the origin, with quasi-periodic boundary
conditions imposed on its boundaries.

We linearize around the steady-state $u_e$ and $v_e$, as calculated
in \S \ref{schnak:equil}, by introducing the perturbation
\begin{equation}
   u = u_e + \e^{\lam t} \eta \,, \qquad
   v = v_e + \e^{\lam t} \phi \,. \label{st:pert}
\end{equation}
By substituting (\ref{st:pert}) into (\ref{1:sc}) and linearizing, we
obtain the following eigenvalue problem for $\phi$ and $\eta$:
\begin{equation}
\label{st:eig}
\begin{split}
 \eps^{2} \Delta \phi - \phi + 2u_e v_e \phi + v_{e}^2 \eta &=\lam
 \phi \,, \quad \xb \in \Omega\,; \qquad
{\cal P}_\kb \phi=0 \,, \quad \xb \in \partial\Omega \,, \\
  D\Delta\eta- 2 \eps^{-2} u_e v_e
 \phi- \eps^{-2} v_{e}^{2} \eta &= \lambda \tau \eta \,, \quad \xb \in
 \Omega \,; \qquad {\cal P}_\kb \eta=0\,, \quad
 \xb\in \partial\Omega\,,
\end{split}
\end{equation}
where ${\cal P}_{\kb}$ is the quasi-periodic boundary operator of (\ref{gr:pk}).

In the inner region near $\xb=\xbo$ we introduce the local variables
$N(\rho)$ and $\Phi(\rho)$ by
\begin{equation}
  \eta = \frac{1}{D} N(\rho) \,, \qquad
   \phi= \Phi(\rho) \,, \qquad \rho=|\yb| \,, \qquad
 \yb=\eps^{-1}\xb \,. \label{st:inn}
\end{equation}
Upon substituting (\ref{st:inn}) into (\ref{st:eig}), and by using
$u_e\sim {U(\rho)/\sqrt{D}}$ and $v_e\sim \sqrt{D} V(\rho)$, where $U$ and $V$
satisfy the core problem (\ref{eq1:core}), we obtain on $0<\rho<\infty$
that
\begin{equation}
\label{st:eprob}
\begin{split}
  \Delta_\rho \Phi - \Phi + & 2 U V \Phi + N V^2 =\lambda \Phi \,, \qquad
  \Phi \to 0 \,, \quad \mbox{as} \quad \rho \to \infty \,,\\ \Delta_\rho N &= 2
  UV \Phi + N V^2 \,, \qquad N \sim C \log \rho + B \,, \quad
  \mbox{as} \quad \rho\to \infty \,,
\end{split}
\end{equation}
with $\Phi^{\prime}(0)=N^{\prime}(0)=0$, and where $B=B(S;\lambda)$.  We
remark that for $\mbox{Re}(\lambda+1)>0$, $\Phi$ in (\ref{st:eprob})
decays exponentially as $\rho\to \infty$. However, in contrast, we
cannot apriori impose that $N$ in (\ref{st:eprob}) is bounded as
$\rho\to \infty$. Instead we must allow for the possibility of a
logarithmic growth for $N$ as $\rho\to \infty$. Upon using the divergence
theorem we identify $C$ as
$C=\int_{0}^{\infty} \left(2 U V \Phi+ N V^2\right)\rho\ d\rho$.
The constant $C$ will be determined by matching $N$ to an outer
eigenfunction $\eta$, valid away from $\xb=\xbo$, that satisfies
(\ref{st:eig}).

To formulate this outer problem, we obtain since $v_e$ is localized
near $\xb=\xbo$ that, in the sense of distributions,
\begin{equation}
  \eps^{-2}\left(2 u_e v_e \phi + \eta v_{e}^2 \right) \rightarrow
 \left(\int_{\R^2} \left( 2 U V \Phi + N V^2 \right)\, d\yb \right) \, \
  \delta(\xb) = 2\pi C  \delta(\xb) \,.\label{st:cval_1}
\end{equation}
By using this expression in (\ref{st:eig}), we conclude that the
outer problem for $\eta$ is
\begin{equation}
  \label{st:etaout}
\begin{split}
   \Delta \eta & -\frac{\tau\lambda}{D} \eta = \frac{2\pi C}{D} \delta(\xb) \,,
 \quad \xb\in \Omega \,; \qquad {\cal P}_\kb \eta=0 \,, \quad \xb \in
 \partial\Omega \,, \\
   \eta & \sim \frac{1}{D}\left(C \log|\xb| + \frac{C}{\nu} + B\right)\,,
  \quad \mbox{as}   \quad \xb\to\xbo\,.
\end{split}
\end{equation}
The solution to (\ref{st:etaout}) is $\eta=-2\pi C D^{-1} G_{b\lam}(\xb)$, where
$G_{b\lam}$ satisfies
\begin{equation}
\label{grlam:eta}
\begin{split}
   \Delta G_{b\lam} -\frac{\tau\lam}{D} G_{b\lam} &=-\delta(\xb) \,,
  \quad \xb \in \Omega \,; \qquad
  {\cal P}_\kb G_{b\lam}=0 \,, \quad \xb\in \partial\Omega\,, \\
   G_{b\lam} &\sim -\frac{1}{2\pi}\log|\xb| + R_{b\lam} \,,
  \quad \mbox{as} \quad \xb\to \zero \,.
\end{split}
\end{equation}
From the requirement that the behavior of $\eta$ as $\xb\to \xbo$
agree with that in (\ref{st:etaout}), we conclude that
$B+{C/\nu}=-2\pi C R_{b\lam}$. Finally, since the stability threshold
occurs in the regime $D={\mathcal O}(\nu^{-1})\gg 1$, we conclude from
Lemma \ref{lemma 2.3} (ii) that for $|\kb|\neq 0$ and ${\kb/(2\pi)}\in
\Omega_B$,
\begin{equation}
  \left(1 + 2\pi \nu R_{b0} + {\mathcal O}(\nu^2) \right) C = - \nu B
  \,, \label{st:cb}
\end{equation}
where $R_{b0}$ is the regular part of the Bloch Green's function
$G_{b0}$ defined by (\ref{green:b0}) on $\Omega$.

We now proceed to determine the portion of the continuous spectrum of
the linearization that lies within an ${\mathcal O}(\nu)$ neighborhood
of the origin, i.~e.~that satisfies $|\lambda|\leq {\mathcal O}(\nu)$,
when $D$ is close to a certain critical value. To do so, we first must
calculate an asymptotic expansion for the solution to (\ref{st:eprob})
together with (\ref{st:cb}).

By using (\ref{score:exp_1}) we first calculate the coefficients in the
differential operator in (\ref{st:eprob}) as
\begin{align*}
  UV &= w + \nu \left(U_0 V_1 + U_1 V_0\right) + \cdots = w +
  \frac{\nu}{\chi_0^2} \left[ V_{1p} + w U_{1p} \right] + \cdots \,,
  \\ V^2 &= \nu\left(V_0^2 + 2\nu V_0 V_1\right) + \cdots = \nu
  \frac{w^2}{\chi_0^2} + \frac{2\nu^2}{\chi_0^3} \left(-\chi_1 w^2 +
  \frac{wV_{1p}}{\chi_0} \right) + \cdots \,,
\end{align*}
so that the local problem (\ref{st:eprob}) on $0<\rho<\infty$ becomes
\begin{equation}\label{st:eprob_1}
\begin{split}
  \Delta_\rho \Phi & - \Phi + \left[ 2 w + \frac{2\nu}{\chi_0^2} \left(V_{1p}+
w U_{1p} \right) + \cdots \right] \Phi =-\nu \left[ \frac{w^2}{\chi_0^2} +
  \frac{2\nu}{\chi_0^3} \left(-\chi_1 w^2 + \frac{wV_{1p}}{\chi_0}\right)+\cdots
 \right] N + \lambda \Phi \,,  \\
  \Delta_\rho N &  = \left[ 2 w + \frac{2\nu}{\chi_0^2} \left(V_{1p}+
w U_{1p} \right) + \cdots \right] \Phi + \nu \left[ \frac{w^2}{\chi_0^2} +
  \frac{2\nu}{\chi_0^3} \left(-\chi_1 w^2 + \frac{wV_{1p}}{\chi_0}\right)
 + \cdots \right] N \,, \\
& \Phi\to 0 \,, \qquad   N \sim C \log \rho + B \,, \quad \mbox{as} \quad
   \rho\to \infty \,; \qquad \Phi^{\prime}(0)=N^{\prime}(0)=0 \,.
\end{split}
\end{equation}
We then introduce the appropriate expansions
\begin{equation} \label{st:exp}
\begin{split}
  N &= \frac{1}{\nu} \left(\hat{N}_0+\nu \hat{N}_1 + \cdots\right)\,,
 \qquad B = \frac{1}{\nu} \left(\hat{B}_0+\nu \hat{B}_1 + \cdots\right)\,,
  \qquad C=C_0 + \nu C_1 + \cdots \,, \qquad \\
  \Phi &=\Phi_0 + \nu \Phi_1 +  \cdots \,, \qquad \lambda=\lambda_0 + \nu
  \lambda_1 + \cdots \,,
\end{split}
\end{equation}
into (\ref{st:eprob_1}) and collect powers of $\nu$.

To leading order, we obtain on $0<\rho<\infty$ that
\begin{equation}
\label{st:zero}
\begin{split}
  L_0 \Phi_0 &\equiv  \Delta_\rho \Phi_0 - \Phi_0 + 2w \Phi_0 =-
 \frac{w^2}{\chi_0^2} \hat{N}_0 + \lambda_0 \Phi_0 \,, \qquad
  \Delta_\rho \hat{N}_0=0 \,, \\
 \Phi_0 &\to 0 \,, \qquad  \hat{N}_0\to \hat{B}_0 \quad \mbox{as} \quad
  \rho\to\infty\,; \qquad \Phi_0^{\prime}(0)=0 \,, \quad
   \hat{N}_0^{\prime}(0)=0  \,,
\end{split}
\end{equation}
where $L_0$ is referred to as the local operator. We conclude that
$\hat{N}_0=\hat{B}_0$ for $\rho\geq 0$.

At next order, we obtain on $\rho>0$ that $\Phi_1$ satisfies
\begin{equation}
  L_0\Phi_1 + \frac{2}{\chi_0^2}\left(V_{1p}+wU_{1p}\right) \Phi_0
  + \frac{2}{\chi_0^3}\left(-\chi_1 w^2 + \frac{wV_{1p}}{\chi_0}\right)
 \hat{N}_0 = -\frac{w^2}{\chi_0^2}\hat{N}_1 + \lambda_1 \Phi_0 \,;\qquad
 \Phi_1\to 0 \,, \quad
\mbox{as} \quad \rho\to \infty \,, \label{st:phi1}
\end{equation}
with $\Phi_1^{\prime}(0)=0$, and that $\hat{N}_1$ on $\rho>0$ satisfies
\begin{equation}
   \Delta_\rho \hat{N}_1 = 2w\Phi_0 + \frac{w^2}{\chi_0^2}\hat{N}_0
   \,; \qquad \hat{N}_1 \sim C_0 \log\rho + \hat{B}_1 \,, \quad
   \mbox{as} \quad \rho\to \infty \,; \qquad \hat{N}_1^{\prime}(0)=0 \,.
 \label{st:n1}
\end{equation}
In our analysis we will also need the problem for $\hat{N}_2$ given by
\begin{equation}
\label{st:n2}
\begin{split}
  \Delta_\rho \hat{N}_2 &= 2w\Phi_1 + \frac{2}{\chi_0^2}
  \left(V_{1p}+wU_{1p}\right) \Phi_0
  + \frac{2}{\chi_0^3}\left(-\chi_1 w^2 + \frac{wV_{1p}}{\chi_0}\right)
 \hat{N}_0 +\frac{w^2}{\chi_0^2}\hat{N}_1 \,, \\
   \hat{N}_2 &\sim C_1\log\rho + \hat{B}_2 \,, \quad \mbox{as} \quad
   \rho\to \infty\,; \qquad \hat{N}_2^{\prime}(0)=0 \,.
\end{split}
\end{equation}
In addition, by substituting (\ref{st:exp}) into (\ref{st:cb}), we obtain
upon collecting powers of $\nu$ that
\begin{equation}
  C_0 = -\hat{B}_0 \,, \qquad C_1+2\pi R_{b0} C_0 = -\hat{B}_1 \,.
  \label{st:cbexp}
\end{equation}

Next, we proceed to analyze (\ref{st:zero})--(\ref{st:n2}).  From the
divergence theorem, we obtain from (\ref{st:n1}) that
\begin{equation}
   C_0 = \int_{0}^{\infty} 2w\Phi_0 \rho \, d\rho + \frac{b}{\chi_0^2} \hat{N}_0
 \,, \qquad b\equiv \int_{0}^{\infty} w^2 \rho \, d\rho \,.
  \label{st:c0b0}
\end{equation}
Since $C_0=-\hat{B}_0$ and $\hat{B}_0=\hat{N}_0$, (\ref{st:c0b0}) yields
that
\begin{equation}
  \hat{N}_0= \hat{B}_0 = -2 \left[ 1+ \frac{b}{\chi_0^2}\right]^{-1}
  \int_{0}^{\infty} w \Phi_0 \rho \, d\rho \,. \label{st:b0}
\end{equation}
With $\hat{N}_0$ known, (\ref{st:zero}) provides the leading-order
nonlocal eigenvalue problem (NLEP)
\begin{equation}
   L_0 \Phi_0 - \frac{2 w^2 b}{\chi_0^2 + b}  \frac{ \int_{0}^{\infty}
   w\Phi_0 \, \rho \, d\rho}{\int_{0}^{\infty} w^2 \rho \, d\rho} = \lambda_0
  \Phi_0\,; \qquad \Phi_0\to 0 \,, \quad \mbox{as} \quad \rho\to \infty \,;
 \qquad \Phi_0^{\prime}(0)=0 \,. \label{st:phi0}
\end{equation}

For this NLEP, the rigorous result of \cite{w} (see also Theorem 3.7
of the survey article \cite{wsurv}) proves that $\mbox{Re}(\lambda_0)<0$
if and only if ${2b/(\chi_0^2+b)}>1$. At the stability threshold where
${2b/(\chi_0^2+b)}=1$, we have from the identity $L_0 w= w^2$ that
$\Phi_0=w$ and $\lambda_0=0$. From (\ref{st:b0}) and (\ref{st:c0b0}) we
can then calculate $\hat{B}_0$ and $C_0$ at this leading-order stability
threshold. In summary, to leading order in $\nu$, we obtain at
$\lambda_0=0$ that
\begin{equation}
   \frac{b}{\chi_0^2}=1 \,, \qquad \Phi_0 = w \,, \qquad
    \hat{B}_0=\hat{N}_0=-b=-\int_0^{\infty} w^2\rho \, d\rho \,, \qquad
   C_0=b \,. \label{st:marg}
\end{equation}

Upon substituting (\ref{st:marg}) into (\ref{st:n1})
we obtain at $\lambda_0=0$ that $\hat{N}_1$ on $\rho>0$ satisfies
\begin{equation}
   \Delta_\rho \hat{N}_1 = w^2 \,; \qquad
   \hat{N}_1 \sim b \log\rho + \hat{B}_1 \,, \quad \mbox{as} \quad \rho\to
\infty \,; \qquad \hat{N}_1^{\prime}(0)=0 \,.\label{st:n1n}
\end{equation}
Upon comparing (\ref{st:n1n}) with the problem for $U_{1p}$, as given in
(\ref{score:exp_3}), we conclude that
\begin{equation}
   \hat{N}_1= U_{1p} + \hat{B}_1 \,. \label{st:n1_solv}
\end{equation}

Next, we observe that for $D={D_0/\nu}\gg 1$, it follows from
(\ref{eq1:sval}) that $S=\nu^{1/2}S_0 + \cdots$, where
$S_0={a|\Omega|/(2\pi\sqrt{D_0})}$. Then, since $S_0={b/\chi_0}$ from
(\ref{score:exp_4}), and ${b/\chi_0^2}=1$ when $\lambda_0=0$ from
(\ref{st:marg}), the critical value of $D_0$ at the leading-order
stability threshold $\lambda_0=0$ is
\bsub\label{st:crit_lead}
\begin{equation}
   D_0 = D_{0c}\equiv \frac{a^2 |\Omega|^2}{4\pi^2 b} \,. \label{st:d0crit}
\end{equation}
This motivates the definition of the bifurcation parameter $\mu$ by
\begin{equation}
   \mu \equiv \frac{4\pi^2 D\nu b}{a^2 |\Omega|^2} \,, \label{st:mudef}
\end{equation}
\esub
so that at criticality where $\chi_0=\sqrt{b}$, we have $\mu=1$.

We then proceed to analyze the effect of the higher order terms in
powers of $\nu$ on the stability threshold. In particular, we
determine the continuous band of spectrum that is contained within an
${\mathcal O}(\nu)$ ball near $\lambda=0$ when the bifurcation
parameter $\mu$ is ${\mathcal O}(\nu)$ close to its leading-order
critical value $\mu=1$. As such, we set
\begin{equation}
   \lambda = \nu \lambda_1 + \cdots \,, \qquad \mbox{for} \qquad
  \mu=1+\nu \mu_1 + \cdots \,,
\end{equation}
and we derive an expression for $\lambda_1$ in terms of $\mu_1$, the Bloch
vector $\kb$, the lattice structure, and certain correction terms to the
core problem.

To determine an expression for $\mu_1$ in terms of $\chi_0$ and
$\chi_1$ we first set $D={D_0/\nu}$ and write the two term expansion
for the source strength $S$ as
\begin{equation*}
 S = \frac{a |\Omega|}{2\pi\sqrt{D}}=\nu^{1/2}\left(S_0+\nu S_1+\cdots\right)\,,
\end{equation*}
where $S_0$ and $S_1$ are given in (\ref{score:exp_4}) in terms of
$\chi_0$ and $\chi_1$. By using (\ref{score:exp_4}) and
(\ref{st:mudef}), this expansion for $S$ becomes
\begin{equation}
 \sqrt{ \frac{b}{\mu}} = \left( \frac{b}{\chi_0} + \nu
  \left[ -\frac{\chi_1 b}{\chi_0} + \frac{1}{\chi_0^3} \int_{0}^{\infty}
  V_{1p} \rho \, d\rho\right] + \cdots \right) \,. \label{st:muexp}
\end{equation}
As expected, to leading order we have $\mu=1$ when $b=\chi_0^2$. At
$\lambda_0=0$ where $\chi_0=\sqrt{b}$, we use $\mu^{-1/2}\sim 1-{\nu
  \mu_1/2}+\cdots$ to relate $\mu_1$ to $\chi_1$ as
\begin{equation}
   \frac{\chi_1}{\sqrt{b}} = \frac{\mu_1}{2} + \frac{1}{b^2} \int_0^{\infty}
  V_{1p}\rho \, d\rho \,. \label{st:mu1}
\end{equation}

Next, we substitute $\Phi_0=w$, $\hat{N}_0=-b$, $\chi_0^2=b$, and
$\hat{N}_1=U_{1p}+\hat{B}_1$ from (\ref{st:n1_solv}), into the equation
(\ref{st:phi1}) for $\Phi_1$. After some algebra, we conclude that
$\Phi_1$ at $\lambda_0=0$ satisfies
\begin{equation}
  L_0 \Phi_1 + \frac{w^2}{b} \hat{B}_1 = - \frac{2\chi_1\chi_0}{b} w^2
  - \frac{3}{b} w^2 U_{1p} + \lambda_1 w \,; \qquad \Phi_1\to 0
\,, \quad \mbox{as} \quad \rho\to \infty \,, \label{st:phi1n}
\end{equation}
with $\Phi_1^{\prime}(0)=0$.  In a similar way, at the leading-order
stability threshold, the problem (\ref{st:n2}) for $\hat{N}_2$ on
$\rho>0$ becomes
\begin{equation}
\label{st:n2n}
\begin{split}
  \Delta_\rho \hat{N}_2 &= 2w\Phi_1 + \frac{w^2}{b} \hat{B}_1
  + \frac{3}{b} w^2 U_{1p} + \frac{2\chi_0\chi_1}{b} w^2 \,, \\
   \hat{N}_2 &\sim C_1\log\rho + \hat{B}_2 \,, \quad \mbox{as} \quad
   \rho\to \infty\,; \qquad \hat{N}_2^{\prime}(0)=0 \,.
\end{split}
\end{equation}

To determine $\hat{B}_1$, as required in (\ref{st:phi1n}), we use the
divergence theorem on (\ref{st:n2n}) to obtain that
\begin{equation*}
  C_1 = 2\int_{0}^{\infty} w\Phi_1 \rho \, d\rho + \hat{B}_1 + \frac{3}{b}
  \int_{0}^{\infty} w^2 U_{1p} \rho \, d\rho + 2\chi_0 \chi_1 \,.
\end{equation*}
Upon combining this expression with $C_1+2\pi R_{b0} C_0 =
-\hat{B}_1$, as obtained from (\ref{st:cbexp}), where $C_0=b$,
we obtain $\hat{B}_1$ as
\begin{equation*}
  \hat{B}_1=-\int_0^{\infty} w\Phi_1 \rho \, d\rho - \pi b R_{b0} -
    \frac{3}{2b} \int_{0}^{\infty} w^2 U_{1p} \rho \, d\rho - \chi_0\chi_1 \,.
\end{equation*}
Upon substituting this expression into (\ref{st:phi1n}), we conclude that
 $\Phi_1$ satisfies
\bsub \label{st:phi1_fin}
\begin{equation}
 {\cal L}\Phi_1 \equiv L_0 \Phi_1 - w^2 \frac{ \int_0^\infty w\Phi_1
   \rho \, d\rho}{ \int_0^\infty w^2\rho \, d\rho} = {\cal R}_s +
 \lambda_1 w \,; \qquad \Phi_1\to 0 \,, \quad \mbox{as} \quad
  \rho \to \infty \,, \label{st:phi1_fin_1}
\end{equation}
with $\Phi_1^{\prime}(0)=0$, where the residual ${\cal R}_s$ is defined by
\begin{equation}
  {\cal R}_s\equiv \pi w^2
   R_{b0} + \frac{3}{2b^2} w^2 \int_0^{\infty} w^2 U_{1p} \rho\, d\rho
   - \frac{\chi_0\chi_1 w^2}{b} -\frac{3}{b} w^2U_{1p} \,.
\label{st:phi1_fin_2}
\end{equation}
\esub

Finally, $\lambda_1$ is determined by imposing a solvability condition
on (\ref{st:phi1_fin}). The homogeneous adjoint operator ${\cal L}^{\star}$
corresponding to (\ref{st:phi1_fin}) is
\begin{equation}
\label{st:phi1_adj}
 {\cal L}^{\star}\Psi \equiv L_0 \Psi - w \frac{ \int_0^\infty w^2\Psi
   \rho \, d\rho}{ \int_0^\infty w^2\rho \, d\rho} \,.
\end{equation}
We define $\Psi^{\star}=w+{\rho w^{\prime}/2}$ and readily verify that
$L_0 \Psi^{\star}=w$ and $L_0 w = w^2$ (see \cite{w}). Then, we use
Green's second identity to obtain $\int_{0}^{\infty} \left[ w
  L_0\Psi^{\star}- \Psi^{\star}L_0 w\right] \rho\,d\rho=
\int_0^{\infty} \left( w^2 - \Psi^{\star} w^2\right)\rho\, d\rho$. By
the decay of $w$ and $\Psi^{\star}$ as $\rho\to\infty$, we obtain that
$\int_{0}^{\infty} w^2 \rho \, d\rho = \int_{0}^{\infty} \Psi^{\star}
w^2\rho \, d\rho$. Therefore, since the ratio of the two integrals in
(\ref{st:phi1_adj}) is unity when $\Psi=\Psi^{\star}$, we conclude
that ${\cal L}^{\star}\Psi^{\star}=0$.

Finally, we impose the solvability condition that the right hand side of
(\ref{st:phi1_fin}) is orthogonal to $\Psi^{\star}$ in the sense that
$\lambda_1 \int_{0}^{\infty} w \Psi^{\star}\rho \, d\rho +
\int_{0}^{\infty} {\cal R}_s \Psi^{*} \rho \, d\rho =0$. By using
(\ref{st:phi1_fin_2}) for ${\cal R}_s$, this solvability condition yields
that
\begin{equation}
  \lambda_1 = -\frac{ \int_{0}^{\infty} w^2 \Psi^{\star} \rho \, d\rho}
  { b \int_{0}^{\infty} w \Psi^{\star} \rho \, d\rho}\left( b \pi R_{b0}
   - \chi_1 \chi_0 + \frac{3}{2b} \int_{0}^{\infty} w^2 U_{1p} \rho \, d\rho
  -3 \frac{ \int_{0}^{\infty} w^2 U_{1p} \Psi^{\star} \rho \, d\rho}
  { \int_{0}^{\infty} w^2 \Psi^{\star} \rho \, d\rho} \right) \,.
\label{st:lam1_1}
\end{equation}
Equation (\ref{st:lam1_1}) is simplified by first calculating the following
integrals by using integration by parts:
\begin{equation}
\label{st:integ}
\begin{split}
   \int_{0}^{\infty} w^2 \Psi^{\star} \rho \, d\rho &=
 \int_0^{\infty} \left(L_0 w\right) \left(L_0^{-1} w\right) = \int_{0}^{\infty}
  w^2 \rho \, d\rho = b \,, \\
   \int_{0}^{\infty} w \Psi^{\star} \rho \, d\rho &=\int_{0}^{\infty} \rho w
  \left( w + \frac{\rho}{2} w^{\prime}\right) \, d\rho =
    \int_{0}^{\infty} w^2\rho \, d\rho + \frac{1}{4} \int_{0}^{\infty}
  \left[w^2\right]^{\prime} \rho^2 \, d\rho = \frac{b}{2} \,.
\end{split}
\end{equation}
In addition, since $L_0 V_{1p}=-w^2 U_{1p}$ from (\ref{score:exp_3})
and $\Psi^{\star}=L_0^{-1}w$, we obtain upon integrating by parts that
\begin{equation*}
  \int_0^{\infty} w^2 U_{1p} \Psi^{\star} \rho \, d\rho = -
 \int_0^{\infty} \left( L_0 V_{1p}\right)\left(L_0^{-1} w\right)\rho \, d\rho=
 - \int_{0}^{\infty} V_{1p} w \rho \, d\rho \,.
\end{equation*}
By substituting this expression and (\ref{st:integ}) into
(\ref{st:lam1_1}), we obtain
\begin{equation}
  \frac{\lambda_1}{2} = -\frac{1}{b} \left[ b\pi R_{b0} -\chi_0\chi_1 +
  \frac{3}{2b} \int_{0}^{\infty} w^2 U_{1p} \rho \, d\rho +
  \frac{3}{b} \int_0^{\infty} w V_{1p} \rho \, d\rho \right] \,.
\label{st:lam_2}
\end{equation}
Next, we use (\ref{score:exp_3}) to calculate $\int_0^{\infty} w^2
U_{1p}\rho \, d\rho = \int_0^\infty \left( V_{1p}-2w V_{1p}\right) \rho
\, d\rho$. Finally, we substitute this expression together with
$\chi_0=\sqrt{b}$ and (\ref{st:mu1}), which relates $\mu_1$ to
$\chi_1$, into (\ref{st:lam_2}) to obtain our final expression for
$\lam_1$.  We summarize our result as follows:

\begin{result}\label{pr 3.2} In the limit
$\eps\to 0$, consider a steady-state periodic pattern of spots for the
  Schnakenburg model (\ref{1:sc}) on the Bravais lattice $\Lambda$
  when $D={\mathcal O}(\nu^{-1})$, where $\nu={-1/\log\eps}$.  Then,
  when
\bsub \label{st:lam1_fin}
\begin{equation}
  D= \frac{a^2|\Omega|^2}{4\pi^2 b\nu} \left(1+\mu_1 \nu\right)\,,
  \label{st:lam1_fin1}
\end{equation}
where $\mu_1={\mathcal O}(1)$, the portion of the continuous spectrum of
the linearization that lies within an ${\mathcal O}(\nu)$ neighborhood
of the origin $\lambda=0$, i.~e.~that satisfies
$|\lambda|\leq {\mathcal O}(\nu)$, is given by
\begin{equation}
  \lambda=\nu \lambda_1 + \cdots \,, \qquad
   \lambda_1= 2 \left[ \frac{\mu_1}{2} - \pi R_{b0} - \frac{1}{2b^2}
  \int_{0}^{\infty} \rho V_{1p} \, d\rho \right] \,. \label{st:lam1_fin2}
\end{equation}
\esub
Here $|\Omega|$ is the area of the Wigner-Seitz cell and
$R_{b0}=R_{b0}(\kb)$ is the regular part of the Bloch Green's function
$G_{b0}$ defined on $\Omega$ by (\ref{green:b0}), with $\kb\neq 0$ and
${\kb/(2\pi)}\in \Omega_B$.
\end{result}

The result (\ref{st:lam1_fin2}) determines how the portion of the band
of continuous spectrum that lies near the origin depends on the
de-tuning parameter $\mu_1$, the correction $V_{1p}$ to the solution of
the core problem, and the lattice structure and Bloch wavevector
$\kb$ as inherited from $R_{b0}(\kb)$.

\begin{remark}\label{remark 1} We need only consider
  ${\kb/(2\pi)}$ in the first Brillouin zone $\Omega_B$, defined as
  the Wigner-Seitz cell centered at the origin for the reciprocal
  lattice. Since $R_{b0}$ is real-valued from Lemma \ref{lemma 2.1},
  it follows that the band of spectrum (\ref{st:lam1_fin2}) lies on
  the real axis in the $\lambda$-plane. Furthermore, since by Lemma
  \ref{lemma 2.2}, $R_{b0}={\mathcal O}\left({1/(\kb^T {\cal Q}
    \kb)}\right)\to +\infty$ as $|\kb|\to 0$ for some positive
  definite matrix ${\cal Q}$, the continuous band of spectrum that
  corresponds to small values of $|\kb|$ is not within an ${\mathcal
    O}(\nu)$ neighborhood of $\lambda=0$, but instead lies at an
  ${\mathcal O}\left({\nu/\kb^T {\cal Q} \kb}\right)$ distance from
  the origin along the negative real axis in the $\lambda$-plane.
\end{remark}

We conclude from (\ref{st:lam1_fin2}) that a periodic arrangement of
spots with a given lattice structure is linearly stable when
\begin{equation}
   \mu_1 < 2\pi R_{b0}^{\star} + \frac{1}{b^2} \int_0^\infty V_{1p}\rho \, d\rho
 \,, \qquad R_{b0}^{\star} \equiv \min_{\kb} R_{b0}(\kb) \,.
  \label{st:mu1_opt}
\end{equation}
For a fixed area $|\Omega|$ of the Wigner-Seitz cell, the optimal
lattice geometry is defined as the one that allows for stability for
the largest inhibitor diffusivity $D$. This leads to one of our main results.

\begin{result} \label{pr 3.3} The optimal
  lattice arrangement for a periodic pattern of spots for the
  Schnakenburg model (\ref{1:sc}) is the one for which ${\cal
    K}_{\textrm{s}}\equiv R_{b0}^{*}$ is maximized. Consequently, this
  optimal lattice allows for stability for the largest possible value
  of $D$. For $\nu={-1/\log\eps}\ll 1$, a two-term asymptotic
  expansion for this maximal stability threshold for $D$ is given
  explicitly by in terms of an objective function ${\cal K}_s$ by
\begin{equation}
    D_{\textrm{optim}}\sim \frac{a^2 |\Omega|^2}{4\pi^2 b \nu} \left[
      1 + \nu \left(2\pi \max_{\Lambda} {\cal K}_s + \frac{1}{b^2}
      \int_{0}^{\infty} V_{1p}\rho \, d\rho \right)\right]\,, \qquad
    {\cal K}_{\textrm{s}}\equiv R_{b0}^{\star}=\min_{\kb}
    R_{b0} \,, \label{st:doptim}
\end{equation}
 where $\max_{\Lambda}{\cal K}_s$ is taken over all lattices
  $\Lambda$ that have a common area $|\Omega|$ of the Wigner-Seitz
  cell. In (\ref{st:doptim}), $V_{1p}$ is the solution to
  (\ref{score:exp_3}) and $b=\int_{0}^{\infty} w^2\rho \, d\rho$ where
  $w(\rho)>0$ is the ground-state solution of $\Delta_\rho w - w +
  w^2=0$. Numerical computations yield $b\approx 4.93$ and
  $\int_{0}^{\infty} V_{1p}\rho \, d\rho \approx 0.481$.
\end{result}

The numerical method to compute ${\cal K}_{\textrm{s}}$ is given in \S
\ref{sec:ewald}. In \S \ref{sec:ewald_0}, we show numerically that
within the class of oblique Bravais lattices, ${\cal K}_{\textrm{s}}$
is maximized for a regular hexagonal lattice.  Thus, the maximal stability
threshold for $D$ is obtained for a regular hexagonal lattice arrangement of
spots.

\setcounter{equation}{0}
\setcounter{section}{3}
\section{Periodic Spot Patterns for the Gierer-Meinhardt Model}\label{gm}

In this section we analyze the linear stability of a steady-state periodic
pattern of spots for the GM model (\ref{1:gm}), where the spots are
centered at the lattice points of the Bravais lattice
(\ref{lattice-def}).

\subsection{The Steady-State Solution}\label{gm:equil}

We first use the method of matched asymptotic expansions to construct
a steady-state one-spot solution to (\ref{1:gm}) centered at the
origin of the Wigner-Seitz cell $\Omega$.

In the inner region near the origin of $\Omega$ we look for a locally
radially symmetric steady-state solution of the form
\begin{equation}
  u = D \, U \,, \qquad v = D V  \,, \qquad
  \yb = \eps^{-1}\xb \,. \label{geq1:var}
\end{equation}
Then, substituting (\ref{geq1:var}) into the steady-state equations of
(\ref{1:gm}), we obtain that $V\sim V(\rho)$ and $U\sim U(\rho)$, with
$\rho=|\yb|$, satisfy the core problem
\bsub \label{geq1:core}
\begin{gather}
 \Delta_\rho V - V + {V^2/U} =0 \,, \qquad
 \Delta_\rho U =-V^2 \,, \qquad 0 < \rho
 < \infty \,,\label{geq1:core_1}\\
 U^{\p}(0)=V^{\p}(0)=0\,; \qquad V
 \to 0 \,, \qquad U \sim -S \log\rho + \chi(S) + o(1) \,, \quad
 \mbox{as} \quad \rho\to \infty \,, \label{geq1:core_2}
\end{gather}
\esub where $\Delta_\rho V\equiv V^{\p\p} + \rho^{-1} V^{\p}$ and
$S=\int_{0}^{\infty} V^2 \rho \, d\rho$.  The unknown source strength
$S$ will be determined by matching the far-field behavior of the core
solution to an outer solution for $u$ valid away from ${\mathcal
  O}(\eps)$ distances from the origin.

Since $v$ is exponentially small in the outer region, we have in the
sense of distributions that $\eps^{-2} v^2 \rightarrow 2\pi D^2 S
\delta(\xb)$. Therefore, from (\ref{1:gm}), the outer steady-state
problem for $u$ is
\begin{equation}
\label{geq1:uout}
\begin{split}
   \Delta u -\frac{1}{D} u &=-2\pi D S \,\delta(\xb) \,,
  \quad \xb \in \Omega \,; \qquad
 {\cal P}_0 u =0 \,, \quad \xb\in \partial\Omega \,, \\
   u &\sim -D S \log|\xb| + D \left(-\frac{S}{\nu} + \chi(S)\right)\,,
  \quad \mbox{as} \quad \xb\to \zero \,,
\end{split}
\end{equation}
where $\nu\equiv {-1/\log\eps}$.  We introduce the reduced-wave
Green's function $G_{p}(\xb)$ and its regular part $R_p$, which satisfy
\begin{equation}
\label{greq1:uout}
\begin{split}
   \Delta G_p -\frac{1}{D} G_p &=-\delta(\xb) \,,
  \quad \xb \in \Omega \,; \qquad
  {\cal P}_0 G_p = 0\,, \quad \xb\in \partial\Omega\,, \\
   G_p(\xb) &\sim -\frac{1}{2\pi}\log|\xb| + R_p \,,
  \quad \mbox{as} \quad \xb\to \zero \,,
\end{split}
\end{equation}
where $R_p$ is the regular part of $G_p$.  The solution to
(\ref{geq1:uout}) is $u(\xb)=2\pi D S G_p(\xb)$. Now as $\xb\to \xbo$
we calculate the local behavior of $u(\xb)$ and compare it with the
required behavior in (\ref{geq1:uout}). This yields
that $S$ satisfies
\begin{equation}
  \left( 1 + 2\pi \nu R_p \right) S = \nu \chi(S) \,. \label{gm:seq}
\end{equation}

Since the stability threshold occurs when $D={\mathcal O}(\nu^{-1})\gg
1$, we expand the solution to (\ref{greq1:uout}) for $D={D_0/\nu}\gg 1$
with $D_0={\mathcal O}(1)$ to obtain
\begin{equation}
    G_{p}= \frac{D_0}{|\Omega|\nu} + G_{0p} + {\mathcal O}(\nu)\,, \qquad
    R_{p}= \frac{D_0}{|\Omega|\nu} + R_{0p} + {\mathcal O}(\nu)\,,
  \label{gm:green_exp}
\end{equation}
where $G_{0p}$ and $R_{0p}$ is the periodic source-neutral Green's
function and its regular part, respectively, defined by
(\ref{gr:source_neut}). By combining (\ref{gm:seq}) and
(\ref{gm:green_exp}), we get that $S$ satisfies
\begin{equation}
  \left( 1 + \mu + 2\pi \nu R_{0p} + {\mathcal O}(\nu^2)\right) S =
  \nu \chi(S) \,, \qquad \mu \equiv \frac{2\pi D_0}{|\Omega|}
  \,. \label{gm:seq_a}
\end{equation}

To determine the appropriate scaling for $S$ in terms of $\nu\ll 1$
for a solution to (\ref{gm:seq_a}), we use $\chi(S)={\mathcal O}(S^{1/2})$ as
$S\to 0$ from Appendix \ref{app:gm}. Thus, to balance the
leading order terms in (\ref{gm:seq_a}), we require that
$S={\mathcal O}(\nu^2)$ as $\nu\to 0$. The next result
determines a two-term expansion for the solution to the core problem
(\ref{geq1:core}) for $\nu\to 0$ when $S={\mathcal O}(\nu^2)$.

\begin{lemma}\label{lemma 4.1} For
$S=S_0 \nu^{2} + S_1 \nu^{3}+\cdots$, where $\nu\equiv {-1/\log\eps}\ll 1$,
the asymptotic solution to the core problem (\ref{geq1:core}) is
\bsub \label{gcore:exp}
\begin{equation}
  V\sim \nu\left( V_0 + \nu V_1 + \cdots\right) \,, \qquad
  U\sim \nu\left( U_0 + \nu U_1 + \nu^2 U_2 + \cdots\right)
  \,, \qquad
  \chi \sim \nu\left( \chi_0 + \nu \chi_1 + \cdots\right)\,,
  \label{gcore:exp_1}
\end{equation}
where $U_0$, $U_{1}(\rho)$, $V_{0}(\rho)$, and $V_{1}(\rho)$ are defined by
\begin{equation}
  U_0 = \chi_0 \,, \qquad U_1 = \chi_1 +  S_0 U_{1p} \,, \qquad
  V_0=\chi_0 w \,, \qquad V_1 = \chi_1 w + S_0 V_{1p} \,. \label{gcore:exp_2}
\end{equation}
 Here $w(\rho)$ is the unique ground-state solution to
  $\Delta_\rho w-w+w^2=0$ with $w(0)>0$, $w^{\p}(0)=0$, and $w\to 0$
  as $\rho\to \infty$. In terms of $w(\rho)$, the functions $U_{1p}$ and
 $V_{1p}$ are the unique solutions on $0\leq\rho<\infty$ to
\begin{equation}
\label{gcore:exp_3}
\begin{split}
   L_0 V_{1p} &=  w^2 U_{1p} \,, \qquad V_{1p}^{\p}(0)=0\,, \quad
   V_{1p}\to 0 \,, \quad \mbox{as} \quad \rho \to \infty \,,\\
  \Delta_\rho U_{1p} & = -{w^2/b} \,, \qquad U_{1p}^{\p}(0)=0 \,, \qquad
 U_{1p}\to -\log\rho + o(1) \,, \quad
  \mbox{as} \quad \rho \to \infty\,; \qquad b\equiv \int_{0}^{\infty} \rho w^2\,
 d\rho \,,
\end{split}
\end{equation}
 where $L_0 V_{1p} \equiv\Delta_\rho V_{1p} - V_{1p} + 2w V_{1p}$.  Finally, in
  (\ref{gcore:exp_1}), the constants $\chi_0$ and $\chi_1$ are related
  to $S_0$ and $S_1$ by
\begin{equation}
   \chi_0 = \sqrt{\frac{S_0}{b}} \,, \qquad \chi_1=\frac{S_1}{2\chi_0 b} -
  \frac{S_0}{b} \int_{0}^{\infty}w V_{1p} \rho \, d\rho \,. \label{gcore:exp_4}
\end{equation}
\esub
\end{lemma}

The derivation of this result is given in Appendix \ref{app:gm} below.
The $o(1)$ condition in the far-field behavior in (\ref{gcore:exp_3})
eliminates an otherwise arbitrary constant in the determination of
$U_{1p}$. Therefore, this condition ensures that the solution to the
linear BVP (\ref{gcore:exp_3}) is unique.

\subsection{The Spectrum of the Linearization Near the Origin}\label{gm:stab}

We linearize around the steady-state solution $u_e$ and $v_e$, as calculated
in \S \ref{gm:equil}, by introducing the perturbation (\ref{st:pert}).
This yields the following eigenvalue problem, where ${\cal P}_{\kb}$ is the
quasi-periodic boundary operator of (\ref{gr:pk}):
\begin{equation}
\label{gm:eig}
\begin{split}
 \eps^{2} \Delta \phi - \phi + \frac{2 v_e}{u_e} \phi - \frac{v_{e}^2}{u_e^2}
   \eta &=\lam \phi \,, \quad \xb \in \Omega\,; \qquad
{\cal P}_\kb \phi=0 \,, \quad \xb \in \partial\Omega \,, \\
 D\Delta\eta- \eta + 2 \eps^{-2} v_e \phi &= \lambda \tau \eta \,,
 \quad \xb \in
 \Omega \,; \qquad {\cal P}_\kb \eta =0\,, \quad \xb\in \partial\Omega\,.
\end{split}
\end{equation}

In the inner region near $\xb=\xbo$ we introduce the local variables
$N(\rho)$ and $\Phi(\rho)$ by
\begin{equation}
  \eta =  N(\rho) \,, \qquad
   \phi= \Phi(\rho) \,, \qquad \rho=|\yb| \,, \qquad
 \yb=\eps^{-1}\xb \,. \label{gm:inn}
\end{equation}
Upon substituting (\ref{gm:inn}) into (\ref{gm:eig}), and by using
$u_e\sim D U$ and $v_e\sim D V$, where $U$ and $V$
satisfy the core problem (\ref{geq1:core}), we obtain on $0<\rho<\infty$
that
\begin{equation}
\label{gm:eprob}
\begin{split}
  \Delta_\rho \Phi - \Phi + &  \frac{2V}{U} \Phi -\frac{V^2}{U^2} N
  =\lambda \Phi \,, \qquad \Phi \to 0 \,, \quad \mbox{as} \quad \rho
  \to \infty \,,\\ \Delta_\rho N &= -2 V \Phi \,, \qquad N
  \sim -C \log \rho + B \,, \quad \mbox{as} \quad \rho\to \infty \,,
\end{split}
\end{equation}
with $\Phi^{\prime}(0)=N^{\prime}(0)=0$ and where $B=B(S;\lambda)$. The
divergence theorem yields the identity
$C = 2 \int_{0}^{\infty} V \Phi \rho\, d\rho$.

To determine the constant $C$ we must match the far field behavior of
the core solution to an outer solution for $\eta$, which is valid away
from $\xb=\xbo$. Since $v_e$ is localized near $\xb=\xbo$, we
calculate in the sense of distributions that $ 2 \eps^{-2} v_e \phi
\rightarrow 2 D \left(\int_{\R^2} V \Phi \, d\yb\right) \, \delta(\xb)
= 2\pi C D \delta(\xb)$.  By using this expression in (\ref{gm:eig}),
we obtain that the outer problem for $\eta$ is
\begin{equation}
  \label{gm:etaout}
\begin{split}
   \Delta \eta & -\theta_\lam^2 \eta = -2\pi C \delta(\xb) \,,
 \quad \xb\in \Omega \,; \qquad {\cal P}_\kb \eta=0 \,, \quad \xb \in
 \partial\Omega \,, \\
   \eta & \sim -C \log|\xb| - \frac{C}{\nu} - B\,,   \quad \mbox{as}
  \quad \xb\to\xbo\,,
\end{split}
\end{equation}
where we have defined $\theta_\lambda\equiv \sqrt{(1+\tau\lambda)/D}$.
The solution to (\ref{gm:etaout}) is $\eta=2\pi C {\cal G}_{b\lam}(\xb)$, where
${\cal G}_{b\lam}$ satisfies
\begin{equation}
\label{gm_grlam:eta}
\begin{split}
   \Delta {\cal G}_{b\lam} -\theta_\lam^2 {\cal G}_{b\lam} &=-\delta(\xb) \,,
  \quad \xb \in \Omega \,; \qquad
  {\cal P}_\kb {\cal G}_{b\lam}=0 \,, \quad \xb\in \partial\Omega\,, \\
   {\cal G}_{b\lam} &\sim -\frac{1}{2\pi}\log|\xb| + {\cal R}_{b\lam} \,,
  \quad \mbox{as} \quad \xb\to \zero \,.
\end{split}
\end{equation}
By imposing that the behavior of $\eta$ as $\xb\to \xbo$ agrees with
that in (\ref{gm:etaout}), we conclude that $\left(1 + 2\pi \nu {\cal
  R}_{b\lam} \right) C = \nu B$. Then, since $D={D_0/\nu}\gg 1$, we
have from Lemma \ref{lemma 2.3}(ii), upon taking the $D\gg 1$ limit in
(\ref{gm_grlam:eta}), that $R_{b\lam}\sim R_{b0} + {\mathcal O}(\nu)$
for $|\kb|>0$ and ${\kb/(2\pi)}\in \Omega_B$. This yields,
\begin{equation}
 \left(1 + 2\pi \nu R_{b0} + {\mathcal O}(\nu^2)\right) C = \nu B
  \,, \label{gm:cb}
\end{equation}
where $R_{b0}=R_{b0}(\kb)$ is the regular part of the Bloch Green's
function $G_{b0}$ defined on $\Omega$ by (\ref{green:b0}).

As in \S \ref{schnak:stab}, we now proceed to determine the portion of
the continuous spectrum of the linearization that lies within an
${\mathcal O}(\nu)$ neighborhood of the origin $\lambda=0$ when $D$ is
close to a certain critical value. To do so, we first must calculate
an asymptotic expansion for the solution to (\ref{gm:eprob}) together
with (\ref{gm:cb}).

By using (\ref{gcore:exp_1}) we first calculate the coefficients in the
differential operator in (\ref{gm:eprob}) as
\begin{equation*}
  \frac{V}{U} = w + \frac{\nu S_0}{\chi_0} \left(V_{1p}-wU_{1p}\right)+
\cdots \,, \qquad \frac{V^2}{U^2} = w^2 + \frac{2\nu S_0}{\chi_0} w
  \left(V_{1p}-w U_{1p}\right) + \cdots \,,
\end{equation*}
so that the local problem (\ref{gm:eprob}) on $0<\rho<\infty$ becomes
\begin{equation}\label{gm:eprob_1}
\begin{split}
  \Delta_\rho \Phi & - \Phi + \left[ 2 w + \frac{2\nu S_0}{\chi_0} w \left(V_{1p}-
w U_{1p} \right) + \cdots \right] \Phi = \left[ w^2 + \frac{2\nu S_0}{\chi_0}
   w \left(V_{1p}-wU_{1p}\right)+\cdots \right]N +  \lambda \Phi \,,  \\
  \Delta_\rho N &  = -2\nu \left[ \chi_0 w + \nu \left(\chi_1 w + S_0 V_{1p}
  \right) +\cdots  \right] \Phi \,, \\
& \Phi\to 0 \,, \qquad   N \sim -C \log \rho + B \,, \quad \mbox{as} \quad
   \rho\to \infty \,; \qquad \Phi^{\prime}(0)=N^{\prime}(0)=0 \,.
\end{split}
\end{equation}
To analyze (\ref{gm:eprob_1}) together with (\ref{gm:cb}), we substitute
the appropriate expansions
\begin{equation} \label{gm:exp}
\begin{split}
  N &= \frac{1}{\nu} \left(\hat{N}_0+\nu \hat{N}_1 + \cdots\right)\,,
 \quad B = \frac{1}{\nu} \left(\hat{B}_0+\nu \hat{B}_1 + \cdots\right)\,,
  \quad C=C_0 + \nu C_1 + \cdots \,, \qquad \\
  \Phi &=\frac{1}{\nu}\left(\Phi_0 + \nu \Phi_1 +  \cdots\right) \,, \qquad
  \lambda=\lambda_0 + \nu \lambda_1 + \cdots \,,
\end{split}
\end{equation}
into (\ref{gm:eprob_1}) and collect powers of $\nu$.

To leading order, we obtain on $0<\rho<\infty$ that
\begin{equation}
\label{gm:zero}
\begin{split}
  L_0 \Phi_0 &\equiv  \Delta_\rho \Phi_0 - \Phi_0 + 2w \Phi_0 =
  w^2 \hat{N}_0 + \lambda_0 \Phi_0 \,, \qquad
  \Delta_\rho \hat{N}_0=0 \,, \\
 \Phi_0 &\to 0 \,, \qquad  \hat{N}_0\to \hat{B}_0 \quad \mbox{as} \quad
  \rho\to\infty \,; \qquad \Phi_0^{\prime}(0)=\hat{N}_0^{\prime}(0)=0 \,,
\end{split}
\end{equation}
where $L_0$ is the local operator. We conclude that
$\hat{N}_0=\hat{B}_0$ for $\rho\geq 0$.

At next order, we obtain on $\rho>0$ that $\Phi_1$ satisfies
\begin{equation}
  L_0\Phi_1 - w^2 \hat{N}_1 = -\frac{2S_0}{\chi_0} \left( V_{1p}-wU_{1p}\right)
 \Phi_0 + \frac{2S_0 }{\chi_0} w \left(V_{1p}-wU_{1p}\right)\hat{N}_0 + \lambda_1
 \Phi_0 \,;
\qquad \Phi_1\to 0 \,, \quad
\mbox{as} \quad \rho\to \infty \,, \label{gm:phi1}
\end{equation}
with $\Phi_1^{\prime}(0)=0$, and that $\hat{N}_1$ on $\rho>0$ satisfies
\begin{equation}
   \Delta_\rho \hat{N}_1 = -2\chi_0 w\Phi_0 \,; \qquad \hat{N}_1 \sim
 - C_0 \log\rho + \hat{B}_1 \,, \quad \mbox{as} \quad \rho\to \infty
   \,; \qquad \hat{N}_1^{\prime}(0)=0 \,. \label{gm:n1}
\end{equation}
At one higher order, the problem for $\hat{N}_2$ on $\rho>0$ is
\begin{equation}
\label{gm:n2}
  \Delta_\rho \hat{N}_2 = -2\chi_0 w \Phi_1  -2 \left(\chi_1 w + S_0
  V_{1p} \right) \Phi_0 \,; \qquad \hat{N}_2 \sim - C_1\log\rho + \hat{B}_2
    \,, \quad \mbox{as} \quad \rho\to \infty\,; \qquad
  \hat{N}_2^{\prime}(0)=0 \,.
\end{equation}
In addition, by substituting (\ref{gm:exp}) into (\ref{gm:cb}) we obtain,
upon collecting powers of $\nu$, that
\begin{equation}
  C_0 = \hat{B}_0 \,, \qquad C_1+2\pi R_{b0} \hat{B}_0 = \hat{B}_1 \,.
  \label{gm:cbexp}
\end{equation}

Next, we proceed to analyze (\ref{gm:zero})--(\ref{gm:n2}).  From the
divergence theorem, we obtain from (\ref{gm:n1}) that
\begin{equation}
   C_0 = 2\chi_0 \int_{0}^{\infty} w\Phi_0 \rho \, d\rho \,.
  \label{gm:c0b0}
\end{equation}
To identify $\chi_0$ in (\ref{gm:c0b0}), we substitute $S=\nu^2 S_0 +
\cdots$ and $\chi\sim \nu \chi_0 +\cdots$ into (\ref{gm:seq_a}) to get
$\left(1+\mu + \cdots\right)\left(\nu^2 S_0+\cdots\right)\sim \nu^2
(\chi_0+\cdots)$. From the leading order terms, we get
$\chi_0=S_0(1+\mu)$. Then, since $S_0=b\chi_0^2$ from
(\ref{gcore:exp_4}), we obtain
\begin{equation}
     \chi_0=\frac{1}{b(1+\mu)} \,, \qquad S_0 = \frac{1}{b(1+\mu)^2}
     \,, \qquad C_0 = \hat{B}_0 = \hat{N}_0 = \frac{2}{b(1+\mu)}
     \int_{0}^{\infty} w \Phi_0 \rho \, d\rho \,.
  \label{gm:cs_0}
\end{equation}
From (\ref{gm:zero}) we then obtain the leading-order NLEP on $\rho>0$,
\begin{equation}
   L_0 \Phi_0 - \frac{2 w^2 }{(1 + \mu)} \frac{ \int_{0}^{\infty}
     w\Phi_0 \, \rho \, d\rho}{\int_{0}^{\infty} w^2 \rho \, d\rho} =
   \lambda_0 \Phi_0\,; \qquad \Phi_0 \to 0 \,, \quad \mbox{as} \quad
   \rho\to \infty \,; \qquad \Phi_0^{\prime}(0)=0\,; \qquad \mu \equiv
   \frac{2\pi D_0}{|\Omega|}\,. \label{gm:phi0}
\end{equation}

For this NLEP, Theorem 3.7 of \cite{wsurv} proves that
$\mbox{Re}(\lambda_0)<0$ if and only if ${2/(1+\mu)}>1$.  Therefore,
the stability threshold where $\lambda_0=0$ and $\Phi_0=w$ occurs when
$\mu=1$. At this stability threshold, we calculate from (\ref{gm:cs_0})
that
\begin{equation}
   \chi_0=\frac{1}{2b} \,, \qquad S_0=\frac{1}{4b} \,, \qquad \Phi_0 = w \,,
 \qquad C_0=\hat{B}_0=\hat{N}_0=\frac{1}{b}\int_0^{\infty} w\Phi_0\rho \, d\rho
  =1 \,. \label{gm:marg}
\end{equation}

Upon substituting (\ref{gm:marg}) into (\ref{gm:n1})
we obtain at $\lambda_0=0$ that $\hat{N}_1$ on $\rho>0$ satisfies
\begin{equation}
   \Delta_\rho \hat{N}_1 = -2\chi_0 w^2 =-\frac{w^2}{b}\,; \qquad
   \hat{N}_1 \sim - \log\rho + \hat{B}_1 \,, \quad \mbox{as} \quad \rho\to
\infty \,; \qquad \hat{N}_1^{\prime}(0)=0 \,.\label{gm:n1n}
\end{equation}
Upon comparing (\ref{gm:n1n}) with the problem for $U_{1p}$ in
(\ref{gcore:exp_3}), we conclude that
\begin{equation}
   \hat{N}_1= U_{1p} + \hat{B}_1 \,. \label{gm:n1_solv}
\end{equation}

As in \S \ref{schnak:stab}, we now proceed to analyze the effect of
the higher order terms by determining the continuous band of spectrum
that is contained within an ${\mathcal O}(\nu)$ ball near $\lambda=0$
when the bifurcation parameter $\mu$ is ${\mathcal O}(\nu)$ close to
the leading-order critical value $\mu=1$. As such, we set
\begin{equation}
   \lambda = \nu \lambda_1 + \cdots \,, \qquad \mbox{for} \qquad
  \mu=1+\nu \mu_1 + \cdots \,,
\end{equation}
and we derive an expression for $\lambda_1$ in terms of the de-tuning
parameter $\mu_1$, the Bloch wavevector $\kb$, the lattice structure,
and certain correction terms to the core problem.

We first use (\ref{gcore:exp_4}) and (\ref{gm:seq_a}) to calculate
$\chi_1$ in terms of $\mu_1$. By substituting $\mu=1+\nu \mu_1+\cdots$
together with (\ref{gcore:exp_1}) into (\ref{gm:seq_a}), we obtain
\begin{equation*}
  \left[ 1 + \left(1+\nu \mu_1\right) + 2\pi \nu R_{0p} +\cdots \right]
  \left[\nu^2 S_0 + \nu^3 S_1 + \cdots \right] = \nu^2 \left(
\chi_0 + \nu \chi_1 + \cdots\right) \,.
\end{equation*}
From the ${\mathcal O}(\nu^3)$ terms, we obtain that $\chi_1=\mu_1S_0
+ 2 S_1 + 2\pi R_{0p} S_0$. Upon combining this result together with
(\ref{gcore:exp_4}) for $\chi_1$, and by using $\chi_0={1/(2b)}$, we
obtain at criticality where $\lam_0=0$ that
\begin{equation}
   \chi_1 = -\frac{\mu_1}{4b} - \frac{\pi R_{0p}}{2b} - \frac{1}{2b^2}
  \int_{0}^{\infty} w V_{1p} \rho \, d\rho \,. \label{gm:chi_to_mu}
\end{equation}
This result is needed below in the evaluation of the solvability condition.

Next, we substitute (\ref{gm:marg}) and (\ref{gm:n1_solv}) into
(\ref{gm:phi1}) for $\Phi_1$ to obtain, after
some algebra, that (\ref{gm:phi1}) reduces at the leading-order
stability threshold $\lambda_0=0$ to
\begin{equation}
  L_0 \Phi_1 -w^2 \hat{B}_1 = \lambda_1 w + w^2 U_{1p} \,; \qquad
 \Phi_1\to 0 \,, \quad \mbox{as} \quad \rho\to \infty \,, \label{gm:phi1n}
\end{equation}
with $\Phi_1^{\prime}(0)=0$. In a similar way, at the leading-order
stability threshold $\lambda_0=0$, the problem (\ref{gm:n2}) for
$\hat{N}_2$ on $\rho>0$ reduces to
\begin{equation}
\label{gm:n2n}
  \Delta_\rho \hat{N}_2 = -\frac{w}{b} \Phi_1 - 2 \left(\chi_1 w + \frac{1}{4b}
   V_{1p} \right) w \,; \qquad
   \hat{N}_2 \sim -C_1\log\rho + \hat{B}_2 \,, \quad \mbox{as} \quad
   \rho\to \infty\,; \qquad \hat{N}_2^{\prime}(0)=0 \,.
\end{equation}
By applying the divergence theorem to (\ref{gm:n2n}) we get
\begin{equation}
  C_1 = \frac{1}{b} \int_{0}^{\infty} w\Phi_1 \rho \, d\rho + 2\chi_1 b +
  \frac{1}{2b} \int_{0}^{\infty} w V_{1p} \rho \, d\rho \,. \label{gm:c1}
\end{equation}
Then, by using (\ref{gm:cbexp}) with $\hat{B}_0=1$ to relate $C_1$ to
$\hat{B}_1$, we determine $\hat{B}_1$ as $\hat{B}_1=C_1+2\pi R_{b0}$ where
$C_1$ is given in (\ref{gm:c1}). With $\hat{B}_1$ obtained in this
way, we find from (\ref{gm:phi1n}) that
$\Phi_1$ satisfies
\bsub \label{gm:phi1_fin}
\begin{equation}
 {\cal L}\Phi_1 \equiv L_0 \Phi_1 - w^2 \frac{ \int_0^\infty w\Phi_1
   \rho \, d\rho}{ \int_0^\infty w^2\rho \, d\rho} = {\cal R}_g +
 \lambda_1 w \,; \qquad \Phi_1\to 0 \,, \quad \mbox{as} \quad
  \rho \to \infty \,, \label{gm:phi1_fin_1}
\end{equation}
with $\Phi_1^{\prime}(0)=0$, where the residual ${\cal R}_g$ is defined by
\begin{equation}
  {\cal R}_g\equiv 2 \pi R_{b0} w^2 + 2\chi_1 b w^2 +
  \frac{1}{2b} w^2 \int_0^{\infty} w V_{1p} \rho \, d\rho +  w^2 U_{1p} \,.
 \label{gm:phi1_fin_2}
\end{equation}
\esub

As discussed in \S \ref{schnak:stab}, the solvability condition for
(\ref{gm:phi1_fin}) is that the right-hand side of
(\ref{gm:phi1_fin_1}) is orthogonal to the homogeneous adjoint
solution $\Psi^{\star}=w+{\rho w^{\prime}/2}$ in the sense that
$\lambda_1 \int_{0}^{\infty} w \Psi^{\star}\rho \, d\rho +
\int_{0}^{\infty} {\cal R}_g \Psi^{*} \rho \, d\rho =0$. Upon using
(\ref{gm:chi_to_mu}), which relates $\chi_1$ to $\mu_1$, to simplify
this solvability condition, we readily obtain by using
(\ref{gm:phi1_fin_2}) for ${\cal R}_g$ that
\begin{equation}
  \lambda_1 = -\frac{ \int_{0}^{\infty} w^2 \Psi^{\star} \rho \, d\rho}
  { \int_{0}^{\infty} w \Psi^{\star} \rho \, d\rho}\left( 2 \pi R_{b0}
   - \frac{\mu_1}{2}  -\pi R_{0p} - \frac{1}{2b} \int_{0}^{\infty}
   w V_{1p} \rho \, d\rho \right) -
  \frac{ \int_{0}^{\infty} w^2 U_{1p} \Psi^{\star} \rho \, d\rho}
  { \int_{0}^{\infty} w \Psi^{\star} \rho \, d\rho}  \,.
\label{gm:lam1_1}
\end{equation}

To simplify the terms in (\ref{gm:lam1_1}), we use
$L_0 V_{1p}=w^2 U_{1p}$ and $\Delta_\rho U_{1p}=-{w^2/b}$ from
(\ref{gcore:exp_3}), together with $w=L_0^{-1}\Psi^{\star}$ to
calculate, after an integration by parts, that
\begin{equation*}
  \int_0^{\infty} w^2 U_{1p} \Psi^{\star} \rho \, d\rho =
 \int_0^{\infty} \left( L_0 V_{1p}\right)\left(L_0^{-1} w\right)\rho \, d\rho=
  \int_{0}^{\infty} V_{1p} w \rho \, d\rho \,.
\end{equation*}
By substituting this expression, together with $\int_0^{\infty}
w^2\Psi^{\star} \rho \, d\rho=b$ and $\int_0^{\infty}
w\Psi^{\star}\rho \, d\rho={b/2}$, as obtained from (\ref{st:integ}),
into (\ref{gm:lam1_1}) we obtain our final result for $\lam_1$. We
summarize our result as follows:

\begin{result} \label{pr 4.2} In the limit
$\eps\to 0$, consider a steady-state periodic pattern of spots for the
  GM model (\ref{1:gm}) where $D={\mathcal O}(\nu^{-1})$ with
  $\nu={-1/\log\eps}$.  Then, when
\bsub \label{gm:lam1_fin}
\begin{equation}
  D \sim \frac{|\Omega|}{2\pi \nu} \left( 1 + \nu \mu_1 \right) \,, \qquad
  \label{gm:lam1_fin1}
\end{equation}
where $\mu_1={\mathcal O}(1)$ and $|\Omega|$ is the area of the Wigner-Seitz
cell, the portion of the continuous spectrum of
the linearization that lies within an ${\mathcal O}(\nu)$ neighborhood
of the origin $\lambda=0$ is given by
\begin{equation}
  \lambda=\nu \lambda_1 + \cdots \,, \qquad \lambda_1= 2 \left[
    \frac{\mu_1}{2} - 2 \pi R_{b0} + \pi R_{0p} - \frac{1}{2b}
    \int_{0}^{\infty} \rho w V_{1p} \, d\rho \right]
  \,. \label{gm:lam1_fin2}
\end{equation}
\esub
Here $R_{b0}=R_{b0}(\kb)$ is the regular part of the Bloch
Green's function $G_{b0}$ defined on $\Omega$ by (\ref{green:b0}),
${\kb/(2\pi)}\in \Omega_B$, and $R_{0p}$ is the regular part of the
periodic source-neutral Green's function $G_{0p}$ satisfying
(\ref{gr:source_neut}).
\end{result}

\begin{remark}\label{remark 2} In comparison with the
  analogous result obtained in Principal Result \ref{pr 3.2} for the
  Schnakenburg model, $\lambda_1$ in (\ref{gm:lam1_fin2}) now depends
  on the regular parts of two different Green's functions. The term
  $R_{0p}$ only depends on the geometry of the lattice, whereas
  $R_{b0}=R_{b0}(\kb)$ depends on both the lattice geometry and the
  Bloch wavevector $\kb$. To calculate $R_{b0}(\kb)$ we again need only
  consider vectors ${\kb/(2\pi)}$ in the first Brillouin zone $\Omega_B$
  of the reciprocal lattice. Since $R_{b0}$ is real-valued from Lemma
  \ref{lemma 2.1}, the band of spectrum (\ref{gm:lam1_fin2}) lies on
  the real axis in the $\lambda$-plane. Moreover, from Lemma
  \ref{lemma 2.2} small values of $|\kb|$ generate spectra that lie at
  an ${\mathcal O}\left({\nu/\kb^T {\cal Q}\kb}\right)$ distance from the
  origin along the negative real axis in the $\lambda$-plane, where
  ${\cal Q}$ is a positive definite matrix.
\end{remark}

For a given lattice geometry, we seek to determine $\mu_1$ so that
$\lambda_1<0$ for all $\kb$. From (\ref{gm:lam1_fin2}), we
conclude that a periodic arrangement of spots with a given lattice
structure is linearly stable when
\begin{equation}
   \mu_1 < 4\pi R_{b0}^{\star} - 2\pi R_{0p} + \frac{1}{b}
 \int_0^\infty w V_{1p}\rho \, d\rho
 \,, \qquad R_{b0}^{\star} \equiv \min_{\kb} R_{b0}(\kb) \,.
  \label{gm:mu1_opt}
\end{equation}
We characterize the optimal lattice as the one with a fixed
area $|\Omega|$ of the Wigner-Seitz cell that allows for stability for
the largest inhibitor diffusivity $D$. This leads to our second main result.

\begin{result} \label{pr 4.3} The optimal
  lattice arrangement for a periodic pattern of spots for the
  GM model (\ref{1:gm}) is the one for which the objective function
 ${\cal K}_{\textrm{gm}}$ is maximized, where
\begin{equation}
   {\cal K}_{\textrm{gm}} \equiv 4\pi R_{b0}^{\star} - 2\pi R_{0p}
 \,, \qquad R_{b0}^{\star} \equiv \min_{\kb} R_{b0}(\kb) \,.
  \label{gm:kap}
\end{equation}
For $\nu={-1/\log\eps} \ll 1$, a two-term asymptotic expansion for
 this maximal stability threshold for $D$ is given explicitly by
\begin{equation}
    D_{\textrm{optim}}\sim \frac{|\Omega|}{2\pi\nu} \left[ 1+
  \nu \left( \max_{\Lambda} {\cal K}_{\textrm{gm}} +
  \frac{1}{b} \int_{0}^{\infty} w V_{1p}\rho \, d\rho \right)\right]\,,
  \label{gm:doptim}
\end{equation}
 where $\max_{\Lambda}{\cal K}_{\textrm{gm}}$ is taken over all
  lattices $\Omega$ having a common area $|\Omega|$ of the
  Wigner-Seitz cell. In (\ref{gm:doptim}), $V_{1p}$ is the solution to
  (\ref{gcore:exp_3}) and $b=\int_{0}^{\infty} w^2\rho \, d\rho\approx
  4.93$ where $w(\rho)>0$ is the ground-state solution of $\Delta_\rho
  w - w + w^2=0$. Numerical computations yield $\int_{0}^{\infty} w
  V_{1p}\rho \, d\rho \approx -0.945$.
\end{result}

The numerical method to compute ${\cal K}_{\textrm{gm}}$ is given in
\S \ref{sec:ewald}. In \S \ref{sec:ewald_0}, we show numerically that
within the class of oblique Bravais lattices, the maximal stability
threshold for $D$ occurs for a regular hexagonal lattice.

\setcounter{equation}{0}
\setcounter{section}{4}
\section{A Simple Approach for Calculating the Optimal Value of the Diffusivity}\label{simp}

 In this section we implement a very simple alternative approach to
 calculate the stability threshold for the Schnakenburg (\ref{1:sc})
 and GM Models (\ref{1:gm}) in \S \ref{simp:schnak} and \S
 \ref{simp:gm}, respectively. In \S \ref{simp:gs} this method is then
 used to determine an optimal stability threshold for the GS
 model. In this alternative approach, we do not calculate the entire
 band of continuous spectrum that lies near the origin when the
 bifurcation parameter $\mu$ is ${\mathcal O}(\nu)$ close to its
 critical value. Instead, we determine the critical value of $\mu$,
 depending on the Bloch wavevector $\kb$, such that $\lambda=0$ is in
 the spectrum of the linearization. We then perform a min-max
 optimization of this critical value of $\mu$ with respect to $\kb$
 and the lattice geometry $\Lambda$ in order to find the optimal value
 of $D$.

\subsection{The Schnakenburg Model}\label{simp:schnak}

 This alternative approach for calculating the stability threshold requires
 the following two-term expansion for $\chi(S)$ in terms of $S$ as $S\to 0$:

\begin{lemma} \label{lemma 5.1} For
$S\to 0$, the asymptotic solution to the core problem (\ref{eq1:core}) is
\begin{equation}
\label{simp:exp}
 \begin{split}
  V &\sim \frac{S}{b} w + \frac{S^3}{b^3}\left(-\hat{\chi}_1 b w +
  V_{1p}\right) + \cdots \,, \qquad U \sim \frac{b}{S} + S \left(
  \hat{\chi}_1 + \frac{U_{1p}}{b}\right)+\cdots\,, \\
  \chi &\sim
  \frac{b}{S} + S \hat{\chi}_1 +\cdots \,; \qquad \hat{\chi}_1 \equiv
  \frac{1}{b^2} \int_{0}^{\infty}  V_{1p}\rho \, d\rho \,.
\end{split}
\end{equation}
 Here $w(\rho)$ is the unique positive ground-state solution to
  $\Delta_\rho w-w+w^2=0$ and $b\equiv \int_{0}^{\infty} w^2\rho
  d\rho$.  In terms of $w(\rho)$, the functions $U_{1p}$ and $V_{1p}$
  are the unique solutions on $0\leq\rho<\infty$ to
  (\ref{score:exp_3}).
\end{lemma}

The derivation of this result, as outlined at the end of Appendix
\ref{app:schnak}, is readily obtained by setting $S_1=0$ and
$S=S_0\nu^{1/2}$ in the results of Lemma \ref{lemma 3.1}.

The key step in the analysis is to note that at $\lambda=0$, the
solution to the inner problem (\ref{st:eprob}) for $\Phi$ and $N$ can
be readily identified by differentiating the core problem
(\ref{eq1:core}) with respect to $S$.  More specifically, at
$\lambda=0$, the solution to (\ref{st:eprob}) is $\Phi=C V_S$, $N=C
U_{S}$, and $B(S;0)=C \chi^{\prime}(S)$. With $B$ known at
$\lambda=0$, we obtain from (\ref{st:cb}) and (\ref{eq1:sval}) that
the critical value of $D$ at $\lambda=0$ satisfies the nonlinear
algebraic problem
\begin{equation}
  1 + 2\pi \nu R_{b0} + {\mathcal O}(\nu^2)  + \nu \chi^{\prime}(S) =0 \,,
 \qquad \mbox{where} \quad  S = \frac{a|\Omega|}{2\pi\sqrt{D}} \,.
 \label{simp:sc_crit}
\end{equation}

To determine the critical threshold in $D$ from (\ref{simp:sc_crit}) we
use the two-term expansion for $\chi(S)$ in (\ref{simp:exp}) to get
$\chi^{\prime}(S)\sim -{b/S^2} + \hat{\chi}_1+\cdots$. By using the
relation for $S$ in terms of $D$ from (\ref{simp:sc_crit}) when
$D={D_0/\nu}\gg 1$, we obtain that
\begin{equation}
  \chi^{\prime}(S)\sim -\frac{\mu}{\nu} + \hat{\chi}_1 + \cdots  \,, \qquad
   \mu \equiv \frac{4\pi^2 D_0 b}{a^2 |\Omega|^2} \,, \qquad
  D=\frac{D_0}{\nu} \,.  \label{simp:sc_crit_1}
\end{equation}
Upon substituting this expression into (\ref{simp:sc_crit}), we obtain that
\begin{equation*}
   1-\mu + \nu \hat{\chi}_1 = - 2\pi \nu R_{b0} + {\mathcal O}(\nu^2) \,,
\end{equation*}
which determines $\mu$ as $\mu\sim 1+\nu(2\pi
R_{b0}+\hat{\chi}_1)$. Upon recalling the definition of $\mu$ in
(\ref{simp:sc_crit_1}), we conclude that $\lambda=0$ when
$D=D^{\star}(\kb)$, where $D^{\star}(\kb)$ is given by
\begin{equation}
   D^{\star}(\kb) \equiv \frac{a^2 |\Omega|^2}{4\pi^2 b \nu} \left[
      1 + \nu \left(2\pi R_{b0}(\kb) + \hat{\chi}_1 \right) +
   {\mathcal O}(\nu^2) \right]\,, \label{simp_sc:optim}
\end{equation}
where $\hat{\chi}_1$ is defined in (\ref{simp:exp}). By minimizing
$R_{b0}(\kb)$ with respect to $\kb$, and then maximizing the result
with respect to the geometry of the lattice $\Lambda$,
(\ref{simp_sc:optim}) recovers the main result (\ref{st:doptim}) of
Principal Result \ref{pr 3.3}. This simple method, which relies
critically on the observation that $B=\chi^{\prime}(S)$ at
$\lambda=0$, provides a rather expedient approach for calculating the
optimal threshold in $D$. However, it does not characterize the
spectrum contained in the small ball $|\lambda|={\mathcal O}(\nu)\ll
1$ near the origin when $D$ is near the leading-order stability
threshold ${a^2 |\Omega|^2/(4\pi^2 b \nu)}$.

\subsection{The Gierer-Meinhardt Model}\label{simp:gm}

Next, we use a similar approach as in \S \ref{simp:schnak} to
re-derive the the stability result in (\ref{gm:doptim}) of Principal
Result \ref{pr 4.3} for the GM model. We first need the following
result that gives a two-term expansion in terms of $S$ for $\chi(S)$
as $S\to 0$:

\begin{lemma} \label{lemma 5.2} For
$S\to 0$, the asymptotic solution to the core problem (\ref{geq1:core}) is
\begin{equation}
\label{gsimp:exp}
 \begin{split}
  V &\sim \sqrt{\frac{S}{b}} w + S \left(\hat{\chi}_1 w + V_{1p}\right) + \cdots
\,, \qquad
U \sim \sqrt{\frac{S}{b}} + S \left(\hat{\chi}_1  + U_{1p}\right) + \cdots\,,\\
 \chi &\sim \sqrt{\frac{S}{b}} + S \hat{\chi}_1 + \cdots \,, \qquad
  \hat{\chi}_1 \equiv -\frac{1}{b} \int_{0}^{\infty} w V_{1p}\rho \, d\rho \,.
\end{split}
\end{equation}
Here $w(\rho)$ is the unique positive ground-state solution to
  $\Delta_\rho w-w+w^2=0$ and $b\equiv \int_{0}^{\infty} w^2\rho
  d\rho$.  In terms of $w(\rho)$, the functions $U_{1p}$ and $V_{1p}$
  are the unique solutions on $0\leq\rho<\infty$ to
  (\ref{gcore:exp_3}).
\end{lemma}

The derivation of this result, as outlined at the end of Appendix
\ref{app:gm}, is readily obtained by setting $S_1=0$ and
$S=S_0\nu^{2}$ in the results of Lemma \ref{lemma 4.1}.

As similar to the analysis in \S \ref{simp:schnak}, the solution to
(\ref{gm:eprob}) for $\Phi$ and $N$ is readily
identified by differentiating the core problem (\ref{geq1:core}) with
respect to $S$. In this way, we get $B(S,0)=C \chi^{\prime}(S)$. Therefore, at
$\lambda=0$, we obtain from (\ref{gm:cb}) and (\ref{gm:seq_a}) that
the critical values of $D$ and $S$ where $\lambda=0$ satisfy the
coupled nonlinear algebraic system
\begin{equation} \label{gsimp:nonlin}
\begin{split}
  \left( 1 + \mu + 2\pi \nu R_{0p} + {\mathcal O}(\nu^2)\right) S &=
  \nu \chi(S) \,, \qquad  \mu \equiv \frac{2\pi D_0}{|\Omega|} \,,
  \qquad D=\frac{D_0}{\nu} \,, \\
1 + 2\pi \nu R_{b0} + {\mathcal O}(\nu^2) -\nu \chi^{\prime}(S) &=0\,.
\end{split}
\end{equation}
We then use the two term expansion in (\ref{gsimp:exp}) for $\chi(S)$
as $S\to 0$ to find an approximate solution to (\ref{gsimp:nonlin}).

In contrast to the related analysis for the Schnakenburg model in \S
\ref{simp:schnak}, this calculation is slightly more involved since
$S$ must first be calculated from a nonlinear algebraic equation. By
substituting (\ref{gsimp:exp}) for $\chi(S)$ into the first equation
of (\ref{gsimp:nonlin}), and expanding $\mu=\mu_0+\nu \mu_1 + \cdots$,
we obtain
\begin{equation*}
  \left[1 + \mu_0 + \nu\left(\mu_1 + 2\pi \nu R_{0p}\right) \right] S
  \sim \nu \left( \sqrt{\frac{S}{b}} + S \hat{\chi}_1\right) \,,
\end{equation*}
which can be solved asymptotically when $\nu\ll 1$ to get the two-term
expansion for $S$ in terms of $\mu_0$ and $\mu_1$ given by
\begin{equation}
  S = \nu^2 \left( \hat{S}_0 + \nu \hat{S}_1 + \cdots\right) \,; \qquad
  \hat{S}_0 \equiv \frac{1}{b(1+\mu_0)^2} \,, \qquad
  \hat{S}_1 \equiv \frac{2}{b(1+\mu_0)^3}\left( \hat{\chi}_1 -\mu_1 -2\pi R_{0p}
  \right) \,. \label{gsimp:sexp}
\end{equation}

From the two-term expansion (\ref{gsimp:sexp}) for $S$ we calculate
$\chi^{\prime}(S)$ from (\ref{gsimp:exp}) as
\begin{equation*}
  \chi^{\p}(S) \sim \frac{1}{2\sqrt{b}\nu} \left(\hat{S}_0+ \nu \hat{S}_1
 + \cdots \right)^{-1/2} + \hat{\chi}_1 \sim
  \frac{\hat{S}_0^{-1/2}}{2\sqrt{b}\nu} + \left[ \hat{\chi}_1 -
  \frac{\hat{S}_1}{4\sqrt{b} \hat{S}_0^{3/2}} \right] + {\mathcal O}(\nu) \,.
\end{equation*}
By using (\ref{gsimp:sexp}) for $\hat{S}_0$ and $\hat{S}_1$, the expression
above becomes
\begin{equation} \label{gsimp:chip}
 \chi^{\p}(S) \sim \frac{1}{2\nu} \left[ (1+\mu_0) + \nu \left( \hat{\chi}_1 +
  \mu_1+ 2\pi R_{0p} \right) + {\mathcal O}(\nu^2) \right] \,.
\end{equation}
Then, upon substituting (\ref{gsimp:chip}) into the second equation of
(\ref{gsimp:nonlin}) we obtain, up to ${\mathcal O}(\nu)$ terms, that
\begin{equation*}
  1 + 2\pi \nu R_{b0} \sim (1+\mu_0) + \frac{\nu}{2} \left(
  \hat{\chi}_1 + 2\pi R_{0p} + \mu_1\right) \,,
\end{equation*}
which determines $\mu_0$ and $\mu_1$ as
\begin{equation}
  \mu_0=1 \,, \qquad \mu_1 = -\hat{\chi}_1 - 2\pi R_{0p} + 4\pi R_{b0} \,.
  \label{gsimp:fin_mu}
\end{equation}

Finally, by recalling the definition of $\mu$ and $\hat{\chi}_1$ in
(\ref{gsimp:nonlin}) and (\ref{gsimp:exp}), respectively, and by using
the two-term expansion $\mu=\mu_0 +\nu \mu_1$ from
(\ref{gsimp:fin_mu}), we conclude that $\lambda=0$ when
$D=D^{\star}(\kb)$, where $D^{\star}(\kb)$ is given by
\begin{equation}
   D^{\star}(\kb) \equiv \frac{|\Omega|}{2\pi\nu} \left[
      1 + \nu \left(4\pi R_{b0}(\kb) -2\pi R_{0p} +
  \frac{1}{b} \int_{0}^{\infty} w V_{1p} \rho \, d\rho \right) +
    {\mathcal O}(\nu^2) \right]\,. \label{gsimp:optim}
\end{equation}
By minimizing $R_{b0}(\kb)$ with respect to $\kb$, and then maximizing
the result with respect to the geometry of the lattice $\Lambda$,
(\ref{gsimp:optim}) recovers the main result (\ref{gm:doptim}) of
Principal Result \ref{pr 4.3}.

\subsection{The Gray-Scott Model}\label{simp:gs}

In this sub-section we employ the simple approach of \S
\ref{simp:schnak} and \S \ref{simp:gm} to optimize a stability
threshold for a periodic pattern of localized spots for the GS model,
where the spots are localized at the lattice points of the Bravais
lattice $\Lambda$ of (\ref{lattice-def}).  In the Wigner-Seitz cell
$\Omega$, the GS model in the dimensionless form of \cite{MO1} is
\begin{equation}
v_t = \eps^2\, \Delta v - v + A u v^2 \,, \quad
\tau u_t = D\, \Delta u +(1-u) - u v^2\,, \quad \xb \in \Omega\,;
 \qquad {\cal P}_0 u = {\cal P}_0 v = 0\,, \quad \xb \in
  \partial \Omega \,, \label{1:gs}
\end{equation}
where $\eps>0$, $D>0$, $\tau>1$, and the feed-rate parameter $A>0$ are
constants. In various parameter regimes of $A$ and $D$, the stability
and self-replication behavior of localized spots for (\ref{1:gs}) have
been studied in \cite{MO1}, \cite{MO2}, \cite{MO3}, \cite{WGS2}, and
\cite{cw_1} (see also the references therein).  We will consider the
parameter regime $D={\mathcal O}(\nu^{-1})\gg 1$ and $A={\mathcal
  O}(\eps)$ of \cite{WGS2}. In this regime, and to leading order in
$\nu$, an existence and stability analysis of $N$-spot patterns in a
finite domain was undertaken via a Lypanunov Schmidt reduction and a
rigorous study of certain nonlocal eigenvalue problems. We briefly
review the main stability result of \cite{WGS2} following
(\ref{gs:aeq}) below.

We first construct a one-spot steady-state solution to (\ref{1:gs})
with spot centered at $\xb=\xbo$ in $\Omega$ in the regime
$D={\mathcal O}(\nu^{-1})$ and $A={\mathcal O}(\eps)$ by using the
approach in \S 2 of \cite{cw_1}.

In the inner region near $\xb=\xbo$ we introduce the local variables
$U$, $V$, and $\yb$, defined by
\begin{equation}
 u = \frac{\eps}{A \sqrt{D}} U \,, \qquad v = \frac{\sqrt{D}}{\eps}
V \,, \qquad \yb=\eps^{-1} \xb \,, \label{3:2dinnvar}
\end{equation}
into the steady-state problem for (\ref{1:gs}). We obtain that
$U\sim U(\rho)$ and $V\sim V(\rho)$, with $\rho=|\yb|$, satisfy the
same core problem
\bsub \label{gs:core}
\begin{gather}
 \Delta_\rho V - V + U V^2 =0 \,, \qquad
 \Delta_\rho U - U V^2=0 \,, \qquad 0 < \rho < \infty \,,\label{gs:core_1}\\
 U^{\prime}(0)=V^{\prime}(0)=0\,; \qquad V
 \to 0 \,, \qquad U \sim S \log\rho + \chi(S) + o(1) \,, \quad
 \mbox{as} \quad \rho\to \infty \,, \label{gs:core_2}
\end{gather}
\esub
as that for the Schnakenburg model studied in \S
\ref{schnak:equil}, where $S\equiv \int_{0}^{\infty} U V^2\rho \,
d\rho$ and $\Delta_\rho V\equiv V^{\p\p} + \rho^{-1}
V^{\p}$. Therefore, for $S\to 0$, the two-term asymptotics of
$\chi(S)$ is given in (\ref{simp:exp}) of Lemma \ref{lemma 5.1}.

To formulate the outer problem for $u$, we observe that since $v$ is
localized near $\xb=\xbo$ we have in the sense of
distributions that $uv^2 \rightarrow \eps^2 \left( \int_{\mathbb{R}^2}
\sqrt{D} \left(A \eps\right)^{-1} U V^2 \, d\yb\right) \, \delta(\xb)
\sim 2 \pi \eps \sqrt{D} A^{-1} S \,\delta(\xb)$. Then, upon matching $u$
to the core solution $U$, we obtain from (\ref{1:gs}) that
\begin{equation} \label{3:uout}
\begin{split}
 \Delta u  +\frac{1}{D}(1-u) &= \frac{2 \pi \,\eps}{A\sqrt{D}} S \,\delta(\xb)
  \,, \quad \xb\in \Omega \,; \qquad {\cal P}_0 u=0 \,, \quad \xb\in
  \partial\Omega \,, \\
  u &\sim  \frac{\eps}{A \sqrt{D}} \left( S \log|\xb| + \frac{S}{\nu} +
  \chi(S) \right) \,, \quad \mbox{as} \quad \xb \to \xbo \,,
\end{split}
\end{equation}
where $\nu\equiv {-1/\log\eps}$.  The solution to (\ref{3:uout}) is $u
= 1 -2\pi\eps S {G_{p}(\xb)/(A\sqrt{D})}$, where $G_p(\xb)$ is the
Green's function of (\ref{greq1:uout}).  Next, we calculate the local
behavior of $u$ as $\xb\to \xbo$ and compare it with the required
behavior in (\ref{3:uout}). This yields that $S$ satisfies
\begin{equation}
 S + \nu \left[ \chi(S) + 2\pi S R_p\right] = \frac{A \nu \sqrt{D}}{\eps}
  \,, \label{gs:seq1}
\end{equation}
where $R_p$ is the regular part of $G_p$ as defined in (\ref{greq1:uout}).

We consider the regime $D={D_0/\nu}\gg 1$ with $D_0={\mathcal O}(1)$. By
using the two-term expansion (\ref{gm:green_exp}) for $R_p$ in terms
of the regular part $R_{0p}$ of the periodic source-neutral Green's
of  (\ref{gr:source_neut}), (\ref{gs:seq1}) becomes
\bsub \label{gs:eq}
\begin{equation}
   S \left(1+\mu\right)+ \nu \left[2\pi S R_{0p} + \chi(S) \right] +
  {\mathcal O}(\nu^2) = {\cal A} \sqrt{\nu\mu} \,,
  \label{gs:eq_1}
\end{equation}
where we have defined $\mu$ and ${\cal A}={\mathcal O}(1)$ in terms of
$A={\mathcal O}(\eps)$ by
\begin{equation}
    {\cal A} = \frac{A}{\eps} \sqrt{ \frac{|\Omega|}{2\pi}} \,, \qquad
  \mu \equiv \frac{2\pi D_0}{|\Omega|} \,, \qquad D=\frac{D_0}{\nu} \,.
  \label{gs:aeq}
\end{equation}
\esub

To illustrate the bifurcation diagram associated with (\ref{gs:eq_1}), we
use $\chi(S)\sim {b/S}$ as $S\to 0$ from (\ref{simp:exp}) of Lemma
\ref{lemma 5.1}. Upon writing $S=\nu^{1/2} {\cal S}$, with
${\cal S}={\mathcal O}(1)$, we obtain from (\ref{gs:eq}) that,
to leading order in $\nu$,
\begin{equation}
    {\cal A} \sqrt{\mu} = {\cal S} (1+\mu) + \frac{b}{{\cal S}} \,; \qquad
     \mu = \frac{2\pi D_0}{|\Omega|} \,, \qquad
 b=\int_{0}^{\infty} w^2\rho \, d\rho \,. \label{gs:bif}
\end{equation}
From Lemma \ref{lemma 5.1} and (\ref{3:2dinnvar}), the spot amplitude
$V(\xbo)$ to leading order in $\nu$ is related to ${\cal S}$ by
$V(\xbo)=\eps^{-1}\sqrt{D_0}{\cal S} {w(0)/b}$. In Fig.~\ref{gsplot}
we use (\ref{gs:bif}) to plot the leading-order saddle-node
bifurcation diagram of ${\cal S}$ versus ${\cal A}$, where the upper
solution branch corresponds a pattern with large amplitude spots. The
saddle-node point occurs when ${\cal S}_f = \sqrt{ b/(1+\mu)}$ and
${\cal A}_f = 2\sqrt{b} \sqrt{ (1+\mu)/\mu}$.  As we show below, there
is a zero eigenvalue crossing corresponding to an instability for some
Bloch wavevector $\kb$ with $|\kb|>0$ and ${\kb/(2\pi)}\in \Omega_B$
that occurs within an ${\mathcal O}(\nu)$ neighborhood of the point
$({\cal S}_0,{\cal A}_0)$ on the upper branch of Fig.~\ref{gsplot}
given by ${\cal S}_0 = \sqrt{b}$ and ${\cal A}_0
=\left(2+\mu\right)\sqrt{{b/\mu}}$. Since $|\kb|>0$ for this
instability, we refer to it as a competition instability. Below, we
will expand ${\cal A}={\cal A}_0 + \nu {\cal A}_1 +\cdots$, and
determine the optimal lattice arrangement of spots that minimizes
${\cal A}_1$. This has the effect of maximizing the extent of the
upper solution branch in Fig.~\ref{gsplot} that is stable to
competition instabilities.

\begin{figure}[htb]
\begin{center}
\includegraphics[width = 8cm,height=5.5cm,clip]{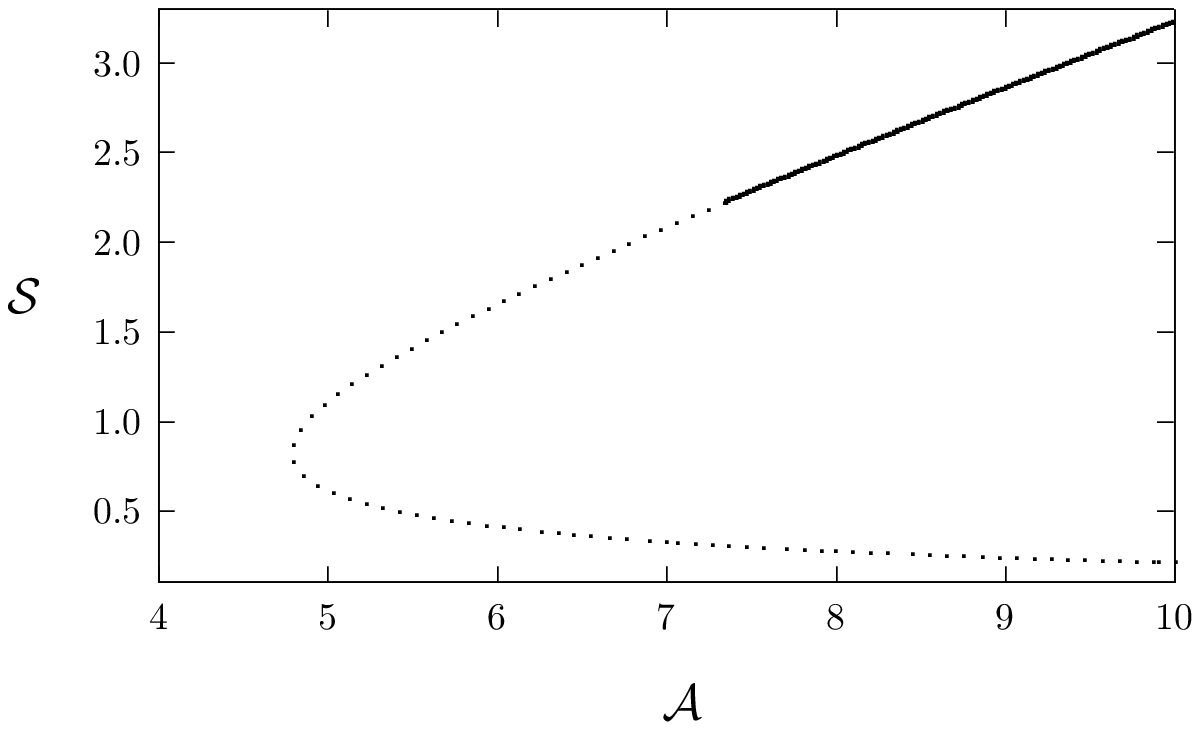}
\caption{Plot, to leading-order in $\nu$, of the saddle-node
  bifurcation diagram ${\cal S}$ versus ${\cal A}$, obtained from
  (\ref{gs:bif}), for the GS model with $|\Omega|=1$ and $D_0=1$. The
  leading-order spot amplitude $V(0)=\eps^{-1}\sqrt{D}_0 {\cal S}{w(0)/b}$ is
  directly proportional to ${\cal S}$. The heavy solid branch of large
  amplitude spots is linearly stable to competition instabilities,
  while the dotted branch is unstable to competition instabilities. To
  leading order in $\nu$, the zero eigenvalue crossing corresponding to the
  competition instability threshold occurs at
  ${\cal A}_0=\left(2+\mu\right)\sqrt{{b/\mu}}\approx 7.34$ where ${\cal
    S}_0=\sqrt{b}\approx 2.22$.}
\label{gsplot}
\end{center}
\end{figure}

Before proceeding with the calculation of the optimal lattice for the
periodic problem, we recall some prior rigorous results of \cite{WGS2}
for the finite domain problem with $N$ localized spots in a finite
domain $\Omega_N$ with homogeneous Neumann boundary conditions. From
\cite{WGS2}, the bifurcation diagram to leading order in $\nu$ is
\begin{equation*}
    {\cal A} \sqrt{\mu_N} = {\cal S} (1+\mu_N) + \frac{b}{{\cal S}} \,, \qquad
    b=\int_{0}^{\infty} w^2\rho \, d\rho \,, \qquad \mu_N=
 \frac{2\pi N D_0}{|\Omega_N|} \,, \qquad {\cal A} = \frac{A}{\eps}
 \sqrt{\frac{|\Omega_N|}{2\pi N}} \,,
\end{equation*}
which shows that we need only replace $|\Omega|$ in (\ref{gs:bif})
with ${|\Omega_N|/N}$. From a rigorous NLEP analysis of the finite
domain problem, it was proved in \cite{WGS2} that the lower solution
branch in Fig.~\ref{gsplot} is unstable to synchronous instabilities,
while the upper branch is stable to such instabilities. In
contrast, it is only the portion of the upper solution branch with
${\cal S}>{\cal S}_0$ that is stable to competition instabilities
(see Fig.~\ref{gsplot}). Therefore, there is two zero eigenvalue
crossings; one at the saddle-node point corresponding to a synchronous
instability, and one at the point $({\cal S}_0,{\cal A}_0)$ on the upper branch
corresponding to a competition instability.

Similarly, for the periodic spot problem we remark that there is also
a zero eigenvalue crossing when $\kb=0$, i.e. the synchronous
instability, which occurs at the saddle-node bifurcation
point. However, since it is the zero eigenvalue crossing for the
competition instability that sets the instability threshold for ${\cal
  A}$ (see Fig.~\ref{gsplot}), we will not analyze the effect of the
lattice geometry on the zero eigenvalue crossing for the synchronous
instability mode.

We now proceed to analyze the zero eigenvalue crossing for the
competition instability. To determine the stability of the
steady-state solution $u_e$ and $v_e$, we introduce (\ref{st:pert})
into (\ref{1:gs}) to obtain the Floquet-Bloch eigenvalue problem
\begin{equation}
\label{gs:eig}
\begin{split}
 \eps^2 \Delta \phi - \phi + 2 A u_e v_e \phi + A v_e^2 \eta &=\lambda \phi\,,
\quad \xb \in \Omega\,; \qquad {\cal P}_\kb \phi=0 \,, \quad \xb \in
\partial\Omega \,, \\
 D \Delta \eta - \eta - 2 u_e v_e \phi - v_e^2 \eta &= \lambda \tau \eta\,;
 \qquad {\cal P}_\kb \phi=0 \,, \quad \xb \in \partial\Omega \,.
\end{split}
\end{equation}

In the inner region near $\xb=\xbo$ we look for a locally radially symmetric
eigenpair $N(\rho)$ and $\Phi(\rho)$, with $\rho=|\yb|$, defined in terms
of $\eta$ and $\phi$ by
\begin{equation}
  \eta =  \frac{\eps}{A\sqrt{D}} N(\rho) \,, \qquad
   \phi= \frac{\sqrt{D}}{\eps} \Phi(\rho) \,, \qquad \rho=|\yb| \,, \qquad
 \yb=\eps^{-1}\xb \,. \label{gs:inn}
\end{equation}
From (\ref{gs:eig}), we obtain to within negligible ${\mathcal O}(\eps^2)$
terms that $N(\rho)$ and $\Phi(\rho)$ satisfy
\begin{equation}\label{gs:eprob}
\begin{split}
  \Delta_\rho \Phi - \Phi + & 2 U V \Phi + N V^2 =\lambda \Phi \,, \qquad
  \Phi \to 0 \,, \quad \mbox{as} \quad \rho \to \infty \,,\\ \Delta_\rho N &= 2
  UV \Phi + N V^2 \,, \qquad N \sim C \log \rho + B \,, \quad
  \mbox{as} \quad \rho\to \infty \,,
\end{split}
\end{equation}
with $\Phi^{\prime}(0)=N^{\prime}(0)=0$, $B=B(S;\lambda)$, and
$C =\int_{0}^{\infty} \left( 2 U V \Phi + N V^2 \right)\rho\, d\rho$.

To determine the outer problem for $\eta$, we first calculate in the sense
of distributions that
\begin{equation}
 2 u_e v_e \phi + v_e^2 \eta \rightarrow \frac{\sqrt{D}}{A\eps} \left[
  \eps^2 \int_{\R^2} \left(2 U V \Phi + V^2 N\right)\, d\yb \right]\delta(\xb)
 =  \frac{2\pi \eps\sqrt{D}}{A} C \delta(\xb) \,. \label{gs:dist}
\end{equation}
Then, by asymptotically matching $\eta$ as $\xb\to \xbo$ with the far-field
behavior of $N$ in (\ref{gs:eprob}), we obtain from (\ref{gs:dist}) and
(\ref{gs:eig}) that the outer problem for $\eta$ is
\begin{equation}  \label{gs:etaout}
\begin{split}
   \Delta \eta & -\theta_\lam^2 \eta = \frac{2\pi\eps}{A\sqrt{D}} C
  \delta(\xb) \,,  \quad \xb\in \Omega \,; \qquad {\cal P}_\kb \eta=0 \,,
  \quad \xb \in \partial\Omega \,, \\
   \eta & \sim \frac{\eps}{A\sqrt{D}}
 \left[C \log|\xb| + \frac{C}{\nu}  B\right]\,,   \quad \mbox{as}
  \quad \xb\to\xbo\,.
\end{split}
\end{equation}
Here we have defined $\theta_\lambda\equiv \sqrt{(1+\tau\lambda)/D}$.
The solution to (\ref{gs:etaout}) is $\eta=-2\pi \eps C {{\cal
    G}_{b\lam}(\xb)/(A\sqrt{D})}$, where ${\cal G}_{b\lam}$ satisfies
(\ref{gm_grlam:eta}). By imposing that the behavior of $\eta$ as
$\xb\to \xbo$ agrees with that in (\ref{gs:etaout}), we conclude that
$\left(1 + 2\pi \nu {\cal R}_{b\lam} \right) C + \nu B=0$, where
$R_{b\lam}$ is the regular part of $G_{b\lam}$ defined in
(\ref{gm_grlam:eta}). Then, since $D={D_0/\nu}\gg 1$, we have from
Lemma \ref{lemma 2.3}(ii) upon taking the $D\gg 1$ limit in
(\ref{gm_grlam:eta}) that $R_{b\lam}\sim R_{b0} + {\mathcal O}(\nu)$
for $|\kb|>0$ and ${\kb/(2\pi)}\in\Omega_B$. Thus, we have
\begin{equation}
 \left(1 + 2\pi \nu R_{b0} + {\mathcal O}(\nu^2)\right) C - \nu B=0
  \,, \label{gs:cb}
\end{equation}
where $R_{b0}=R_{b0}(\kb)$ is the regular part of the Bloch Green's
function $G_{b0}$ defined on $\Omega$ by (\ref{green:b0}). We remark
that if we were to consider zero-eigenvalue crossings for a synchronous
instability where $\kb=0$, we would instead use
$R_{b\lam}=R_{p}\sim {D_0/\nu|\Omega|} + R_{0p}+\cdots$ from (\ref{gm:green_exp})
to obtain $\left(1+\mu + 2\pi\nu R_{p0} +{\mathcal O}(\nu^2)\right)C +
\nu B=0$ in place of (\ref{gs:cb}).

As in \S \ref{simp:schnak} we use the key fact that at $\lambda=0$, we
have $B(S;0)=C\chi^{\prime}(S)$. Therefore, at $\lambda=0$, we obtain
from (\ref{gs:cb}) and (\ref{gs:eq_1}) that the critical values of
${\cal A}$ and $S$ where $\lambda=0$ satisfy the coupled nonlinear
algebraic system
\begin{equation} \label{gs:coup}
   S \left(1+\mu\right)+ \nu \left[2\pi S R_{0p} + \chi(S) \right] +
  {\mathcal O}(\nu^2) = {\cal A} \sqrt{\nu\mu} \,,
\qquad 1 + 2\pi \nu R_{b0} + {\mathcal O}(\nu^2) + \nu \chi^{\prime}(S) =0\,.
\end{equation}

The final step in the calculation is to use the two term expansion for
$\chi(S)$, as given in (\ref{simp:exp}) of Lemma \ref{lemma 5.1}, to obtain a
two-term approximate solution in powers of $\nu$ to
(\ref{gs:coup}). By substituting $\chi^{\prime}(S)\sim -b S^{-2} +
\hat{\chi}_1$ for $S\ll 1$ into the second equation of
(\ref{gs:coup}), we readily calculate a two-term expansion for $S$
as
\begin{equation}
   S \sim \sqrt{b \nu} \left( 1 + \nu \hat{S}_1 + \cdots\right) \,,
   \qquad \hat{S}_1 \equiv -\frac{1}{2}\left(\hat{\chi}_1 + 2\pi
   R_{b0}\right)\,.
  \label{gs:hats}
\end{equation}
Then, we substitute (\ref{gs:hats}), together with the two-term expansion
\begin{equation}
   {\cal A}={\cal A}_0 + \nu {\cal A}_1 + \cdots \,, \label{gs:a_exp}
\end{equation}
into the first equation of (\ref{gs:coup}), and equate powers of $\nu$.
From the ${\mathcal O}(\nu^{1/2})$ terms in the resulting expression we
obtain that ${\cal A}_0 = \sqrt{b}{(2+\mu)/\sqrt{\mu}}$, while at order
${\mathcal O}(\nu^{3/2})$ we get that ${\cal A}_1={\cal A}_1(\kb)$ satisfies
\begin{equation}
  \frac{{\cal A}_1}{\sqrt{b\mu}} = \frac{2\pi R_{0p}}{\mu}
 + \frac{\hat{\chi}_1}{\mu} + \hat{S}_1 = \frac{2\pi R_{0p}}{\mu}  -
 \pi R_{b0}(\kb) + \hat{\chi}_1 \frac{(2-\mu)}{2\mu} \,,\label{gs:a1}
\end{equation}
where $\hat{\chi}_1$ is given in (\ref{simp:exp}) of Lemma \ref{lemma 5.1}.

To determine the optimal lattice that allows for stability for the
smallest value of ${\cal A}$, we first fix a lattice $\Lambda$ and
then maximize ${\cal A}_1$ in (\ref{gs:a1}) through minimizing
$R_{b0}(\kb)$ with respect to the Bloch wavevector $\kb$. Then, we
minimize ${\cal A}_1$ with respect to the lattice geometry $\Lambda$
while fixing $|\Omega|$.  We summarize this third main result as follows:

\begin{result}\label{pr 5.3} The optimal
  lattice arrangement for a steady-state periodic pattern of spots for
  the GS model (\ref{1:gs}) in the regime $D={D_0/\nu}\gg 1$ and
  $A={\mathcal O}(\eps)$ is the one for which the objective function
  ${\cal K}_{\textrm{gs}}$ is maximized, where
\begin{equation}
   {\cal K}_{\textrm{gs}} \equiv \pi \mu R_{b0}^{\star} - 2\pi R_{0p}
 \,, \qquad R_{b0}^{\star} \equiv \min_{\kb} R_{b0}(\kb) \,, \qquad
   \mu \equiv \frac{2\pi D_0}{|\Omega|} \,.
  \label{gs:kap}
\end{equation}
 For $\nu={-1/\log\eps} \ll 1$, a two-term asymptotic expansion for
 the competition instability threshold of $A$ on the optimal lattice is
\begin{equation}
  A_\textrm{optim}= \eps\sqrt{ \frac{2\pi}{|\Omega|}} {\cal A}_\textrm{optim}
 \,, \qquad
  {\cal A}_\textrm{optim} \sim \frac{\sqrt{b}(2+\mu)}{\sqrt{\mu}} +
  \nu\sqrt{\frac{b}{\mu}} \left( -\max_{\Lambda} {\cal K}_{\textrm{gs}} +
  \frac{1}{b^2}\left(1-\frac{\mu}{2}\right) \int_{0}^{\infty}
  V_{1p} \rho \, d\rho \right) + \cdots \,, \label{gs:optim}
\end{equation}
 where $\max_{\Lambda}{\cal K}_{\textrm{gs}}$ is taken over all
  lattices $\Lambda$ having a common area $|\Omega|$ of the
  Wigner-Seitz cell. In (\ref{gs:optim}), $V_{1p}$ is the solution to
  (\ref{score:exp_3}), while $b=\int_{0}^{\infty} w^2\rho \, d\rho\approx
  4.93$, where $w(\rho)>0$ is the ground-state solution of $\Delta_\rho
  w - w + w^2=0$, and $\int_{0}^{\infty} V_{1p}\rho \, d\rho \approx 0.481$.
\end{result}

We remark that (\ref{gs:optim}) can also be derived through the more
lengthy but systematic approach given in \S~\ref{schnak} and \S~\ref{gm} of
first calculating the portion of the continuous spectrum that
satisfies $|\lambda|\leq {\mathcal O}(\nu)$ when ${\cal A}={\cal A}_0
+ {\mathcal O}(\nu)$.

The numerical method to compute ${\cal K}_{\textrm{gs}}$ is given in
\S \ref{sec:ewald}. In \S \ref{sec:ewald_0}, we show numerically that
within the class of oblique Bravais lattices, ${\cal K}_{\textrm{gs}}$
is maximized for a regular hexagonal lattice.  Thus, the minimal stability
threshold for the feed-rate $A$ occurs for this hexagonal lattice.

\setcounter{equation}{0}
\setcounter{section}{5}
\section{Numerical Computation of the Bloch Green's Function}\label{sec:ewald}

We seek a rapidly converging expansion for the Bloch Green's function
$G_{b0}$ satisfying (\ref{green:b0}) on the Wigner-Seitz cell $\Omega$
for the Bravais lattice $\Lambda$ of (\ref{lattice-def}).  It is the
regular part $R_{b0}$ of this Green's function that is needed in
Principal Results \ref{pr 3.3}, \ref{pr 4.3}, and \ref{pr 5.3}. Since
only one Green's function needs to be calculated numerically in this
section, for clarity of notation we remove its subscript. In \S
\ref{sec:ewald_0} we will revert to the notation of
\S~\ref{2:latt_gr}--\ref{simp} to determine the optimal lattice for
the stability thresholds in Principal Results \ref{pr 3.3}, \ref{pr 4.3},
and \ref{pr 5.3}.

Instead of computing the Bloch Green's function on $\Omega$, it is
computationally more convenient to equivalently compute the Bloch
Green's function $G\equiv G_{b0}$ on all of $\R^2$ that satisfies
\begin{equation}
\Delta G(\xb) = -\delta(\xb)\,; \qquad
 G(\xb + \pmb l) = e^{-i\pmb k\cdot\pmb l}\, G(\xb)\,,
 \quad \pmb l\in\Lambda \,, \label{Green-PDE}
\end{equation}
where ${\kb/(2\pi)}\in \Omega_B$.  The regular part $R(\xbo)\equiv
R_{b0}(\xbo)$ of this Bloch Green's function is defined
\begin{equation}
     R(\xbo)\equiv \lim_{\xb\to \xbo} \left( G(\xb) + \frac{1}{2\pi} \log|\xb|
   \right) \,. \label{ew:r}
\end{equation}
To derive a computationally tractable expression for $R(\xbo)$ we will
follow closely the methodology of \cite{Beyl-1}.

We construct the solution to (\ref{Green-PDE}) as the sum of free-space
Green's functions
 \begin{equation}\label{Green-Lattice}
 G(\xb) = \sum_{\pmb l\in\Lambda} G_{\rm free}(\xb+\pmb l)\,
  e^{i\pmb k\cdot \pmb l}\,.
 \end{equation}
This sum guarantees that the quasi-periodicity condition in
(\ref{Green-PDE}) is satisfied. That is,
if $G(\xb) = \sum_{\pmb l\in\Lambda}G_{\rm free}(\xb + \pmb l)
\,e^{i\kb\cdot\pmb l}\,,$ then, upon choosing any $\pmb l^{\star}\in\Lambda$, it
follows that $G(\xb + \pmb l^{\star}) = e^{-i\pmb k\cdot \pmb l^{\star}}\,G(\xb)\,.$
To show this, we use $\pmb \l^{\star}+\pmb \l \in\Lambda$ and calculate
 \begin{equation*}
 G(\xb + \pmb l^*) = \sum_{\pmb l\in \Lambda}
  G_{\rm free}(\xb + \pmb l^* + \pmb l)\,e^{i\pmb k\cdot\pmb l}
  =\sum_{\pmb l\in \Lambda}
  G_{\rm free}(\xb + \pmb l^* + \pmb l)
  \,e^{i\pmb k\cdot(\pmb l^*+\pmb l)} \,e^{-i\pmb k\cdot\pmb l^*}
  = e^{-i\pmb k\cdot\pmb l^*}\,G(\xb)\,.
 \end{equation*}

In order to analyze (\ref{Green-Lattice}), we will
use the Poisson summation formula with converts a sum over $\Lambda$ to
a sum over the reciprocal lattice $\Lambda^{\star}$ of (\ref{lattice-recip}).
In the notation of \cite{Beyl-1}, we have (see Proposition 2.1 of
\cite{Beyl-1})
 \begin{equation}\label{Poisson-sum}
   \sum_{\pmb l\in\Lambda} f(\xb + \pmb l)\,e^{i\pmb k\cdot\pmb l}
  = \frac1 V\sum_{\pmb d\in\Lambda^*} \hat f(2\pi\pmb d - \pmb k)
 \,e^{i\xb\cdot(2\pi \pmb d-\pmb k)}\,,
 \quad  \xb\,, \pmb k\in\R^2 \,,
 \end{equation}
where $\hat f$ is the Fourier transform of $f$, and $V=|\Omega|$ is the
area of the primitive cell of the lattice.

\begin{remark}\label{remark 3} Other
  authors (cf.~\cite{Linton}, \cite{Moroz}) define the reciprocal
  lattice as $\Lambda^{\star}=\left\{ 2\pi m\, \pmb d_1, 2\pi n\,\pmb
  d_2\right\}_{m,n\in\Z}$, so that for any $\pmb l\in\Lambda$ and
  $\pmb d \in \Lambda^*$, it follows that $\pmb l\cdot\pmb d = 2K\pi$
  for some integer $K$ and hence $e^{i\pmb l\cdot\pmb d} =1\,.$ The
  form of the Poisson summation formula will then differ slightly from
  \eqref{Poisson-sum}.
\end{remark}

By applying (\ref{Poisson-sum}) to (\ref{Green-Lattice}), it follows
that the sum over the reciprocal lattice consists of free-space
Green's functions in the Fourier domain, and we will split each
Green's function in the Fourier domain into two parts in order to
obtain a rapidly converging series. In $\R^2$, we write the Fourier
transform pair as
 \begin{equation}\label{invF-trans-def}
  \hat f(\pmb p) = \int_{\R^2} f(\xb)\,e^{-i\xb\cdot\pmb p}\, {\rm
d}\xb\,, \qquad   f(\xb) = \frac{1}{4\pi^2}\int_{\R^2}
 \hat f(\pmb p) \,e^{i\pmb p\cdot\xb}\, {\rm d}\pmb p\,.
\end{equation}

The free space Green's function satisfies $\Delta G_{\rm free}=-\delta(\xb)$.
By taking Fourier transforms, we get
$-|\pmb p|^2\,\hat G_{\rm free}(\pmb p) = -1$, so that
\begin{equation}\label{G-hat}
  \hat G_{\rm free}(\pmb p)  = \frac{1}{|\pmb p|^2} \,.
\end{equation}
With the right-hand side of the Poisson summation formula \eqref{Poisson-sum}
in mind, we write
\begin{equation}
  \frac{1}{V}\sum_{\pmb d\in\Lambda^*}
  \hat G_{\rm free}(2\pi \pmb d - \pmb k)\,e^{i\xb\cdot(2\pi \pmb d-\pmb k)}
 = \sum_{\pmb d\in\Lambda^*}\frac{e^{i\xb\cdot(2\pi \pmb d-\pmb k)}}
 {|2\pi \pmb d - \pmb k|^2}\,,
\end{equation}
since $V=1$. To obtain a rapidly converging series expansion, we introduce
the decomposition
\begin{equation}
 \hat G_{\rm free}(2\pi\pmb d-\pmb k)
  = \alpha(2\pi\pmb d-\pmb k,\eta)\,
  \hat G_{\rm free}(2\pi\pmb d-\pmb k)
+ \Bigl(1-\alpha(2\pi\pmb d-\pmb k,\eta)\Bigr)\,
 \hat G_{\rm free}(2\pi\pmb d-\pmb
k)\,, \label{g:split}
\end{equation}
for some function $\alpha(2\pi\pmb d-\pmb k,\eta)$. We choose
$\alpha(2\pi\pmb d-\pmb k,\eta)$ so that the sum over $\pmb
d\in\Lambda^*$ of the first set of terms converges absolutely.  We
apply \eqref{invF-trans-def} to the second set of terms after first
writing $(1-\alpha)\,\hat G_{\rm free}$ as an integral. In the decomposition
(\ref{g:split}) we choose the function $\alpha$ as
 \begin{equation}\label{alpha-try}
 \alpha(2\pi\pmb d-\pmb k,\eta)
 = \exp\!\left(-\frac{|2\pi\pmb d-\pmb k|^2}{4\eta^2}\right)\,,
 \end{equation}
where $\eta>0$ is a cutoff parameter to be chosen. We readily observe that
 \begin{equation*}
      \lim_{\eta\to 0}\alpha(2\pi\pmb d-\pmb k,\eta) = 0\,; \qquad
      \lim_{\eta\to \infty}\alpha(2\pi\pmb d-\pmb k,\eta) = 1\,; \qquad
   \frac{\partial \alpha}{\partial\eta}
     = \frac{|2\pi\pmb d-\pmb k|^2\,\alpha}{2\eta^3} >0\,,\quad
  \mbox{since } \alpha>0\,,\,\,\, \eta>0\,,
\end{equation*}
which shows that $0<\alpha<1$ when $0<\eta<\infty$. Since
$0<\alpha<1$, the choice of $\eta$ determines the portion of the
Green's function that is determined from the sum of terms in the
reciprocal lattice $\Lambda^*$ and the portion that is determined from
the sum of terms in the lattice $\Lambda$.

 With the expressions \eqref{alpha-try} for $\alpha$
 and \eqref{G-hat} for $\hat G_{\rm free}$, we get
\begin{equation}
 \alpha(2\pi\pmb d-\pmb k,\eta)\,
  \hat G_{\rm free}(2\pi\pmb d-\pmb k)\,e^{i\xb\cdot(2\pi \pmb d-\pmb k)}
 =  \exp\!\left(-\frac{|2\pi\pmb d-\pmb k|^2}{4\eta^2}\right)\,
  \frac{
  \,e^{i\xb\cdot(2\pi \pmb d-\pmb k)}
  }
 {|2\pi \pmb d - \pmb k|^2}\,.  \label{g:split_1}
\end{equation}

Since $2\pi \pmb d - \pmb k\ne 0$, which follows since
${\kb/(2\pi)}\in\Omega_B$, the sum of these terms over $\pmb d\in
\Lambda^*$ converges absolutely. Following \cite{Beyl-1}, we
define
 \begin{equation}\label{G-fourier}
G_{\rm fourier}(\xb) \equiv \sum_{\pmb d\in\Lambda^*}
  \exp\!\left(-\frac{|2\pi\pmb d-\pmb k|^2}{4\eta^2}\right)\,
  \frac{\,e^{i\xb\cdot(2\pi \pmb d-\pmb k)} } {|2\pi \pmb d - \pmb k|^2}\,.
\end{equation}
For the $(1-\alpha)\,\hat G_{\rm free}$ term, we define $\rho$ by
$\rho\equiv |2\pi\pmb d-\pmb k|$, so that from \eqref{alpha-try},
\eqref{G-hat}, and $\hat{G}_{\rm free}=\hat{G}_{\rm free}(|\pmb p|)$, we get
\begin{equation}
\left(1- \alpha(2\pi\pmb d-\pmb k,\eta)\right)\,
  \hat G_{\rm free}(2\pi\pmb d-\pmb k)
  =\frac{1}{\rho^2}
  \left(1-e^{-{\rho^2/(4\eta^2)}}\right)\,. \label{g:split_2}
\end{equation}
Since $\int e^{-\rho^2\,e^{2s}+2s}\, ds =-{e^{-\rho^2\,e^{2s}}/(2\rho^2)}$, the
right hand side of (\ref{g:split_2}) can be calculated as
 \begin{equation*}
  2\int_{-\infty}^{-\log(2\eta)}\!\!
  e^{-\rho^2\,e^{2s}+2s}\, ds
 = \frac{1}{\rho^2}\,
  \Bigl(1 - e^{-{\rho^2/(4\eta^2)}}\Bigr)\,,
 \end{equation*}
so that
 \begin{equation}
\left(1- \alpha(2\pi\pmb d-\pmb k,\eta)\right)
 G_{\rm free}(2\pi\pmb d-\pmb k) =  2
 \int^{\infty}_{\log(2\eta)} e^{-\rho^2\,e^{-2s}-2s}\, ds\,. \label{g:fsing}
 \end{equation}

To take the inverse Fourier transform of (\ref{g:fsing}), we recall that
the inverse Fourier transform of a radially symmetric function is the
inverse Hankel transform of order zero (cf.~\cite{P}), so that
$f(r)=(2\pi)^{-1}\int_0^\infty \hat f(\rho)\,J_0(\rho r)\,\rho \, d\rho$.
Upon using the well-known inverse Hankel transform (cf.~\cite{P})
\begin{equation*}
\int_0^\infty e^{-\rho^2\,e^{-2s}}\rho\,J_0(\rho r)\,{\rm d}\rho
 = \frac12\,e^{2s-{r^2\,e^{2s}/4}}\,,
\end{equation*}
we calculate the inverse Fourier transform of (\ref{g:fsing}) as
 \begin{align*}
 \frac{1}{2\pi}\int_0^\infty
    \!\!\int^{\infty}_{\log(2\eta)}\!\! 2\,
  e^{-\rho^2 e^{-2s}-2s}\,
  \rho\, J_0(\rho r)\, {\rm d}s\, d\rho
 &=\frac{1}{\pi}
 \!\!\int^{\infty}_{\log(2\eta)}\!\! e^{-2s} \left(
 \int_0^\infty
   e^{-\rho^2\,e^{-2s}}\,
   \rho\, J_0(\rho r) \,{\rm d}\rho\, \right) ds
 \cr
 &=\frac{1}{2\pi}
 \!\!\int^{\infty}_{\log(2\eta)}\!\! e^{-2s}
 \,e^{2s-\frac{r^2}{4}\,e^{2s}}\, ds = \frac{1}{2\pi}
 \!\!\int^{\infty}_{\log(2\eta))}\!\!
 \,e^{-\frac{r^2}{4}\,e^{2s}}\, ds \,.
 \end{align*}
In the notation of \cite{Beyl-1}, we then define $F_{\rm sing}(\xb)$ as
\begin{equation}\label{F-sing}
 F_{\rm sing}(\xb)\equiv
 \frac{1}{2\pi}
 \!\!\int^{\infty}_{\log(2\eta)}\!\!
 \,e^{-\frac{|\xb|^2}{4}\,e^{2s}}\, ds \,,
 \end{equation}
so that by the Poisson summation formula (\ref{Poisson-sum}), we have
 \begin{equation}\label{G-spatial}
 G_{\rm spatial}(\xb) \equiv \sum_{\pmb l\in\Lambda}
  e^{i\pmb k\cdot\pmb l}\,F_{\rm sing}(\xb+\pmb l)\,.
 \end{equation}

In this way, for ${\kb/(2\pi)}\in \Omega_B$, we write the Bloch
Green's function in the spatial domain as the sum of \eqref{G-fourier}
and \eqref{G-spatial}
 \begin{equation}
G(\xb) = \sum_{\pmb d\in\Lambda^*}
  \exp\left(-\frac{|2\pi\pmb d-\pmb k|^2}{4\eta^2}\right)\,
  \frac{ e^{i\xb\cdot(2\pi \pmb d-\pmb k)}  }
 {|2\pi \pmb d - \pmb k|^2}
+ \frac1{2\pi} \sum_{\pmb l \in\Lambda}e^{i\pmb k\cdot\pmb l}
 \int^{\infty}_{\log(2\eta)}\!\!
 \,e^{-\frac{|\xb+\pmb l|^2}{4}\,e^{2s}}\, ds
\,.   \label{g:full_exp}
\end{equation}
From \eqref{G-fourier} and \eqref{G-spatial}, it readily follows that
$G_{\rm Fourier} \to 0$ as $\eta\to 0$, while
$G_{\rm spatial} \to 0$ as $\eta\to \infty$.

Now consider the behaviour of the Bloch Green's function as $\xb \to \xbo$.
From \eqref{G-fourier}, we have
 \begin{equation}
   G_{\rm Fourier}(0)
   = \sum_{\pmb d\in\Lambda^*}
  \exp\!\left(-\frac{|2\pi\pmb d-\pmb k|^2}{4\eta^2}\right)\,
  \frac{1   }
 {|2\pi \pmb d - \pmb k|^2}\,, \qquad \mbox{for} \quad {\kb/(2\pi)}\in
  \Omega_B\,, \label{g:spat0}
 \end{equation}
 which is finite since $|2\pi\pmb d -\pmb k|\ne 0$ and $\eta<\infty$. It
is also real-valued. Next, we decompose  $G_{\rm spatial}$ in (\ref{G-spatial}) as
 \begin{equation}
 G_{\rm spatial}(\xb) =
   F_{\rm sing}(\xb) +  \sum_{\genfrac{}{}{0pt}{}{\pmb l\in\Lambda}{\pmb
l\ne \xbo}}e^{i\pmb k\cdot \pmb l}\,F_{\rm sing}(\xb+\pmb l)\,. \label{gsp:dec}
 \end{equation}
For the second term in (\ref{gsp:dec}), we can take the limit $\xb\to\xbo$
since from \eqref{F-sing} we have
 \begin{equation*}
 \Big{\vert} \sum_{\genfrac{}{}{0pt}{}{\pmb l\in\Lambda}{\pmb
l\ne 0}}e^{i\pmb k\cdot \pmb l}\,F_{\rm sing}(\pmb l) \Big{\vert}
 < \infty\,.
 \end{equation*}
In contrast, $F_{\rm sing}(\xb)$ is singular at $\xb = \xbo$. To calculate
its singular behavior as $\xb\to \xbo$, we write $F_{\rm sing}(\xb) =
F_{\rm sing}(r)$, with $r=|\xb|$, and convert $F_{\rm sing}(r)$ to
an exponential integral by introducing $u$ by $u={r^2 e^{2s}/4}$
in (\ref{F-sing}). This gives
\begin{equation}\label{F-sing-E1}
 F_{\rm sing}(r)
  = \frac1{2\pi}\int_{\log(2\eta)}^\infty
    e^{-\frac{r^2}{4}\,e^{2s}}\, ds
  = \frac1{4\pi}\int_{r^2\,\eta^2}^\infty
    \frac{e^{-u}}u\, du
  = \frac1{4\pi}\,E_1(r^2\eta^2)\,,
\end{equation}
where $E_{1}(z)=\int_{z}^{\infty} t^{-1} e^{-t}\, dt$ is the
exponential integral (cf.~\S5.1.1 of \cite{AS}). Upon using the
series expansion of $E_{1}(z)$
 \begin{equation}\label{AS-E1-exp}
  E_1(z) = -\gamma - \log(z) - \sum_{n=1}^\infty \frac{(-1)^n\,
   z^n}{n\,n!}\,, \qquad \mbox{for} \quad |\arg z|<\pi \,,
 \end{equation}
as given in \S 5.1.11 of \cite{AS}, where $\gamma=0.57721\cdots$ is
Euler's constant, we have from \eqref{F-sing-E1} and \eqref{AS-E1-exp} that
 \begin{equation} \label{F-sing-asymp}
 F_{\rm sing}(r) \sim -\frac{\gamma}{4\pi}
   - \frac{\log\eta}{2\pi} - \frac{\log r}{2\pi} + o(1)\,,
 \qquad \text{as}\quad r\to 0\,.
 \end{equation}
This shows that the Bloch Green's function in (\ref{g:full_exp}) has the
expected logarithmic singularity as $\xb\to \xbo$.

We write the Bloch Green's function as the sum of regular and singular
parts as
\begin{equation}
 G(\xb) =    -\frac1{2\pi}\,\log |\xb| + R(\xb) \,, \qquad
 R(\xb) = G_{\rm Fourier}(\xb) + G_{\rm Spatial}(\xb)
        +  \frac1{2\pi}\,\log |\xb|\,. \label{gb:rb}
\end{equation}
By letting $\xb\to \xbo$, we have from (\ref{gsp:dec}),
(\ref{F-sing-asymp}), (\ref{g:spat0}), and (\ref{gb:rb}), that
for ${\kb/(2\pi)}\in \Omega_B$
\begin{equation}
  R(\xbo) = \sum_{\pmb d\in\Lambda^*}
  \exp\!\left(-\frac{|2\pi\pmb d-\pmb k|^2}{4\eta^2}\right)\,
  \frac{1   }
 {|2\pi \pmb d - \pmb k|^2} + \sum_{\genfrac{}{}{0pt}{}{\pmb l\in\Lambda}{\pmb
l\ne \xbo}}e^{i\pmb k\cdot \pmb l}\,F_{\rm sing}(\pmb l) - \frac{\gamma}{4\pi}
 - \frac{\log\eta}{2\pi} \,, \label{g:rb0_key}
\end{equation}
where $F_{\rm sing}(\pmb l)={E_{1}(|\pmb l|^2\eta^2)/(4\pi)}$.

For a square lattice, with unit area of the primitive cell and with
$\eta=2$ and $\pmb k=(\sin\frac\pi 3,\cos\frac\pi 3)$, in Table
\ref{tab:rb} we give numerical results for $R(\xb)$ for various values
of $\xb$ as $\xb\to\xbo$. The computations show that
$\im\left(R(\xb)\right)\to 0$ as $\xb\to \xbo$, as expected from
Lemma \ref{lemma 2.1} of \S~\ref{2:gr_lattice}.

\begin{table}
\begin{center}
\begin{tabular}{|c|c|c}
\hline $\xb$ & $G(\xb)$ & $R(\xb)$ \\
\hline
(.1,.1)           &  1.1027-.12568\,i &  .79138-.12568\,i \\
(.01,01)          &  1.4730-.012593\,i &  .79526-.012593\,i \\
$(10^{-3},10^{-3})$ &  1.8396-.0012593\,i &  .79531-.0012593\,i \\
$(10^{-4},10^{-4})$ &  2.2060-.00012593\,i &  .79530-.00012593\,i \\
$(10^{-5},10^{-5})$ &  2.5725-.000012593\,i &  .79529-.000012593\,i\\
$(10^{-6},10^{-6})$ &  2.9389-.0000012593\,i &  .79531-.0000012593\,i \\
$(10^{-7},10^{-7})$ &  3.3054-.00000012593\,i &  .79530-.00000012593\,i \\
$(10^{-8},10^{-8})$ &  3.6719-.000000012593\,i &  .79531-.000000012593\,i \\
$(10^{-9},10^{-9})$ &  4.0383-.0000000012593\,i &  .79529-.0000000012593\,i\\
$(10^{-10},10^{-10})$ &  4.4048-.00000000012593\,i &  .79530-.00000000012593\,i\\
$(10^{-11},10^{-11})$ &  4.7713-.000000000012594\,i &  .79529-.000000000012594\,i
\\
\hline
\end{tabular}
\end{center}
\caption{The regular part $R(\xb)$ of the Bloch Green's function, as
defined in (\ref{gb:rb}), for $\xb$ tending to the origin. Notice that
the imaginary part of $R(\xb)$ becomes increasingly small as $\xb\to \xbo$,
as expected from Lemma \ref{lemma 2.1} of \S~\ref{2:gr_lattice} where it was
established that $R(\xbo)$ is real-valued.} \label{tab:rb}
\end{table}

\subsection{An Optimal Lattice for Stability Thresholds}\label{sec:ewald_0}

In this sub-section we determine the lattice that optimizes the
stability thresholds given in Principal Results \ref{pr 3.3}, \ref{pr 4.3},
and \ref{pr 5.3}, for the Schnakenburg, GM, and GS models,
respectively.  Recall that in the notation of \S
\ref{2:latt_gr}--\ref{simp}, $R_{b0}(\kb)=R(\xbo)$, where $R(\xbo)$ is
given in (\ref{g:rb0_key}). The minimum of $R(0)$ with respect to
$\kb$ is denoted by $R_{b0}^{\star}$.

In our numerical computations of $R(\xbo)$ from (\ref{g:rb0_key}) we
truncate the direct and reciprocal lattices $\Lambda$ and
$\Lambda^{*}$ by the subsets $\bar{\Lambda}$ and $\bar{\Lambda}^*$ of
$\Lambda$ and $\Lambda^*$, respectively, defined by
$$\bar{\Lambda}=\left\{ n_1\pmb l_1+n_2 \pmb
l_2\,\big|-\!\!M_1<n_1,n_2<M_1\right\}\,,\quad \bar{\Lambda}^*=\left\{
n_1\pmb d_1+n_2 \pmb d_2\,\big|-\!\!M_2<n_1,n_2<M_2\right\}\,,\quad
n_1,n_2\in \mathbb{Z}\,.$$

For each lattice, we must pick $M_1$, $M_2$ and $\eta$ so that $G$ can
be calculated accurately with relatively few terms in the sum. These
parameters are found by numerical experimentation.  For the two
regular lattices (square, hexagonal) we used $(M_1,M_2,\eta) =
(2,5,3)$. For an arbitrary oblique lattice with angle $\theta$ between
$\lb_1$ and $\lb_2$ we took $M_1=5$, $M_2=3$, and we set $\eta=3$.

In Table \ref{tab:optim} we give numerical results for
$R_{b0}^{\star}$ for the square and hexagonal
lattices. These results show that $R_{b0}^{\star}$ is largest for the
hexagonal lattice. For these two simple lattices, in Table
\ref{tab:optim} we also give numerical results for $R_{0p}$, defined
by (\ref{gr:source_neut}), as obtained from the explicit formula in
Theorem 1 of \cite{chen} and \S 4 of \cite{chen}.  In Theorem 2 of
\cite{chen} it was proved that, within the class of oblique Bravais
lattices with unit area of the primitive cell, $R_{0p}$ is minimized
for a hexagonal lattice. Finally, in the fourth and fifth columns of
Table~\ref{tab:optim} we give numerical results for ${\cal
  K}_\textrm{s}$ and ${\cal K}_\textrm{gm}$, as defined in Principal
Results \ref{pr 3.3} and \ref{pr 4.3}. Of the two lattices, we
conclude that ${\cal K}_\textrm{s}$ and ${\cal K}_\textrm{gm}$ are
largest for the hexagonal lattice. In addition, since $R_{b0}^{\star}$
is maximized and $R_{0p}$ is minimized for a hexagonal lattice, it
follows that ${\cal K}_{\textrm{gs}}$ in Principal Result \ref{pr 5.3}
is also largest for a hexagonal lattice. Thus, with respect to the
two simple lattices, we conclude that the optimal stability
thresholds in Principal Results \ref{pr 3.3}, \ref{pr 4.3}, and
\ref{pr 5.3}, occur for a hexagonal lattice.

To show that the same conclusion regarding the optimal stability
thresholds occurs for the class of oblique lattices, we need only show
that $R_{b0}^{\star}$ is still maximized for the hexagonal lattice. This
is done numerically below.

\begin{table}
\begin{center}
\begin{tabular}{ |c|c|c|c|c|}
\hline
Lattice &$R_{b0}^{\star}$&$R_{0p}$&${\cal K}_{\textrm{s}}$&${\cal K}_{\textrm{gm}}$
\\ \hline
Square     & $-0.098259$ & $-0.20706$ & $-0.098259$ &  $0.06624$\\ \hline
Hexagonal  & $-0.079124$ & $-0.21027$ & $-0.079124$ &  $0.32685$\\ \hline
\end{tabular}
\end{center}
\caption{Numerical values for $R_{b0}^{\star}=\min_{\kb} R(\xbo)$,
  where $R(\xbo)$ is computed from (\ref{g:rb0_key}), for the square
  and hexagonal lattice for which $|\Omega|=1$. The third
  column is the regular part $R_{0p}$ of the periodic source-neutral Green's
  function (\ref{gr:source_neut}). The last two columns are ${\cal
    K}_{\textrm{s}}$ and ${\cal K}_{\textrm{gm}}$, as defined in
  Principal Results \ref{pr 3.3} and \ref{pr 4.3}, respectively. Of
  the two lattices, the hexagonal lattice gives the largest values
  for ${\cal K}_{\textrm{s}}$ and ${\cal K}_{\textrm{gm}}$.}
\label{tab:optim}
\end{table}

\begin{figure}[htb]
\begin{center}
{\includegraphics[width = 8cm,height=5.5cm,clip]{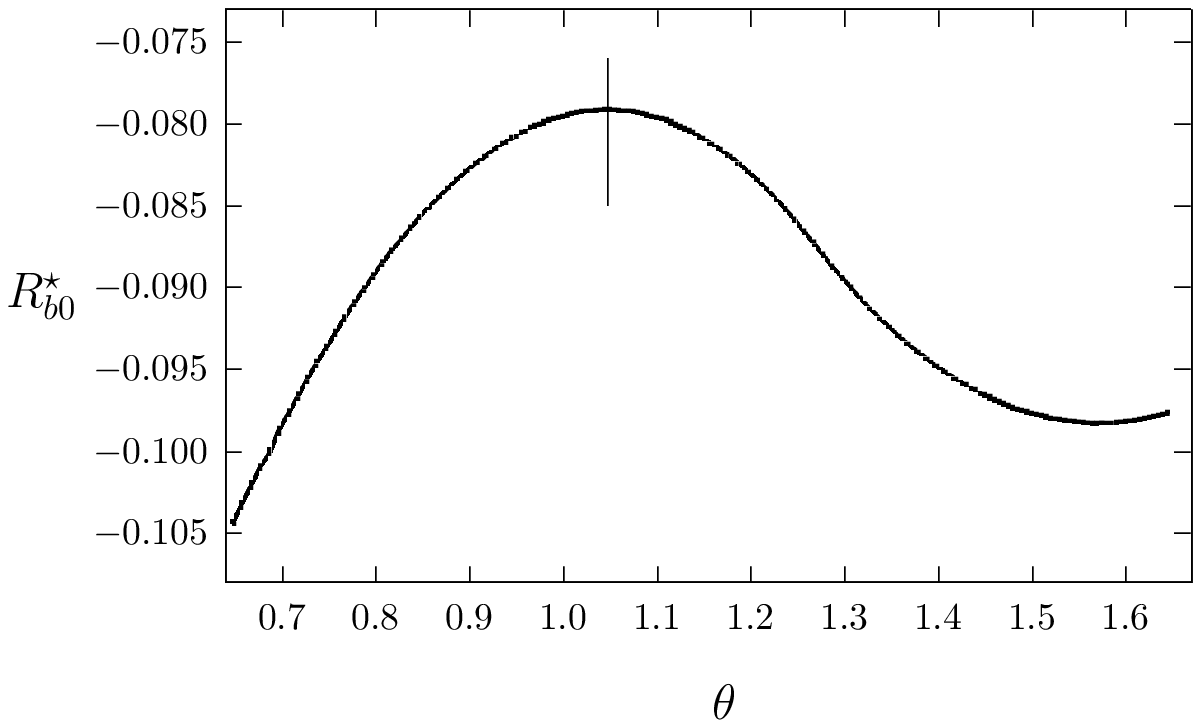}}
{\includegraphics[width = 8cm,height=5.5cm,clip]{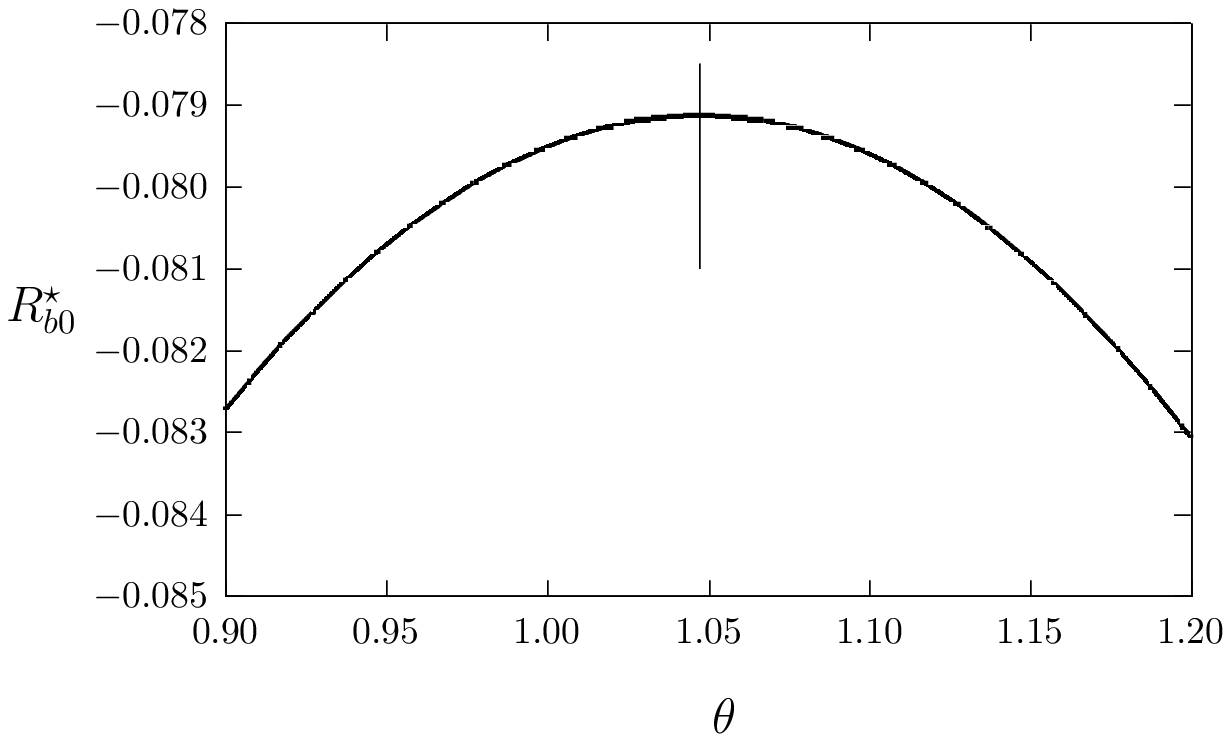}}
\caption{Minimum value $R_{b0}^{\star}$ of $R_{b0}(\kb)$ for all oblique
  lattices of unit area for which $\lb_1=(1/\sqrt{\sin(\theta)},0)$ and
 $\lb_2=(\cos(\theta)/\sqrt{\sin(\theta)},\sqrt{\sin(\theta)})$, so that
  $|\lb_1|=|\lb_2|$ and $|\Omega|=1$. The vertical line denotes
 the hexagonal lattice for which $\theta={\pi/3}$. Left figure:
 the angle $\theta$ between the lattice vectors ranges over
 $0.6<\theta<1.7$.  Right figure: enlargement of the left figure near
 $\theta={\pi/3}$. The vertical line again denotes the hexagonal lattice.}
\label{ang-rb0}
\end{center}
\end{figure}

\begin{figure}[htb]
\begin{center}
{\includegraphics[width = 10cm,height=5.5cm,clip]{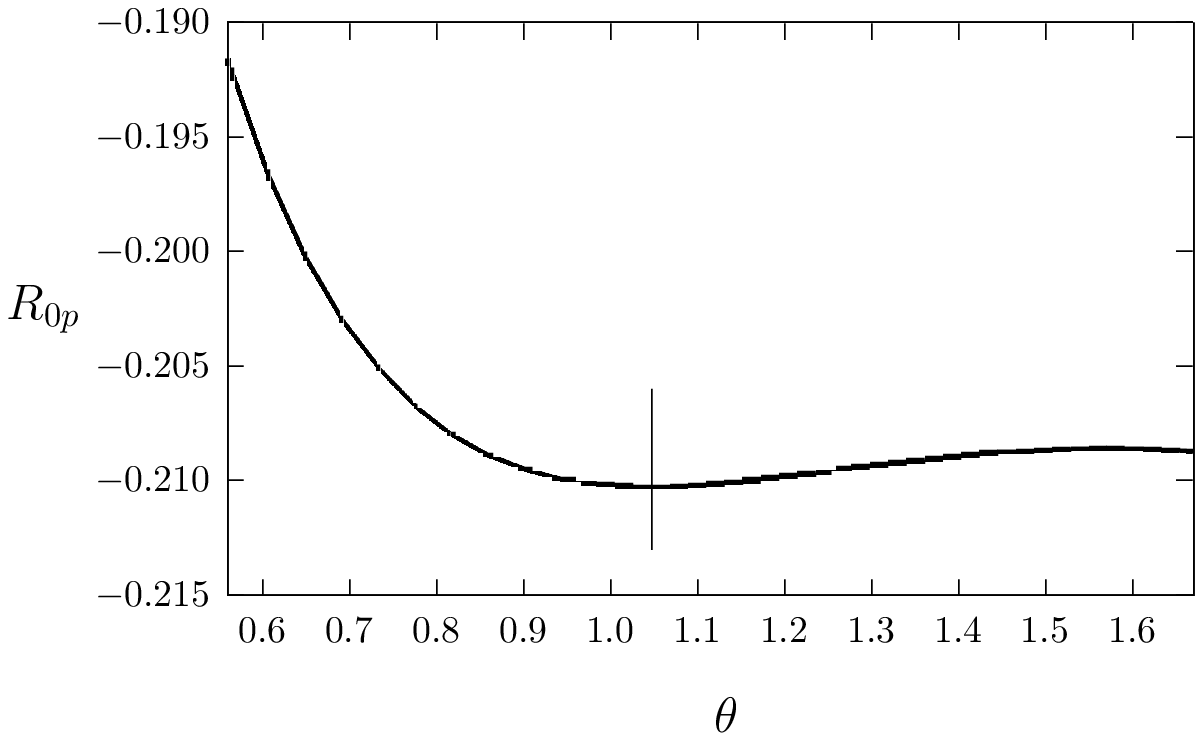}}
\caption{Plot of the the regular part $R_{0p}$, as given in
  (\ref{r0p:chen}) (cf.~\cite{chen}), of the periodic source-neutral
  Green's function for all oblique lattices of unit area for which
  $\lb_1=(1/\sqrt{\sin(\theta)},0)$ and
  $\lb_2=(\cos(\theta)/\sqrt{\sin(\theta)},\sqrt{\sin(\theta)})$, so
  that $|\lb_1|=|\lb_2|$ and $|\Omega|=1$. The vertical line denotes
  the hexagonal lattice for which $\theta={\pi/3}$. The minimum occurs for
  the hexagon.}
\label{fig:r0p}
\end{center}
\end{figure}

We first consider lattices for which $|\lb_1|=|\lb_2|$. For this
subclass of lattices, the lattice vectors are
$\lb_1=(1/\sqrt{\sin(\theta)},0)$ and
$\lb_2=(\cos(\theta)/\sqrt{\sin(\theta)},\sqrt{\sin(\theta)})$.  In
our computations, we first use a coarse grid to find an approximate
location in $\kb$-space of the minimum of $R(\xbo)$ and then we refine
the search.  After establishing by a coarse discretization that the
minimum arises near a vertex of the adjoint lattice, we then sample
more finely near this vertex.  The finest mesh has a resolution of
${\pi/100}$.  To determine the value of $R_{b0}^{\star}$ we
interpolate a paraboloid through the approximate minimum and the four
neighbouring points and evaluate the minimum of the paraboloid.  As we
vary the lattice by increasing $\theta$, we use the approximate
location of the previous minimum as an initial guess.  The value of
$\theta$ is increased by increments of $0.01$.  Our numerical results
in Fig.~\ref{ang-rb0} show that the optimum lattice where
$R_{b0}^{\star}\equiv \min_{\kb} R(0)$ is maximized occurs for the
hexagonal lattice where $\theta={\pi/3}$. In Fig.~\ref{fig:r0p} we
also plot $R_{0p}$ versus $\theta$ (cf.~Theorem 1 of \cite{chen}), given by
\begin{equation}
   R_{0p}=-\frac{1}{2\pi}\log(2\pi) -\frac{1}{2\pi} \ln\Big{\vert}
  \sqrt{ \sin\theta} \, e\left({\xi/12}\right) \prod_{n=1}^{\infty}
  \left(1-e(n\xi)\right)^2\Big{\vert} \,, \quad e(z)\equiv e^{2\pi iz}\,,
 \quad \xi=e^{i\theta} \,. \label{r0p:chen}
\end{equation}

Finally, we consider a more general sweep through the class of oblique
Bravais lattices. We let $\lb_1=(a,0)$ and $\lb_2=(b,c)$, so that with
unit area of the primitive cell, we have $ac=1$ and $b = a^{-1}
\cot\theta$, where $\theta$ is the angle between $\l_1$ and $\l_2$.
We introduce a parameter $\alpha$ by $a=\left(\sin\theta\right)^\alpha$ so that
  \begin{equation}
    c =  \left(\sin\theta\right)^{-\alpha}
    \qquad\mbox{and}\qquad  b = \cos\theta\,\left(\sin\theta\right)^{-\alpha-1}
 \,. \label{new_par}
  \end{equation}
Then $|\lb_1|=|\lb_2|$ when $\alpha=-1/2$, $|\lb_1|=1$ (which is
independent of $\theta$) when $\alpha=0$, and $|\lb_2|=1$ when
$\alpha=-1$.  In the left panel of Fig.~\ref{hex_alpha_1}, we plot
$R_{b0}^{\star}$ versus $\theta$ for $\alpha =-.5,-.4,-.3,-.2,-.1,0$.
The angle, $\theta$, at which the maximum occurs, increases from
$\pi/3$ at $\alpha=-.5$ to about $1.107=\pi/3+.06$ for
$\alpha=0$. However, the value of the maximum is largest for
$\alpha=-.5$ and decreases as $\alpha$ increases to zero.  The regular
hexagon occurs only at $\alpha=-.5$ and $\theta=\pi/3$.  The vertical
line in the plot is at $\theta=\pi/3$. Similarly, in the right panel
of Fig.~\ref{hex_alpha_1} we plot $R_{b0}^{\star}$ versus $\theta$ for
$\alpha = -.5,-.6,-.7,-.8,-.9,-1.0$.  Since there is no preferred
angular orientation for the lattice and since the scale is arbitrary,
the plot is identical to the previous plot, in the sense that the
curves for $\alpha=-0.6$ and $\alpha=-0.4$ in Fig.~\ref{hex_alpha_1}
are identical.  We conclude that it is the regular hexagon that
maximizes $R_{b0}^{\star}$.  These computational results lead to the
following conjecture:

\begin{figure}[htb]
\begin{center}
{\psfrag{Rb0}{$R_{b0}^{\star}$}
\psfrag{theta}{$\theta$}
\psfrag{a=0}{$\alpha=0$}
\psfrag{a=-.1}{$\alpha=-.1$}
\psfrag{a=-.2}{$\alpha=-.2$}
\psfrag{a=-.3}{$\alpha=-.3$}
\psfrag{a=-.4}{$\alpha=-.4$}
\psfrag{a=-.5}{$\alpha=-.5$}
\includegraphics[width = 5.5cm,height=8.0cm,clip,angle=-90]{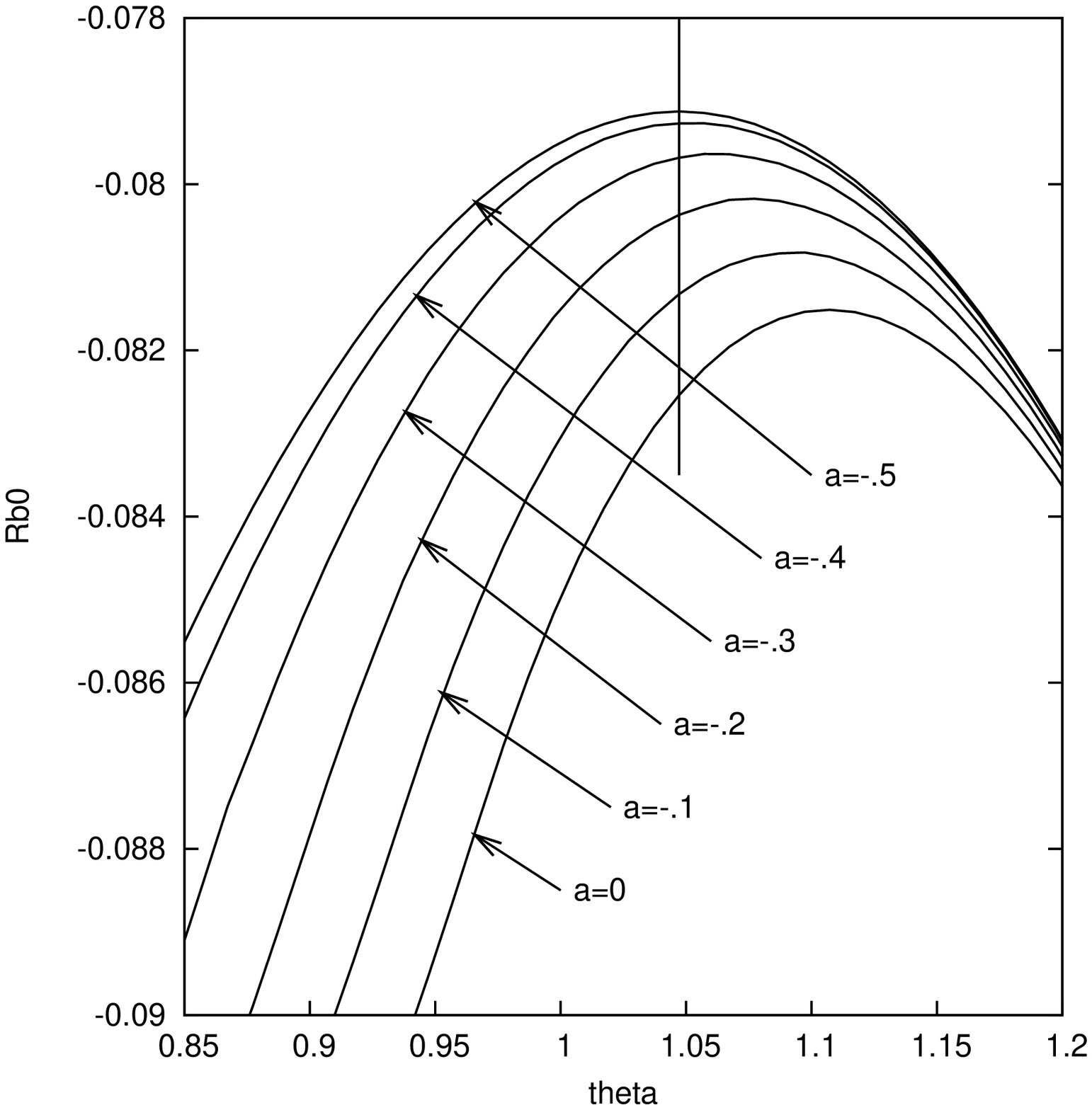}}
{\psfrag{Rb0}{$R_{b0}^{\star}$}
\psfrag{theta}{$\theta$}
\psfrag{a=0}{$\alpha=1$}
\psfrag{a=-.1}{$\alpha=-.9$}
\psfrag{a=-.2}{$\alpha=-.8$}
\psfrag{a=-.3}{$\alpha=-.7$}
\psfrag{a=-.4}{$\alpha=-.6$}
\psfrag{a=-.5}{$\alpha=-.5$}
\includegraphics[width = 5.5cm,height=8.0cm,clip,angle=-90]{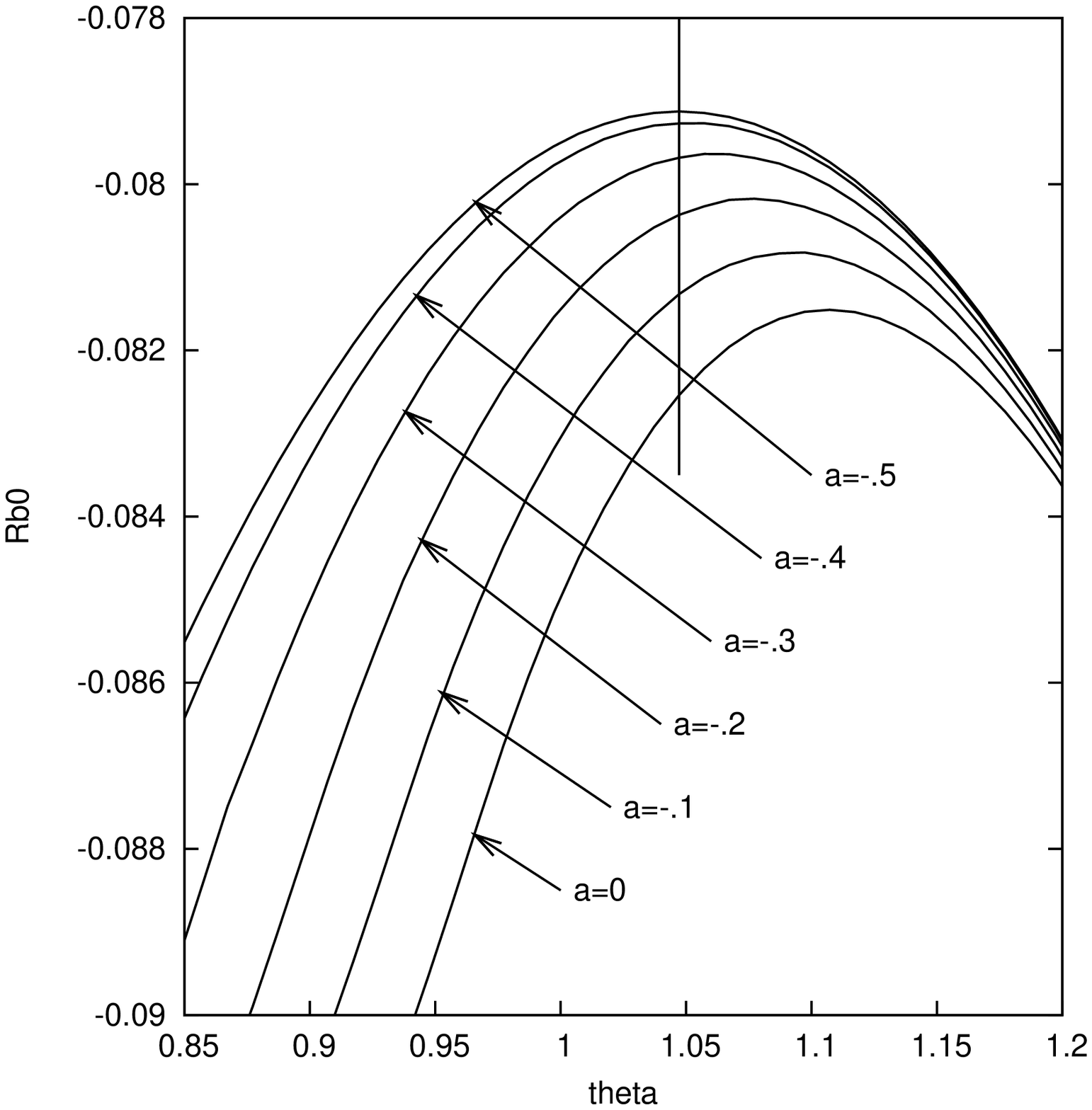}}
\caption{Plot of $R_{b0}^{\star}$ versus $\theta$ for oblique lattices with
$\lb_1=(a,0)$ and $\lb_2=(b,c)$, where $a=\left(\sin\theta\right)^\alpha$
with $b$ and $c$ given in (\ref{new_par}). Left panel: plots are for
$\alpha =-.5,-.4,-.3,-.2,-.1,0$. Right panel: plots are for
$\alpha =-.5,-.6,-.7,-.8,-.9,-1.0$.}
\label{hex_alpha_1}
\end{center}
\end{figure}

\begin{conjecture}\label{conj 1}  Within the class of Bravais
lattices of a common area, $R_{b0}^{\star}$ is maximized for a regular
hexagonal lattice.
\end{conjecture}

\setcounter{equation}{0} \setcounter{section}{6}
\section{Discussion}\label{sect:disc}

We have studied the linear stability of steady-state periodic patterns
of localized spots for the GM and Schnakenburg RD models when the
spots are centered for $\eps \to 0$ at the lattice points of a Bravais
lattice with constant area $|\Omega|$.  To leading order in
$\nu={-1/\log\eps}$, the linearization of the steady-state periodic
spot pattern has a zero eigenvalue when $D={D_0/\nu}$ for some $D_0$
independent of the lattice and the Bloch wavevector $\kb$. The
critical value $D_0$ can be identified from the leading-order NLEP
theory of \cite{WGM1} and \cite{survey_Wei:2008}.  This zero
eigenvalue corresponds to a competition instability of the spot
amplitudes (cf.~\cite{WGM1}, \cite{KWW_schnak}, \cite{cw_1}, and
\cite{survey_Wei:2008}). By using a combination of the method of
matched asymptotic expansions, Floquet-Bloch theory, and the rigorous
imposition of solvability conditions for perturbations of certain
nonlocal eigenvalue problems, we have explicitly determined the
continuous band of spectrum that lies within an ${\mathcal O}(\nu)$
neighborhood of the origin in the spectral plane when $D={D_0/\nu} +
D_1$, where $D_1={\mathcal O}(1)$ is a de-tuning parameter.  This
continuous band is real-valued, and depends on the regular part of the
Bloch Green's function and $D_1$. In this way, for each RD model, we
have derived a specific objective function that must be maximized in
order to determine the specific periodic arrangement of localized
spots that is linearly stable for the largest value of $D$. A simple
alternative method to derive this objective function was also given
and applied to the GS model.  From a numerical computation, based on
an Ewald-type algorithm, of the regular part of the Bloch Green's function
that defines the objective function, we have shown within the class of
oblique Bravais lattices that a hexagonal lattice arrangement of spots
is the most stable to competition instabilities.

Although we have focused our analysis only on the Schnakenburg, GM,
and GS models, our asymptotic methodology to derive the
model-dependent objective function that determines the optimally
stable lattice arrangement of spots is readily extended to general RD
systems in the semi-strong interaction regime, such as the Brusselator
RD model (cf.~\cite{RRW}). Either the simple method of \S \ref{simp},
or the more elaborate but systematic method of \S~\ref{schnak} and
\S~\ref{gm}, can then be used to derive the objective function.

There are a few open problems that warrant further investigation. One
central issue is to place our formal asymptotic theory on a more
rigorous footing.  In this direction, it is an open problem to
rigorously characterize the continuous band of spectrum that lies near
the origin when $D$ is near the critical value. In addition, is it
possible to analytically prove Conjecture \ref{conj 1} that, within
the class of oblique Bravais lattices of a common area,
$R_{b0}^{\star}$ is maximized for a hexagonal lattice?

As possible extensions to this work, it would be interesting to
characterize lattice arrangements of spots that maximize the Hopf
bifurcation threshold in $\tau$. To analyze this problem, one would
have to calculate any continuous band of spectra that lies within an
${\mathcal O}(\nu)$ neighborhood of the Hopf bifurcation frequency
$\lambda=i\lambda_{I0}$ when $\tau-\tau_I\ll 1$, where $\tau_I$ and
$\lambda_{I0}$ is the Hopf bifurcation threshold and frequency,
respectively, on the Wigner-Seitz cell.

We remark that we have not analyzed any weak instabilities due to
eigenvalues of order $\lambda={\mathcal O}(\eps^2)$ associated with
the translation modes. It would be interesting to determine
steady-state lattice arrangements of localized spots that optimize the
linear stability properties of these modes. For these translation
modes we might expect, in contrast to what we found in this paper for
competition instabilities (see Remark \ref{remark 1} and Lemma
\ref{lemma 2.2}), that it is the long-wavelength instabilities with
$|\kb|\ll 1$ that destabilize the pattern. Long-wavelength
instabilities have been shown to be the destabilizing mechanism for
periodic solutions on 3-D Bravais lattices of two-component RD systems
in the weakly nonlinear Turing regime (cf.~\cite{call_1},
\cite{call_2}).

Finally, it would be interesting to examine the linear stability properties
of a collection of $N\gg 1$ regularly-spaced localized spots on a
large but finite domain with Neumann boundary conditions, and to
compare the spectral properties of this finite domain problem with that of
the periodic problem in $\R^2$. For the finite domain problem, we
expect that there are $N$ discrete eigenvalues (counting multiplicity)
that are asymptotically close to the origin in the spectral plane when
$D$ is close to a critical threshold. Research in this direction is in
progress.

\vspace*{-0.2cm}
\section*{Acknowledgements}
D.~I. and M.~J.~W. were supported by NSERC (Canada). Prof. Juncheng
Wei was partially supported by an Earmarked Grant from RGC of Hong
Kong and by NSERC (Canada). M.J.W.~is grateful to Prof.~Edgar Knobloch
(U.C. Berkeley) for his comments regarding the de-stabilizing
mechanisms of periodic weakly-nonlinear Turing patterns on
lattices.

\vspace*{-0.2cm}
\appendix
\newcommand{\newsection}[1]{{\setcounter{equation}{0}}\section{#1}}
\renewcommand{\theequation}{\Alph{section}.\arabic{equation}}
\newsection{Schnakenburg Model: Expansion of the Core Problem}\label{app:schnak}

We outline the derivation of the results of Lemma \ref{lemma 3.1}, as
given in \S 6 of \cite{KWW_schnak}, and those of Lemma \ref{lemma 5.1}.
To motivate the appropriate scaling for solutions $U$, $V$, and $\chi$
to (\ref{eq1:core}) for $S\to 0$. Upon writing $U={\cal U} S^{-p}$,
$V={\cal V} S^{p}$, where ${\cal U}$ and ${\cal V}$ are ${\mathcal O}(1)$
as $S\to 0$, we obtain that the $V$-equation in (\ref{eq1:core}) is
unchanged, but that the $U$ equation becomes
\begin{equation*}
  \Delta_\rho {\cal U} = S^{2p}\,  {\cal U} {\cal V}^2 \,; \qquad
  {\cal U} \sim S^{1+p}\log\rho + S^{p} \chi \, \quad \mbox{as}\quad
 \rho\to \infty \,.
\end{equation*}
From equating powers of $S$ after first applying the divergence theorem,
we obtain that $2p=p+1$, which yields $p=1$. Then,
to ensure that ${\cal U}={\mathcal O}(1)$, we must have
$\chi={\mathcal O}(S^{-p})$. This shows that if $S=S_0\nu^{1/2}$ where
$\nu\ll 1$, the appropriate scalings are $V={\mathcal O}(\nu^{1/2})$,
$U={\mathcal O}(\nu^{-1/2})$, and $\chi={\mathcal O}(\nu^{-1/2})$.

With this basic scaling, we then proceed to calculate higher order
terms in the expansion of the solution to the core problem by writing
$S=S_0\nu^{1/2} + S_1\nu^{3/2} + \cdots$ and then determining the first two
terms in the asymptotic solution $U$, $V$, and $\chi$ to
(\ref{eq1:core}) in terms of $S_0$ and $S_1$. The appropriate
expansion for these quantities is (see (6.2) of\cite{KWW_schnak})
\begin{equation}
 V\sim \nu^{1/2}\left(V_{0} + \nu V_1 + \cdots\right) \,, \qquad
    \left(\chi\,, U \right) =  \nu^{-1/2}\left[ \left(\chi_{0}\,, U_{0}
   \right) + \nu \left(\chi_{1}\,, U_{1} \right) + \cdots \right] \,.
  \label{nlep:expan}
\end{equation}
Upon substituting (\ref{nlep:expan}) into (\ref{eq1:core}), and collecting
powers of $\nu$, we obtain that $U_{0}$ and $V_{0}$ satisfy
\begin{equation} \label{nlep:eq1}
\begin{split}
 \Delta_\rho V_0  - V_{0} + U_{0} V_{0}^2
 &= 0 \,; \qquad  \Delta_\rho U_0 =0 \,,
 \qquad 0\leq \rho<\infty \,,
  \\ V_{0} \to 0 \,, \qquad U_{0} &\to \chi_{0}
 \quad \mbox{as} \quad \rho\to \infty \,; \qquad
V_0^{\prime}(0)=U_0^{\prime}(0)=0
 \,,
\end{split}
\end{equation}
where $\Delta_\rho V_0\equiv V_{0}^{\p\p}+\rho^{-1}V_0^{\p}$.  At next
order, $U_{1}$ and $V_{1}$ satisfy
\begin{equation}
\label{nlep:eq2}
\begin{split}
  \Delta_\rho V_1 - V_{1} +
 2 U_{0} V_{0} V_{1} &= -U_{1} V_{0}^2\,;
   \qquad \Delta_\rho U_1 = U_{0} V_{0}^2 \,,
 \qquad 0\leq \rho<\infty \,,   \\
   V_{1} \to 0 \,, \quad & \quad U_{1}\to S_{0}\log\rho + \chi_{1}
  \quad \mbox{as} \quad \rho\to \infty \,;
\qquad V_1^{\prime}(0)=U_1^{\prime}(0)=0 \,.
\end{split}
\end{equation}
Then, at one higher order, we get that $U_{2}$ satisfies
\begin{equation}
  \Delta_\rho U_2 = U_{1} V_{0}^2 + 2 U_{0} V_{0} V_{1} \,, \qquad
  0\leq \rho<\infty \,; \qquad U_{2}\sim S_{1}\log\rho + \chi_{2}
  \quad \mbox{as} \quad \rho\to \infty \,; \qquad U_2^{\prime}(0)=0
  \,. \label{nlep:eq3}
\end{equation}

The solution to (\ref{nlep:eq1}) is simply $U_{0}=\chi_{0}$ and
$V_{0}={ w/ \chi_{0}}$, where $w(\rho)>0$ is the unique radially
symmetric solution of $\Delta_{\rho} w-w + w^2=0$ with $w(0)>0$ and
$w\to 0$ as $\rho\to\infty$.  To determine $\chi_0$ in terms of $S_0$
we apply the divergence theorem to the $U_{1}$ equation in
(\ref{nlep:eq2}) to obtain
\begin{equation}
   S_{0}= \int_{0}^{\infty} U_{0} V_{0}^2 \rho \, d\rho =
 \frac{b}{\chi_{0}}  \,, \qquad b\equiv \int_{0}^{\infty} \rho w^2 \, d\rho \,.
  \label{nlep:ssol1}
\end{equation}

It is then convenient to decompose $U_{1}$ and $V_{1}$ in terms of
new variables $U_{1p}$ and $V_{1p}$ by
\begin{equation}
  U_{1}= \chi_1 + \frac{U_{1p}}{\chi_0} \,, \qquad
  V_{1}= -\frac{\chi_1 w}{\chi_0^2}  + \frac{V_{1p}}{\chi_0^3} \,.
 \label{nlep:newvar}
\end{equation}
Upon substituting $U_0=\chi_0$, $V_0={w/\chi_0}$, (\ref{nlep:ssol1}), and
(\ref{nlep:newvar}) into (\ref{nlep:eq2}), and by using
$\Delta_{\rho} w - w + 2 w^2 = w^2$, we readily obtain that $U_{1p}$
and $V_{1p}$ are the unique radially symmetric solutions of
(\ref{score:exp_3}). Finally, we use the divergence theorem on the
$U_{2}$ equation in (\ref{nlep:eq3}) to determine $\chi_1$ in terms of
$S_1$ as
\begin{equation*}
  S_1 = \int_{0}^{\infty} \left( 2 U_0 V_0 V_1 + U_1 V_0^2\right)\rho \, d\rho
  = -\frac{\chi_1}{\chi_0^2} \int_{0}^{\infty} w^2 \rho \, d\rho +
  \frac{1}{\chi_0^3} \int_{0}^{\infty} \left(2w V_{1p} + w^2 U_{1p}\right)\rho
 \, d\rho \,.
\end{equation*}
We then use $\Delta_\rho V_{1p}-V_{1p}=-w^2 U_{1p}-2w V_{1p}$ in the
integral, as obtained from (\ref{score:exp_3}), and we simplify the
resulting expression by using $U_0=\chi_0$ and $V_0={w/\chi_0}$. This
yields $S_{1} = -b^{-1}\chi_1 S_0^2 + b^{-3}S_0^3\int_{0}^{\infty} V_{1p} \rho \,
d\rho$, which gives (\ref{score:exp_4}) for $\chi_1$.
This completes the derivation of Lemma \ref{lemma 3.1}.

To obtain the result in Lemma \ref{lemma 5.1}, we set $S=S_0\nu^{1/2}$
and $S_1=0$ in (\ref{score:exp}) to obtain
\begin{equation}\label{app:sch:short}
  V \sim \frac{S}{S_0} \left( \frac{w}{\chi_0} + \frac{S^2}{S_0^2}
  \left(-\frac{\chi_1 w}{\chi_0^2} + \frac{V_{1p}}{\chi_0^3}\right)\right)\,,
\qquad
  U \sim \frac{S_0}{S} \left( \chi_0 + \frac{S^2}{S_0^2}
  \left(\chi_1 + \frac{U_{1p}}{\chi_0}\right)\right)\,, \qquad
 \chi\sim \frac{S_0\chi_0}{S} + \frac{S}{b^2} \int_0^{\infty} V_{1p}\rho\,
 d\rho \,,
\end{equation}
since $\chi_1=S_0 b^{-2}\int_0^{\infty} V_{1p}\rho \, d\rho$ from
(\ref{score:exp_4}). Finally, since $S_0\chi_0=b$ from
(\ref{score:exp_4}), (\ref{app:sch:short}) reduces to (\ref{simp:exp})
of Lemma \ref{lemma 5.1}.

\appendix
\setcounter{section}{1}
\renewcommand{\theequation}{\Alph{section}.\arabic{equation}}
\newsection{Gierer-Meinhardt Model: Expansion of the Core Problem}\label{app:gm}

We outline the derivation of the results of Lemma \ref{lemma 4.1} and
Lemma \ref{lemma 5.2}.  To motivate the scalings for the solution $U$,
$V$, and $\chi$ to (\ref{geq1:core}) as $S\to 0$, we write $U={\cal U}
S^{p}$, $V={\cal V} S^{p}$, where ${\cal U}$ and ${\cal V}$ are
${\mathcal O}(1)$ as $S\to 0$. We obtain that the $V$-equation in
(\ref{geq1:core}) is unchanged, but that the $U$ equation becomes
\begin{equation*}
  \Delta_\rho {\cal U} = -S^{p} {\cal V}^2 \,; \qquad
  {\cal U} \sim -S^{1-p}\log\rho + S^{-p} \chi \, \quad \mbox{as}\quad
 \rho\to \infty \,.
\end{equation*}
From equating powers of $S$ after applying the divergence theorem it
follows that $p=1-p$, which yields $p={1/2}$.  Then, $\chi={\mathcal
  O}(S^{1/2})$ ensures that ${\cal U}={\mathcal O}(1)$.  This shows
that if $S=S_0\nu^{2}$ where $\nu\ll 1$, the appropriate scalings are
that $V$, $U$, and $\chi$ are all ${\mathcal O}(\nu)$.
To obtain a two-term expansion for the solution to the core problem,
as given in Lemma \ref{lemma 4.1}, we expand $S=S_0\nu^{2} + S_1\nu^3 + \cdots$
and we seek to determine the solution $U$, $V$, and $\chi$ to
(\ref{geq1:core}) in terms of $S_0$ and $S_1$. The appropriate
expansion for these quantities has the form
\begin{equation}
 \left( V \,,  U \,, \chi \right) =
 \nu \left(  V_0 \,, U_{0}\,, \chi_0 \right) +
 \nu^2 \left(  V_1 \,, U_{1}\,, \chi_1 \right) +
 \nu^3 \left(  V_2 \,, U_{2}\,, \chi_2 \right) +\cdots\,. \label{g_nlep:expan}
\end{equation}
Upon substituting (\ref{g_nlep:expan}) into (\ref{geq1:core}), and collecting
powers of $\nu$, we obtain that $U_{0}$ and $V_{0}$ satisfy
\begin{equation}\label{g_nlep:eq1}
\begin{split}
 \Delta_\rho V_0  - V_{0} + {V_{0}^2/U_0}
 &= 0 \,; \qquad  \Delta_\rho U_0 =0 \,,
 \qquad 0\leq \rho<\infty \,,
  \\ V_{0} \to 0 \,, \qquad U_{0} &\to \chi_{0}
 \quad \mbox{as} \quad \rho\to \infty \,; \qquad V_0^{\prime}(0)=
 U_0^{\prime}(0)=0 \,,
\end{split}
\end{equation}
where $\Delta_\rho V_0\equiv V_{0}^{\p\p}+\rho^{-1}V_0^{\p}$.  At next
order, $U_{1}$ and $V_{1}$ satisfy
\begin{equation}
\label{g_nlep:eq2}
\begin{split}
  \Delta_\rho V_1 - V_{1} +
  \frac{2V_{0}}{U_0} V_{1} &= \frac{V_{0}^2}{U_0^2} U_1 \,;
   \qquad \Delta_\rho U_1 = - V_{0}^2 \,,
 \qquad 0\leq \rho<\infty \,,   \\
   V_{1} \to 0 \,, \quad & \quad U_{1}\to -S_{0}\log\rho + \chi_{1}
  \quad \mbox{as} \quad \rho\to \infty \,;
\qquad V_1^{\prime}(0)=U_1^{\prime}(0)=0 \,.
\end{split}
\end{equation}
Then, at one higher order, we get that $U_{2}$ satisfies
\begin{equation}
  \Delta_\rho U_2 = -2V_0 V_1 \,, \qquad
 0\leq \rho<\infty \,; \qquad U_{2}\sim -S_{1}\log\rho + \chi_{2}
  \quad \mbox{as} \quad \rho\to \infty \,; \qquad U_2^{\prime}(0)=0 \,.
\label{g_nlep:eq3}
\end{equation}

The solution to (\ref{g_nlep:eq1}) is simply $U_{0}=\chi_{0}$ and
$V_{0}= \chi_0 w$, where $w(\rho)>0$ is the radially symmetric
ground-state solution of $\Delta_{\rho} w-w + w^2=0$. Next, by
applying the divergence theorem to the $U_{1}$ equation in
(\ref{g_nlep:eq2}) we obtain
\begin{equation}
   S_{0}= \int_{0}^{\infty} \rho V_{0}^2 \, d\rho =
 \chi_0^2 b  \,, \qquad b\equiv \int_{0}^{\infty} \rho w^2 \, d\rho \,.
  \label{g_nlep:ssol1}
\end{equation}

It is then convenient to decompose $U_{1}$ and $V_{1}$ in terms of
new variables $U_{1p}$ and $V_{1p}$ by
\begin{equation}
  U_{1}= \chi_1 + S_0 U_{1p}  \,, \qquad
  V_{1}= \chi_1 w + S_0 V_{1p} \,.
 \label{g_nlep:newvar}
\end{equation}
Upon substituting $U_0=\chi_0$, $V_0=\chi_0 w$, (\ref{g_nlep:ssol1}),
and (\ref{g_nlep:newvar}) into (\ref{g_nlep:eq2}), and by using
$\Delta_{\rho} w - w + 2 w^2 = w^2$, we readily obtain that $U_{1p}$
and $V_{1p}$ are the unique radially symmetric solutions of
(\ref{gcore:exp_3}). Finally, we use the divergence theorem on the
$U_{2}$ equation in (\ref{g_nlep:eq3}) to obtain
$2\chi_0\chi_1 b  +  2\chi_0 S_0 \int_{0}^{\infty} w V_{1p}\rho\,d\rho=S_1$,
which readily yields (\ref{gcore:exp_4}).

To obtain the result in Lemma \ref{lemma 5.2}, we set $S=S_0\nu^{2}$
and $S_1=0$ in (\ref{gcore:exp}), with $\chi_0^2={S_0/b}$ from
(\ref{gcore:exp_4}), to get
\begin{equation}\label{app:gm:short}
  V \sim \sqrt{\frac{S}{S_0}} \chi_0 w + \frac{S}{S_0} \left(\chi_1 w +
  S_0 V_{1p}\right) \,, \qquad
  U \sim \sqrt{\frac{S}{S_0}} \chi_0  + \frac{S}{S_0} \left(\chi_1 +
  S_0 U_{1p}\right) \,, \qquad
  \chi \sim \sqrt{\frac{S}{S_0}} \chi_0  + \frac{S}{S_0} \chi_1 \,,
\end{equation}
where $\chi_1=-S_0 b^{-1}\int_0^{\infty} w V_{1p}\rho \, d\rho$ from
(\ref{gcore:exp_4}). Since $S_0=b\chi_0^2$ from
(\ref{gcore:exp_4}), (\ref{app:gm:short}) reduces to (\ref{gsimp:exp})
of Lemma \ref{lemma 5.2}.

\end{document}